\documentclass[superscriptaddress,
 amsmath,amssymb,
 aps,
 reprint,
pra
]{revtex4-2}

\usepackage{graphicx}
\usepackage{dcolumn}
\usepackage{array}
\usepackage{physics}
\usepackage{amsfonts, amsmath}
\usepackage{dsfont}
\usepackage{amsthm}
\usepackage{bm}
\usepackage{hyperref}
\usepackage{bbold}
\usepackage{xcolor}
\usepackage{hyperref}
\usepackage[normalem]{ulem}
\usepackage{tabularx}
\usepackage{tikz}
\usetikzlibrary{arrows.meta,positioning,calc,fit}

\newcommand{\commentout}[1]{}

\newcommand{\per}{\mathrm{Per}}
\newcommand{\poly}{\mathrm{poly}}

\newtheorem{theorem}{Theorem}
\newtheorem{corollary}{Corollary}
\newtheorem{proposition}{Proposition}

\newtheorem{problem}{Problem}

\newtheorem{lemma}{Lemma}

\usepackage{mathtools}
\providecommand{\pf}{\operatorname{pf}}
\providecommand{\C}{\mathbb C}
\providecommand{\qv}{\boldsymbol q}

\newcommand{\AppendixTOCDiscard}[4]{}
\newcommand{\AppendixTOCHide}{%
    \let\contentsline\AppendixTOCDiscard
}
\newcommand{\AppendixTOCShow}{%
    \let\contentsline\AppendixTOCSavedContentsline
}

\begin{document}
\addtocontents{toc}{\protect\AppendixTOCHide}
\definecolor{navy}{RGB}{46,72,102}
\definecolor{pink}{RGB}{219,48,122}
\definecolor{grey}{RGB}{184,184,184}
\definecolor{yellow}{RGB}{255,192,0}
\definecolor{grey1}{RGB}{217,217,217}
\definecolor{grey2}{RGB}{166,166,166}
\definecolor{grey3}{RGB}{89,89,89}
\definecolor{red}{RGB}{255,0,0}

\preprint{APS/123-QED}

\title{Fast classical simulation algorithms for free-fermion dynamics with magic input states}

\author{Jiwon Heo}

\affiliation{
 Graduate School of Quantum Science and Technology, Korea Advanced Institute of Science and Technology~(KAIST)
}

\author{Oliver Reardon-Smith}
\affiliation{Center for Quantum Enabled-Computing, Center for Theoretical Physics of the Polish Academy of Sciences, Al. Lotników 32/46, 02-668 Warsaw, Poland}

\author{Michał Oszmaniec}
\affiliation{Center for Quantum Enabled-Computing, Center for Theoretical Physics of the Polish Academy of Sciences, Al. Lotników 32/46, 02-668 Warsaw, Poland}

\author{Zolt\'an Zimbor\'as}
\affiliation{University of Helsinki, Yliopistonkatu 4 00100 Helsinki, Finland}
\affiliation{HUN-REN Wigner Research Centre for Physics, Budapest, Hungary}
\affiliation{Algorithmiq Ltd, Kanavakatu 3C 00160 Helsinki, Finland}

\author{Changhun Oh}
\email{changhun0218@gmail.com}
\affiliation{Department of Physics, Korea Advanced Institute of Science and Technology, Daejeon 34141, Korea}
\affiliation{
 Graduate School of Quantum Science and Technology, Korea Advanced Institute of Science and Technology~(KAIST)
}

\begin{abstract}
The classical cost of simulating free-fermion (matchgate) circuits depends on both the input state and the computational task. We develop classical algorithms for circuits with non-Gaussian inputs, addressing exact sampling and both additive-error and exact estimation of expectation values. For inputs consisting of $n$ copies of the four-mode magic state, we give exact sampling algorithms, generalizing the Clifford and Clifford algorithm for Boson sampling, for passive and active free-fermion dynamics with worst-case arithmetic cost $O(2^n)$ per sample, without any multiplicative polynomial prefactor. The passive algorithm enables an exact simulation of the non-interacting regime of a recent trapped-ion experiment. We also extend recently developed additive-error estimation of number correlators and related observables from passive to active dynamics
with runtime polynomial in $n$ and the inverse additive-error, for input states formed by products of $n$ four-mode even parity states. Our estimator uses a Gaussian-state decomposition adapted to each sampled state to control its second moment. This includes estimation of individual output probabilities. Finally, we show that for even-parity product inputs composed of constant-size blocks, expectation values of Majorana monomials of logarithmic weight can be computed exactly in polynomial time, improving on prior Majorana propagation-style algorithms with quasi-polynomial runtime. The exact algorithm processes the input blocks by dynamic programming, sharing calculations across subsets of Majorana factors. Our suite of algorithms provide classical benchmarks for fermionic quantum simulations across a broad spectrum of computational tasks.
\end{abstract}

\maketitle

\section{Introduction}

Quantum simulation of fermionic systems is a key application of quantum computers, motivated by a variety of problems including in quantum chemistry and many-body physics~\cite{mcardle2020quantum}, as well as certifiable quantum advantage~\cite{oszmaniec2022fermion}. Free-fermion dynamics, generated by quadratic Hamiltonians, are efficiently classically simulable for Gaussian inputs~\cite{valiant2002quantum,Bravyi2005}. The input state changes this picture: suitable non-Gaussian states can make fermionic linear-optical circuits computationally powerful~\cite{hebenstreit2020computational,oszmaniec2022fermion}. 
In fact, recent trapped-ion experiments have implemented fermionic Gaussian operations on non-Gaussian magic inputs~\cite{Phasecraft}. These developments motivate a systematic analysis of how classical simulation costs depend on the structure and number of non-Gaussian input blocks.

Two tasks are particularly important for assessing the classical complexity of quantum circuits. First, sampling asks for outcomes drawn from the full output distribution. Sampling problems, including random circuit sampling~\cite{bouland2019complexity}, boson sampling~\cite{aaronson2013computational}, and fermion sampling~\cite{oszmaniec2022fermion}, play a central role in complexity-theoretic proposals for quantum computational advantage. Second, observable estimation asks for expectation values of specified physical quantities, including probabilities, correlations, and energies. It is closely tied to practical quantum simulation, where the goal is to predict or measure particular properties of a system. The two tasks can have different classical costs, so understanding both is necessary to assess the computational landscape formed by these circuits.

Previous work has approached these tasks from a variety of different directions. Methods based on Gaussian decompositions give classical simulation algorithms for non-Gaussian fermionic circuits, with costs governed by resource-theoretic monotones quantifying non-Gaussianity in the circuit~\cite{DiasKonig2024,reardon2024improved,cudby2023gaussian,dias2026optimal}. For exact sampling, combining the gate-by-gate method of Ref.~\cite{bravyi2022simulate} with Gaussian decompositions of the input states~\cite{DiasKonig2024} implies a runtime of $O(n^5 2^n)$ for $n$ four-mode magic input blocks. For exact expectation values, direct Majorana propagation expands a weight-$k$ observable into up to $(2M)^k$ terms on $M$ modes~\cite{jozsa2008matchgates}. Recent Gaussian-operator decompositions improve this scaling for four-mode magic inputs, with a runtime bound that still scales as $n^{O(k)}$~\cite{bako2025fermionic}. These bounds yield polynomial runtime for a constant Majorana weight. For additive-error estimation, output probabilities and observables can be estimated efficiently for structured magic inputs under passive free-fermion dynamics~\cite{oh2026classical}.

Here, we address both tasks for free-fermion circuits with non-Gaussian inputs. For sampling, we take $n$ copies of the four-mode magic state $\ket{\phi_4}=(\ket{1100}+\ket{0011})/\sqrt{2}$. The passive and active samplers have worst-case arithmetic costs $O(n^4+2^n)$ and $O(n^5+2^n)$ per sample, respectively. In both cases, the exponential contribution is $O(2^n)$, with no multiplicative polynomial factor. This scaling follows from sharing probability calculations within each recovery step and summing their costs over the recursive sampling path. The samples follow the full occupation-number distribution and can be used to evaluate many observables from the same set of outcomes. We use the passive algorithm to simulate the noninteracting regime of the trapped-ion experiment in Ref.~\cite{Phasecraft}. The resulting samples reproduce the free-fermion benchmarks for the doublon number and triplet density and allow us to evaluate the high-weight Wilson-loop observables directly.

For expectation values, we first develop an exact algorithm for products of arbitrary pure even-parity states on constant-size input blocks. After active free-fermion evolution, the expectation value of a Majorana string of weight $k$ on $M$ modes can be computed in $O(n^2+nK2^K)$ time, where $n$ is the number of input blocks and $K=\min\{k,2M-k\}$. In particular, the runtime is polynomial when $K=O(\log n)$. This extends the polynomial-time regime from constant to logarithmic Majorana weight for arbitrary constant-size even-parity input blocks, without requiring the special form of $|\phi_4\rangle$.
The algorithm combines blockwise factorization with dynamic programming to sum the contributions from different assignments of Majorana factors to input blocks.

Beyond exact computation, we also consider additive-error estimation of number correlators of arbitrary weight. For products of arbitrary pure even-parity four-mode states, we extend the corresponding passive-evolution results of Ref.~\cite{oh2026classical} to active free-fermion dynamics. After a one-time preprocessing of the Gaussian circuit, our algorithm uses $O(n^3\epsilon^{-2}\log(2/\delta))$ arithmetic operations for additive error $\epsilon$ and failure probability at most $\delta$. A direct extension of the passive estimator can have exponentially growing variance. The main ingredient is an unbiased estimator of transition amplitudes between block-product states, based on weighted Gaussian overlaps. We decompose one of the block-product states into Gaussian components and choose the phases in this decomposition using the covariance matrix of the Gaussian vector in each sample. Together with a Gaussian compression bound, this yields an estimator whose second moment is at most one. Expressing number correlators as signed averages of these amplitudes then yields the stated estimation guarantee. The same framework also estimates individual output probabilities and occupation marginals.

The remainder of this paper is organized as follows. Sections~\ref{Sec:Preliminaries} and~\ref{Sec:Problem setup} introduce the notation and define the simulation tasks. Section~\ref{sec:instrument-sampling} presents a common recursive framework for exact sampling, which we apply to passive and active Gaussian evolution in Secs.~\ref{Sec:Passive sampling} and~\ref{Sec:Active sampling}, respectively. Sections~\ref{Sec:Mean value} and~\ref{Sec:Approximate mean value} address exact computation of Majorana expectation values and additive-error estimation of number correlators, respectively. We discuss the implications and open questions in Sec.~\ref{Sec:Discussion}. We summarize the key findings of our work in Table~\ref{tab:informal-summary}.

\begin{table*}[t]
    \centering
    \small
    \renewcommand{\arraystretch}{1.2}
    \setlength{\tabcolsep}{4pt}
    \begin{tabular*}{\textwidth}{@{\extracolsep{\fill}}ccccccc@{}}
        \hline\hline
        Task & Error & Input state & Operation & Output & Result & Runtime \\
        \hline
        \noalign{\vskip4pt}
        Sampling
        & Exact
        & $|\phi_4\rangle^{\otimes n}$
        & Passive
        & Fock basis sample
        & Thm.~\ref{thm:sampling-passive}
        & $O(n^4+2^n)$
        \\\hline
        \noalign{\vskip5pt}
        Sampling
        & Exact
        & $|\phi_4\rangle^{\otimes n}$ 
        & Active
        & Fock basis sample
        & Thm.~\ref{thm:sampling-active-worst}
        & $O(n^5+2^n)$
        \\
        \noalign{\vskip5pt}
        \hline
        \noalign{\vskip4pt}
        \parbox[c]{1.7cm}{\centering Expectation\\value}
        & Exact
        & \parbox[c]{3.4cm}{\centering Arbitrary constant-size\\block-product state}
        & Active
        & \begin{tabular}[c]{@{}c@{}}
    Expectation value of product \\
   of $k$ Majoranas
  \end{tabular}
        & Thm.~\ref{thm:exact-expectation-value}
        & $O(n^2+nk2^k)$
        \\
        \noalign{\vskip5pt}
        \hline
        \noalign{\vskip4pt}
        \parbox[c]{1.7cm}{\centering Overlap}
        & \parbox[c]{1.3cm}{\centering Additive\\error}
        & \parbox[c]{3.6cm}{\centering Four-mode block-product\\ even states}
        & Active
        & \parbox[c]{1.3cm}{Overlap\\$\bra{\Phi}\hat{G}\ket{\Psi}$}
        & Thm.~\ref{thm:active-flo-expectation-value-estimation}
        & $O(n^3\epsilon^{-2}\log(2/\delta))$
        \\
        \noalign{\vskip5pt}
        \hline
        \noalign{\vskip4pt}
        \parbox[c]{1.7cm}{\centering Expectation\\value}
        & \parbox[c]{1.3cm}{\centering Additive\\error}
        & \parbox[c]{3.6cm}{\centering Four-mode block-product\\ even states}
        & Active
        & \parbox[c]{3cm}{Expectation value of\\number correlator}
        & Thm.~\ref{thm:active-number-correlators}
        & $O(n^3\epsilon^{-2}\log(2/\delta))$
        \\
        \noalign{\vskip5pt}
        \hline\hline
    \end{tabular*}
    \caption{Summary of our main results. All input blocks are pure and have even parity. In each case $n$ is the number of tensor-product blocks in the input state. The four-mode magic state $|\phi_4\rangle$ is defined in Eq.~\eqref{eq:magic-block}. For Majorana strings on $M$ modes, we restrict to $k\leq M$; longer strings are reduced to their complements using parity symmetry. The additive-error results return an estimate with error $\epsilon$ and failure probability at most $\delta$ and the number correlator algorithm applies to correlators of arbitrary weight. The additive-error runtimes exclude the one-time polynomial-time preprocessing of a circuit description.}
    \label{tab:informal-summary}
\end{table*}

\section{Preliminaries}\label{Sec:Preliminaries}
We use hats for operators on Fock space and leave their single-particle and Majorana transformation matrices unhatted. We consider a system of $M$ fermionic modes described by annihilation and creation operators $\hat c_j$ and $\hat c_j^\dagger$, which satisfy
\begin{align}
    \{\hat{c}_i,\hat{c}_j\} &= 0, & \{\hat{c}_i,\hat{c}_j^\dagger\} &= \delta_{ij}\hat I,
\end{align}
where $\{\hat A,\hat B\}:=\hat A\hat B+\hat B\hat A$. The vacuum $|0\rangle$ is annihilated by every $\hat c_i$, and an ordered product of creation operators defines the occupation basis
\begin{align}
    |x\rangle &= (\hat c_1^\dagger)^{x_1}\cdots(\hat c_M^\dagger)^{x_M}|0\rangle,
\end{align}
where $x\in\{0,1\}^M$ and $x_j$ is the eigenvalue of the number operator $\hat n_j:=\hat c_j^\dagger\hat c_j$.

To describe the Gaussian dynamics and the observables considered below, we use the Majorana operators
\begin{align}
    \hat\gamma_{2j-1} &= \hat c_j+\hat c_j^\dagger, &
    \hat\gamma_{2j} &= -i(\hat c_j-\hat c_j^\dagger),
\end{align}
which obey $\{\hat\gamma_a,\hat\gamma_b\}=2\delta_{ab}\hat I$. Products of these operators are the observables addressed by our exact expectation-value algorithm. We write a Hermitian Majorana monomial of weight $k$ as
\begin{align}
    \hat\Gamma &= i^{k(k-1)/2}\hat\gamma_{a_1}\cdots\hat\gamma_{a_k},
    \label{eq:majorana-string}
\end{align}
where $1\leq a_1<\cdots<a_k\leq2M$ and the weight counts the number of Majorana factors. We also call these operators Majorana strings, with the convention that the weight-zero string is the identity.

The evolution is described by fermionic Gaussian unitaries generated by quadratic Hamiltonians. When the Hamiltonian preserves particle number, the resulting unitary is \textit{passive} and acts linearly on the annihilation operators through a matrix $U\in\operatorname{U}(M)$:
\begin{align}
    \hat U^\dagger\hat c_i\hat U &= \sum_{j=1}^M U_{ij}\hat c_j.
\end{align}
Allowing pairing terms gives the general class of Gaussian unitaries, which we call \textit{active}; this class includes passive unitaries. Their action is linear in the Majorana representation, where a matrix $O\in\operatorname{SO}(2M)$ specifies
\begin{align}
    \hat O^\dagger\hat\gamma_a\hat O &= \sum_{b=1}^{2M}O_{ab}\hat\gamma_b.
\end{align}
Gaussian unitaries preserve fermionic parity, but generally do not conserve the number of particles. We will occasionally have reason to employ non-unitary Gaussian operations, as defined in Ref.~\cite{knill2001fermioniclinearopticsmatchgates,Bravyi2005}, which fall into two classes, invertible Gaussian operations are products of exponentials of even complex quadratics in the Majorana fermion operators and not necessarily invertible operators which may be written as limits of the first class. This latter class includes, for example, single mode number projectors. Classical simulation algorithms of purely Gaussian circuits remains efficient when the definition of Gaussian operations is generalized to include these operators. 

\section{Problem setup}\label{Sec:Problem setup}

We study sampling and expectation-value computation after Gaussian evolution of a block-product input, as illustrated in Fig.~\ref{fig:Structure}. The $M$ modes are partitioned into $n$ blocks of sizes $m_1,\ldots,m_n$ and ordered block by block, so that the input takes the form
\begin{align}
    |\Psi_{\mathrm{in}}\rangle &= \bigotimes_{j=1}^n|\psi_j\rangle,
    & M &= \sum_{j=1}^n m_j,
\end{align}
where each $|\psi_j\rangle$ is normalized. We use $\hat G$ for either a passive or an active Gaussian unitary. Below, we define the computational tasks in this common setting and specify the input and observable classes covered by our algorithms.

\begin{figure}[t]
    \centering
    \includegraphics[width=0.9\linewidth]{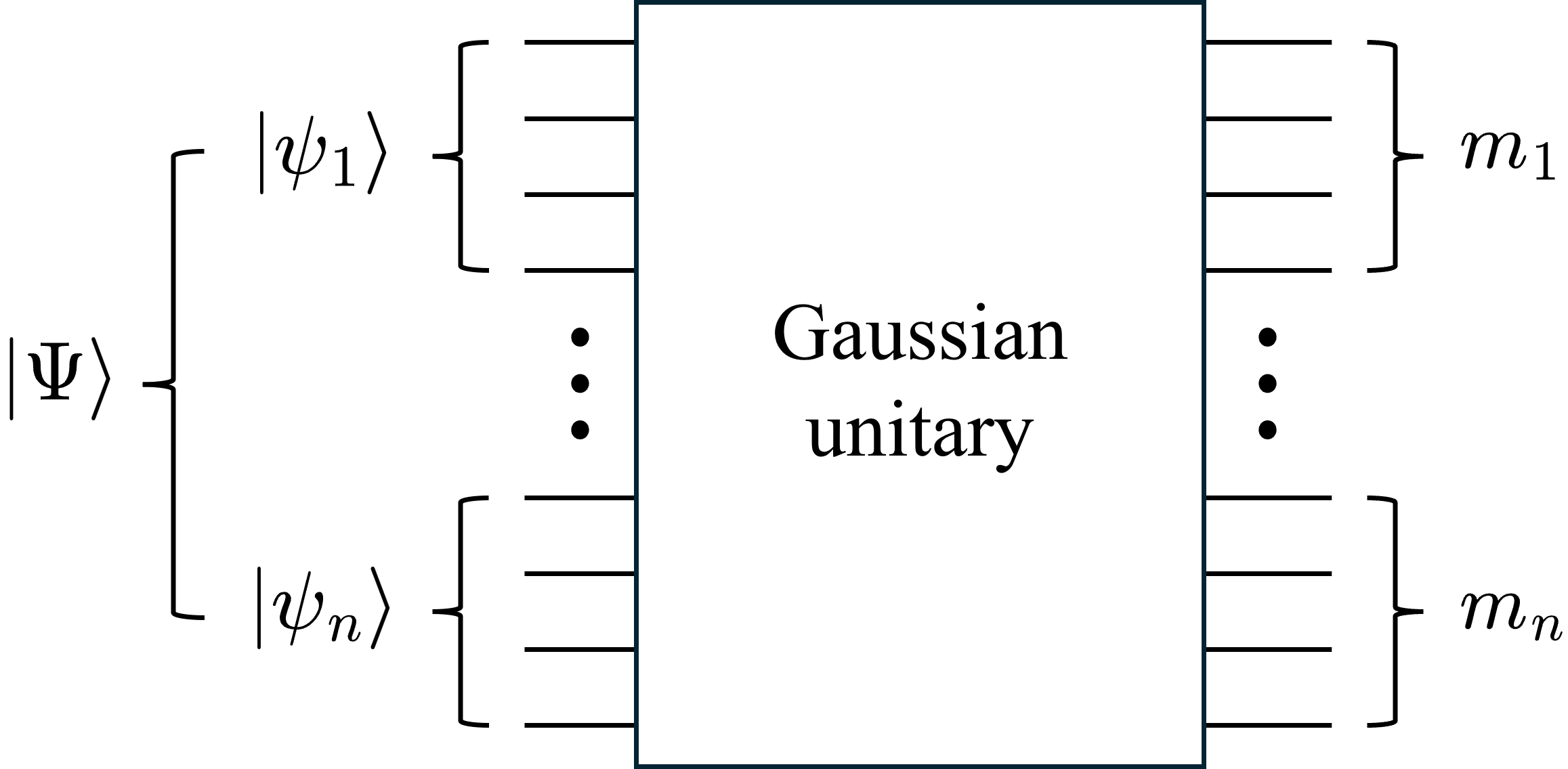}
    \caption{Gaussian evolution of a block-product input. The $M$ modes are partitioned into $n$ blocks, each prepared in a state $|\psi_j\rangle$. We study occupation-number sampling and expectation values after the Gaussian evolution.}
    \label{fig:Structure}
\end{figure}

For sampling, the Gaussian evolution is followed by an occupation-number measurement. Each outcome $x\in\{0,1\}^M$ records the occupations of all $M$ modes, with at most one fermion per mode by the Pauli exclusion principle.

\begin{problem}[Exact sampling]
Given $|\Psi_{\mathrm{in}}\rangle$ and $\hat G$, generate $x\in\{0,1\}^M$ with probability
\begin{align}
    p(x) &= \left|\langle x|\hat G|\Psi_{\mathrm{in}}\rangle\right|^2.
\end{align}
\end{problem}

Our exact sampling algorithms take $M=4n$ and $|\Psi_{\mathrm{in}}\rangle=|\Psi_n\rangle:=|\phi_4\rangle^{\otimes n}$, where
\begin{align}
    |\phi_4\rangle &= \frac{|1100\rangle+|0011\rangle}{\sqrt2}.
    \label{eq:magic-block}
\end{align}
The algorithms combine recursive sampling with shared calculations of the probabilities required to recover a parent sample. They generate samples from the full occupation-number distribution in $O(2^n)$ time under both passive and active Gaussian evolution, as shown in Secs.~\ref{Sec:Passive sampling} and~\ref{Sec:Active sampling}, respectively.

For expectation values, we consider a Hermitian observable $\hat A$. Its expectation value yields the mean outcome of measuring $\hat A$ on the output state.

\begin{problem}[Exact expectation-value computation]\label{Def:Mean value}
Given $|\Psi_{\mathrm{in}}\rangle$, $\hat G$, and an operator $\hat A$, exactly compute
\begin{align}
    \mu &= \langle\Psi_{\mathrm{in}}|\hat G^\dagger\hat A\hat G|\Psi_{\mathrm{in}}\rangle.
\end{align}
\end{problem}

Our exact expectation-value algorithm allows arbitrary block states of definite fermionic parity, with all block sizes bounded by a constant independent of $n$. Each block may have either even or odd parity. For simplicity, we present the algorithm for even-parity blocks; the same runtime bound holds for any choice of block parities. We take $\hat A$ to be a Hermitian Majorana monomial as defined in Eq.~\eqref{eq:majorana-string}. The algorithm uses dynamic programming to exploit the block-product structure, with a computational cost that depends on the Majorana weight, as shown in Sec.~\ref{Sec:Mean value}.

We also consider estimation with an additive error tolerance $\epsilon$ and failure probability $\delta$.

\begin{problem}[Additive-error expectation-value estimation]
Given the setting of Problem~\ref{Def:Mean value}, an error tolerance $\epsilon>0$, and a failure probability $0<\delta<1$, output an estimate $\widetilde\mu$ satisfying
\begin{align}
    \Pr\!\left[|\widetilde\mu-\mu|\leq\epsilon\right] &\geq 1-\delta.
\end{align}
\end{problem}

For additive-error estimation, we allow arbitrary pure even-parity states on four-mode blocks. Local Gaussian unitaries, which can be absorbed into $\hat G$, bring each block to the form~\cite{bako2025fermionic}
\begin{align}
    |\psi_j\rangle &= \alpha_j|1100\rangle+\beta_j|0011\rangle,
    \label{eq:paired-block}
\end{align}
where $\alpha_j,\beta_j\in\mathbb C$ and $|\alpha_j|^2+|\beta_j|^2=1$. Our algorithm uses an unbiased estimator based on weighted Gaussian overlaps. The observables include multipoint number correlators of arbitrary weight and occupation-number projectors, whose expectation values give individual output probabilities and marginal probabilities for specified occupations on a subset of modes. We also consider binned probabilities for integer-weighted occupation sums with polynomially bounded range. Section~\ref{Sec:Approximate mean value} extends the corresponding passive-FLO reductions of Ref.~\cite{oh2026classical} to active Gaussian evolution.

\section{Instrument-based recursive sampling}
\label{sec:instrument-sampling}

We present a general physics-inspired principle for sampling from highly symmetric quantum circuits. Our approach uses a quantum instrument whose average channel (i) is covariant with respect to the quantum evolution and (ii) acts on the measured probabilities as a classical stochastic map.
The general principle, illustrated in Fig.~\ref{fig:instrument-overview}, is to first apply a suitable quantum instrument, sample the output of the resulting conditional process, and recover a parent output using a small set of parent Born probabilities. This construction provides the recursive principle underlying the exact samplers for fermionic circuits. When applied to bosonic linear optics it recovers the original boson sampling algorithm of Clifford and Clifford~\cite{clifford2018classical}, as shown in App.~\ref{app:bosonic-cc}.

\subsection{General sampling principle}

Let $\mathcal H_{\rm p}$ and $\mathcal H_{\rm c}$ denote Hilbert spaces we call parent and child, with circuit actions $U_{\rm p}$ and $U_{\rm c}$ and finite POVMs $\{P_x\}$ and $\{Q_y\}$, respectively. Consider a completely positive trace-preserving map
\begin{equation}
 \mathcal C:\mathcal B(\mathcal H_{\rm p})\longrightarrow
 \mathcal B(\mathcal H_{\rm c}).
\end{equation}
An instrument unraveling of $\mathcal C$ is a family of completely positive maps $\{\mathcal I_\alpha\}$ satisfying $\sum_\alpha\mathcal I_\alpha=\mathcal C$~\cite{DaviesLewis1970,Watrous2018}. Applying the instrument to an input state $\rho$ produces a classical label $\alpha$ with probability $\pi_\alpha=\operatorname{Tr}\mathcal I_\alpha(\rho)$, together with the \emph{post-measurement} quantum state $\rho_\alpha=\frac{\mathcal I_\alpha(\rho)}{\pi_\alpha}$, whenever $\pi_\alpha>0$. The channel $\mathcal{C}$ determines the average state when the label $\alpha$ is disregarded.

We require both covariance with the circuit and compatibility with the final measurement:
\begin{align}
 \mathcal C\circ\operatorname{Ad}_{U_{\rm p}}
 &=\operatorname{Ad}_{U_{\rm c}}\circ\,\mathcal{C},
 \label{eq:instrument-covariance}\\
 \mathcal C^\dagger(Q_y)&=\sum_x K(y\mid x)P_x.
 \label{eq:instrument-measurement}
\end{align}
Here $\operatorname{Ad}_U(\rho)=U\rho U^\dagger$, $\mathcal C^\dagger$ is the trace adjoint, and $K$ is a stochastic matrix: $K(y\mid x)\ge0$ and $\sum_yK(y\mid x)=1$. Equation~\eqref{eq:instrument-measurement} says that the channel acts on the measured probabilities as classical noise (note that this is an extra condition and is not a consequence of covariance). The following theorem shows how an instrument satisfying these conditions can be used to generate samples from the parent quantum circuit $U_{\rm p}$ on a general input state; see Fig.~\ref{fig:instrument-overview}~(a).

\begin{theorem}[Covariant-instrument sampling and Bayesian recovery]
\label{thm:instrument-recovery}
Assume Eqs.~\eqref{eq:instrument-covariance} and~\eqref{eq:instrument-measurement}. Define
\begin{align}
 p(x)&=\operatorname{Tr}(P_xU_{\rm p}\rho U_{\rm p}^\dagger), & q_\alpha(y)&=\operatorname{Tr}(Q_yU_{\rm c}\rho_\alpha U_{\rm c}^\dagger).
\end{align}
Sample $\alpha$ with probability $\pi_\alpha$, then sample $y$ from $q_\alpha$. The marginal law of $y$ is
\begin{equation}
 q(y)=\sum_\alpha\pi_\alpha q_\alpha(y)
     =\sum_xK(y\mid x)p(x).
 \label{eq:instrument-mixture}
\end{equation}
For a sampled $y$, draw $x$ according to
\begin{equation}
 R_p(x\mid y)=\frac{K(y\mid x)p(x)}{q(y)}.
 \label{eq:instrument-recovery}
\end{equation}
The returned outcome has distribution $p$ exactly.
\end{theorem}
\begin{proof}
Averaging the instrument outputs and using the two compatibility identities gives
\begin{align}
 \sum_\alpha\pi_\alpha q_\alpha(y)
 &=\operatorname{Tr}\!\left[Q_yU_{\rm c}\mathcal C(\rho)U_{\rm c}^\dagger\right]\nonumber\\
 &= \operatorname{Tr}\big[Q_y\,\mathcal{C}(U_p\rho U_p^\dagger)\big] \nonumber \\
&=\operatorname{Tr}\!\left[\mathcal C^\dagger(Q_y)U_{\rm p}\rho U_{\rm p}^\dagger\right]\nonumber\\
 &=\sum_xK(y\mid x)p(x).
\end{align}

Thus, Eq.~\eqref{eq:instrument-recovery} is normalized on the positive support of $q$. Its output marginal is
\begin{equation}
 \sum_{y:q(y)>0}q(y)R_p(x\mid y)=p(x).
\end{equation}
Indeed, if $p(x)>0$, any omitted $y$ has $K(y\mid x)=0$; the assertion is immediate when $p(x)=0$.
\end{proof}

\begin{figure*}[t]
\centering
\begingroup
\definecolor{instNavy}{HTML}{203B56}
\definecolor{instTeal}{HTML}{247E83}
\definecolor{instOchre}{HTML}{AC7135}
\begin{tikzpicture}[
  x=1cm,y=1cm,
  every node/.style={font=\small,align=center},
  theory/.style={draw=instNavy!55,rounded corners=2pt,fill=instNavy!3,
    text width=2.85cm,minimum height=0.93cm,inner sep=5pt},
  stage/.style={draw=instTeal!75,rounded corners=2pt,fill=instTeal!4,
    text width=6.1cm,minimum height=0.89cm,inner sep=5pt},
  cert/.style={-{Stealth[length=2mm]},draw=black,densely dashed},
  run/.style={-{Stealth[length=2mm]},draw=black,semithick},
  lab/.style={font=\footnotesize,fill=white,inner sep=2pt}
]
\node[anchor=west,font=\small\bfseries,text=black] at (0,0.5)
  {(a) Covariance and the averaged law};
\node[anchor=west,font=\small\bfseries,text=black] at (8.35,0.5)
  {(b) One recursive recovery step};

\node[theory] (rho) at (1.55,-0.55) {$\rho$\\parent input};
\node[theory] (child) at (5.55,-0.55)
  {$\mathcal C(\rho)=\sum_\alpha\pi_\alpha\rho_\alpha$\\average over children};
\node[theory] (out) at (1.55,-2.65)
  {$U_{\rm p}\rho U_{\rm p}^\dagger$\\parent output state};
\node[theory] (cout) at (5.55,-2.65)
  {$U_{\rm c}\mathcal C(\rho)U_{\rm c}^\dagger$\\$=\mathcal C(U_{\rm p}\rho U_{\rm p}^\dagger)$};
\node[theory,minimum height=0.75cm] (p) at (1.55,-4.9)
  {$p(x)$\\target distribution};
\node[theory,minimum height=0.75cm] (q) at (5.55,-4.9)
  {$q(y)=\sum_\alpha\pi_\alpha q_\alpha(y)$};
\draw[cert] (rho) -- node[lab,above] {$\mathcal C$} (child);
\draw[cert] (rho) -- node[lab,left] {$U_{\rm p}$} (out);
\draw[cert] (child) -- node[lab,right] {$U_{\rm c}$} (cout);
\draw[cert] (out) -- node[lab,above] {$\mathcal C$} (cout);
\draw[cert] (out) -- node[lab,left] {measure} (p);
\draw[cert] (cout) -- node[lab,right] {measure} (q);
\draw[cert] (p) -- node[lab,above] {$K$} (q);

\node[stage] (choose) at (11.65,-0.55)
  {Draw $\alpha\sim\pi$; form the child $\rho_\alpha$\\
   {\footnotesize Retain $\rho$ for the recovery step.}};
\node[stage] (sample) at (11.65,-2.05)
  {Recursively draw $y\sim q_\alpha$\\
   {\footnotesize Averaging over $\alpha$ gives the marginal $q(y)$.}};
\node[stage,draw=instOchre!85,fill=instOchre!5] (weights) at (11.65,-3.55)
  {Evaluate $w_x=K(y\mid x)p(x)$\\
   for $x\in\mathcal N(y)=\{x:K(y\mid x)>0\}$.};
\node[stage] (return) at (11.65,-5.05)
  {Draw $x\sim R_p(\cdot\mid y)$,\quad
   $R_p(x\mid y)=\dfrac{w_x}{\sum_z w_z}$\\
   \footnotesize Return $x$; its distribution is exactly $p$.};
\draw[run] (choose) -- (sample);
\draw[run] (sample) -- (weights);
\draw[run] (weights) -- (return);
\end{tikzpicture}
\caption{Instrument-based sampling protocol for symmetric quantum circuits. A covariant channel $\mathcal C$ admits an input-side
unravelling into child states $\rho_\alpha$ with probabilities
$\pi_\alpha$. Covariance and compatibility with the output measurement
give the averaged identity $\sum_\alpha\pi_\alpha q_\alpha=Kp$.
Theorem~\ref{thm:instrument-recovery} then reconstructs a parent sample
from a recursively generated child sample and parent probabilities on
the predecessor set $\mathcal N(y)$. In particular, one does not require
$q_\alpha=Kp$ for a fixed branch $\alpha$. For a pure parent input and a rank-one output measurement,
evaluating these probabilities amounts to computing the corresponding
parent amplitudes. The parent input is retained throughout; the
instrument label $\alpha$ and the recovered outcome $x$ play distinct
roles.}
\label{fig:instrument-overview}
\endgroup
\end{figure*}
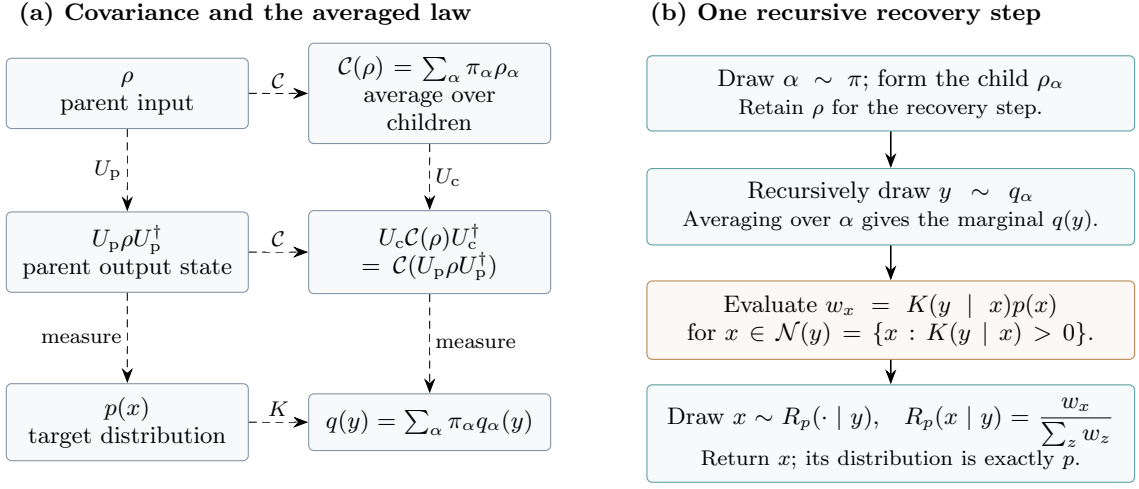

Theorem~\ref{thm:instrument-recovery} is the basis of the recursion. To sample a potentially hard parent law $p$ corresponding to an input $\rho$ and circuit $U_{\rm p}$, first draw $\alpha$ with probability $\pi_\alpha$ and form $\rho_\alpha$. Then sample from the conditional output law $q_\alpha(y)=\operatorname{Tr}(Q_yU_{\rm c}\rho_\alpha U_{\rm c}^\dagger)$. To recover the parent sample, the probabilities $p(x)=\operatorname{Tr}(P_xU_{\rm p}\rho U_{\rm p}^\dagger)$ need to be evaluated only on
\begin{equation}
 \mathcal N(y)=\{x:K(y\mid x)>0\}.
 \label{eq:instrument-neighborhood}
\end{equation}
One need not compute $q(y)$ separately: by Eq.~\eqref{eq:instrument-mixture}, summing $K(y\mid x)p(x)$ over $\mathcal N(y)$ gives exactly $q(y)$. For a pure input and a rank-one measurement, evaluating these probabilities amounts to computing the corresponding amplitudes. These amplitudes may themselves be costly, but in the cases we have studied in sections~\ref{Sec:Passive sampling} and~\ref{Sec:Active sampling} they share an algebraic structure and can be evaluated together. The number of relevant amplitudes can be small even when the full output distribution has exponentially many outcomes.

If the same covariant instrument structure is available for the conditional states $\rho_\alpha$, this procedure can be applied recursively; see Fig.~\ref{fig:instrument-overview}(b). The recursive procedure saves the parent input, draws an instrument branch, calls the child sampler, and applies Eq.~\eqref{eq:instrument-recovery} on return. The base case is an exactly sampleable ``terminal family''. For efficiency, the algorithm requires an accessible instrument distribution $\pi_\alpha$, tractable descriptions of the conditional states $\rho_\alpha$, manageable predecessor sets $\mathcal{N}(y)$, and a termination argument. These requirements do not follow from Thm.~\ref{thm:instrument-recovery} alone. The key intuition is that successive instrument applications can make the child states easier to sample from. For the FLO circuits considered here, this happens because the instrument can reduce the fermionic magic in the states.

In Secs.~\ref{sec:instrument-passive} and~\ref{sec:instrument-active}, we apply this general approach to passive and active fermionic linear optics, respectively, and explain how these requirements are met in each case; see Fig.~\ref{fig:fermionic-recursions}.

\subsection{Passive fermionic linear optics}
\label{sec:instrument-passive}
For $V=\mathbb C^M$ and $k\ge1$, we define the parent and child spaces by
\begin{equation}
 \mathcal H_{\rm p}=\bigwedge^kV,\qquad
 \mathcal H_{\rm c}=\bigwedge^{k-1}V.
 \label{eq:passive-parent-child}
\end{equation}
The circuit actions are $U_{\rm p}=\bigwedge^kU$ and $U_{\rm c}=\bigwedge^{k-1}U$, namely the restrictions of $\hat U$ to the corresponding particle sectors.

For sampling from passive FLO circuits on $k$ fermions, we choose the instrument $\mathcal I^{(k)}_a(\rho)=\frac1k\hat c_a\rho\hat c_a^\dagger$, which describes particle loss in orthogonal modes. It averages to the normalized one-particle loss channel
\begin{equation}
 \mathcal D_k(\rho)=\frac1k\sum_{a=1}^M
       \hat c_a\rho\hat c_a^\dagger.
 \label{eq:passive-loss-channel}
\end{equation}
On $\bigwedge^kV$, $\sum_a\hat c_a^\dagger\hat c_a=kI$, so this map is trace preserving. The transformation law for the annihilation operators and unitarity of $U$ give Eq.~\eqref{eq:instrument-covariance}. 

For a pure parent state $\ket{\psi}$, the branch probabilities and normalized child states are
\begin{equation}
 \pi_a=\frac{\langle\psi|\hat n_a|\psi\rangle}{k},\qquad
 |\psi_a\rangle=\frac{\hat c_a|\psi\rangle}
 {\|\hat c_a|\psi\rangle\|}.
 \label{eq:passive-branches}
\end{equation}
Only branches with $\pi_a>0$ are sampled.

Write an occupation outcome as a set of occupied modes $S\subset [M]$. For $|S|=k-1$, the adjoint action is
\begin{equation}
 \mathcal D_k^\dagger(|S\rangle\langle S|)
  =\frac1k\sum_{j\notin S}|S\cup\{j\}\rangle\langle S\cup\{j\}|.
\end{equation}
Consequently, $K(S\mid T)=1/k$ when $S\subset T$ and $|T|=k$, and it is zero otherwise.

We write $p_\psi^G(x)=|\langle x|\hat G|\psi\rangle|^2$ when the circuit must be specified, and omit the superscript when it is fixed. With $p_\psi(T)=|\langle T|\hat U|\psi\rangle|^2$, Theorem~\ref{thm:instrument-recovery} gives
\begin{align}
 q_\psi(S)&=\frac1k\sum_{j\notin S}p_\psi(S\cup\{j\}),\\
 R_{p_\psi}(S\cup\{j\}\mid S)
 &=\frac{p_\psi(S\cup\{j\})}
 {\sum_{\ell\notin S}p_\psi(S\cup\{\ell\})}.
 \label{eq:passive-return-rule}
\end{align}
Thus recovering a parent sample from one child sample requires only $M-k+1$ parent amplitudes to restore a particle. Note that the input mode annihilated in Eq.~\eqref{eq:passive-branches} and the restored output mode need not coincide.

For the sampling algorithms, consider the magic block of Eq.~\eqref{eq:magic-block}, with optional vacuum modes:
\begin{align}
 |\Psi_n\rangle=|\phi_4\rangle^{\otimes n}\otimes
 |0\rangle^{\otimes(M-4n)},\qquad M\ge4n.
 \label{eq:instrument-magic-input}
\end{align}
An annihilation resolves an intact block into a one-particle Fock state or removes an occupied resolved mode. Thus, the conditional inputs remain products of intact blocks and resolved occupations. In Sec.~\ref{sec:passive-algorithm}, we specialize the general approach to prove Thm.~\ref{thm:sampling-passive} and derive the runtime using the subroutines of App.~\ref{s1:sec:passive}.

\subsection{Active fermionic linear optics}
\label{sec:instrument-active}
For $V=\mathbb C^M$, an active Gaussian circuit preserves parity rather than particle number. Its invariant spaces are the even- and odd-parity subspaces of the full fermionic Fock space,
\begin{equation}
 \mathcal{F}^\sigma=\bigoplus_{k:\,(-1)^k=\sigma}\bigwedge^kV,
 \qquad \sigma\in\{+1,-1\}.
\end{equation}
The parent and child spaces are $\mathcal H_{\rm p}=\mathcal{F}^\sigma$ and $\mathcal H_{\rm c}=\mathcal{F}^{-\sigma}$. The circuit actions $U_{\rm p}$ and $U_{\rm c}$ are the corresponding restrictions of $\hat O$. Both spaces have dimension $2^{M-1}$. Fixed-particle sectors are not suitable invariant spaces, even when the input has a definite particle number, because the circuit need not preserve it.

For sampling from active FLO circuits, we choose the ladder-operator instrument with Kraus operators $\hat c_j/\sqrt M$ and $\hat c_j^\dagger/\sqrt M$, which describes particle loss and gain in orthogonal modes. It averages to the balanced loss-and-gain channel
\begin{align}
 \mathcal E(\rho)&=\frac1M\sum_{j=1}^M
 (\hat c_j\rho\hat c_j^\dagger+\hat c_j^\dagger\rho\hat c_j) =\frac1{2M}\sum_{a=1}^{2M}\hat\gamma_a\rho\hat\gamma_a.
 \label{eq:active-channel}
\end{align}
The anti-commutation relations imply trace preservation, and the orthogonality of $O$ gives Eq.~\eqref{eq:instrument-covariance}. 

For a pure parent state $\ket{\psi}$, the branch probabilities are
\begin{equation}
 \pi_{j,-}=\frac{\langle\hat n_j\rangle_\psi}{M},\qquad
 \pi_{j,+}=\frac{1-\langle\hat n_j\rangle_\psi}{M}.
 \label{eq:active-instrument-probabilities}
\end{equation}
The normalized child states are $|\psi_{j,-}\rangle=\hat c_j|\psi\rangle/\|\hat c_j|\psi\rangle\|$ and $|\psi_{j,+}\rangle=\hat c_j^\dagger|\psi\rangle/\|\hat c_j^\dagger|\psi\rangle\|$, respectively. Only branches with positive probability are sampled. Equivalently, choose a mode uniformly and then select loss or gain according to its occupation expectation. The Majorana Kraus operators $\hat\gamma_a/\sqrt{2M}$ give the same channel, but each branch is a reversible Gaussian transformation and therefore cannot reduce non-Gaussianity during the recursion.

Write an occupation outcome as a bit string $y$ of child parity. The adjoint action is
\begin{equation}
 \mathcal E^\dagger(|y\rangle\langle y|)
   =\frac1M\sum_{j=1}^M|y\oplus e_j\rangle\langle y\oplus e_j|,
\end{equation}
where $e_j$ is the $j$th unit bit string. Consequently, $K(y\mid x)=1/M$ when $x=y\oplus e_j$ for some $j\in[M]$, and it is zero otherwise. The classical channel flips one uniformly chosen bit.

With $p_\psi(x)=|\langle x|\hat O|\psi\rangle|^2$, Theorem~\ref{thm:instrument-recovery} gives
\begin{align}
 q_\psi(y)&=\frac1M\sum_{j=1}^Mp_\psi(y\oplus e_j),\\
 R_{p_\psi}(y\oplus e_j\mid y)
 &=\frac{p_\psi(y\oplus e_j)}{\sum_{\ell=1}^Mp_\psi(y\oplus e_\ell)}.
 \label{eq:active-return-rule}
\end{align}
Thus recovering a parent sample from one child sample requires only $M$ parent amplitudes to restore the parent parity, despite the exponential dimension of the parity sector. The input mode acted on by the instrument and the output mode flipped during recovery need not coincide.

For the sampling algorithms, consider the magic-block input of Eq.~\eqref{eq:instrument-magic-input}. A gain or loss resolves an intact block into a Fock state, whereas an operation on a resolved mode only flips its occupation. Thus, the conditional inputs remain products of intact blocks and resolved occupations. Unlike passive recursion, an active step does not deterministically reduce the particle number: it either resolves an intact block or flips the occupation of an already resolved mode. Section~\ref{sec:active-expected} turns this instrument into an almost-surely terminating sampler and bounds its expected runtime. In Sec.~\ref{sec:active-worst}, we group finite batches of steps to enforce progress in every recursive call, thereby obtaining a worst-case runtime bound. Both versions rely on the Pfaffian calculations of App.~\ref{s1:sec:active}.

\begin{figure*}[t]
\centering
\begingroup
\definecolor{recNavy}{HTML}{203B56}
\definecolor{recTeal}{HTML}{247E83}
\definecolor{recOchre}{HTML}{AC7135}
\begin{tikzpicture}[
  x=1cm,y=1cm,
  every node/.style={font=\small,align=center},
  space/.style={draw=recNavy!55,rounded corners=2pt,fill=recNavy!3,
    text width=2.75cm,minimum height=1.03cm,inner sep=4pt},
  recstep/.style={draw=recTeal!75,rounded corners=2pt,fill=recTeal!4,
    text width=6.5cm,minimum height=0.8cm,inner sep=4pt},
  run/.style={-{Stealth[length=2mm]},draw=black,semithick},
  lab/.style={font=\footnotesize,fill=white,inner sep=2pt}
]
\node[anchor=west,font=\small\bfseries,text=black] at (0,0.5)
  {(a) Passive FLO: remove one particle};
\node[anchor=west,font=\small\bfseries,text=black] at (8.3,0.5)
  {(b) Active FLO: switch parity};
\node[space] (pp) at (1.5,-0.55)
  {Parent\\$\bigwedge^{k}\mathbb C^M$};
\node[space] (pc) at (5.35,-0.55)
  {Child\\$\bigwedge^{k-1}\mathbb C^M$};
\draw[run] (pp) --
  node[lab,above,text height=2.5ex,text depth=0.8ex]
  {$\hat c_a$} (pc);
\node[space,text width=2.3cm] (ap) at (9.63,-0.55)
  {Parent\\$\mathcal{F}^{\sigma}$};
\node[space,text width=2.3cm] (ac) at (13.83,-0.55)
  {Child\\$\mathcal{F}^{-\sigma}$};
\draw[run] (ap) --
  node[lab,above,text height=2.5ex,text depth=0.8ex]
  {$\hat c_a$ or $\hat c_a^\dagger$} (ac);

\node[recstep] (ps) at (3.43,-2.15)
  {Recursively sample $S$,\quad $|S|=k-1$.};
\node[recstep] (as) at (11.73,-2.15)
  {Recursively sample $y$ of the opposite parity.};
\draw[run] (pc.south) -- (5.35,-1.55) -- (3.43,-1.55) -- (ps.north);
\draw[run] (ac.south)
  -- (13.83,-1.55)
  -- (11.73,-1.55)
  -- (as.north);

\node[recstep,draw=recOchre!85,fill=recOchre!5] (pr) at (3.43,-3.63)
  {Restore an output particle: $S\longmapsto S\cup\{j\}$\\
   $j\notin S$,\quad weight $p(S\cup\{j\})$.};
\node[recstep,draw=recOchre!85,fill=recOchre!5] (ar) at (11.73,-3.63)
  {Restore the parent parity: $y\longmapsto y\oplus e_\ell$\\
   $\ell\in[M]$,\quad weight $p(y\oplus e_\ell)$.};
\draw[run] (ps) -- (pr);
\draw[run] (as) -- (ar);

\node[anchor=north,align=left,text width=6.6cm,font=\footnotesize]
  at (3.43,-4.45)
  {\textbf{Recursion parameter.} Each step lowers the particle
   number: $k\to k-1$. For $n$ magic blocks, the initial values
   are $k=2n$ and $g=n$. Stop at $k=1$, where a single-particle
   distribution is sampled directly.};
\node[anchor=north,align=left,text width=6.6cm,font=\footnotesize]
  at (11.73,-4.45)
  {\textbf{Recursion parameter.} The mode is uniform in $[M]$.
   Hitting an intact magic block lowers $g\to g-1$, with probability
   $4g/M$; a resolved-mode hit leaves $g$ unchanged.
   Start at $g=n$ and stop at $g=0$ (Gaussian sampling).};
\end{tikzpicture}
\caption{The two fermionic instances of the instrument recursion.
(a) The passive channel $\mathcal D_k$ maps the $k$-particle space to
the $(k-1)$-particle space. Its classical action deletes a uniformly
chosen occupied output mode, and recovery adds one mode using parent
probabilities. (b) The active channel $\mathcal E$ interchanges the
two parity sectors without reducing their dimension. Its classical
action flips a uniformly chosen output bit; recovery chooses among
the $M$ bit-flipped parent outcomes. For the magic-block inputs,
progress of the active recursion is measured by the number $g$ of
intact non-Gaussian blocks, rather than by the dimension of the
Hilbert space. The input mode used by the instrument and the output
mode chosen during recovery need not coincide.}
\label{fig:fermionic-recursions}
\endgroup
\end{figure*}
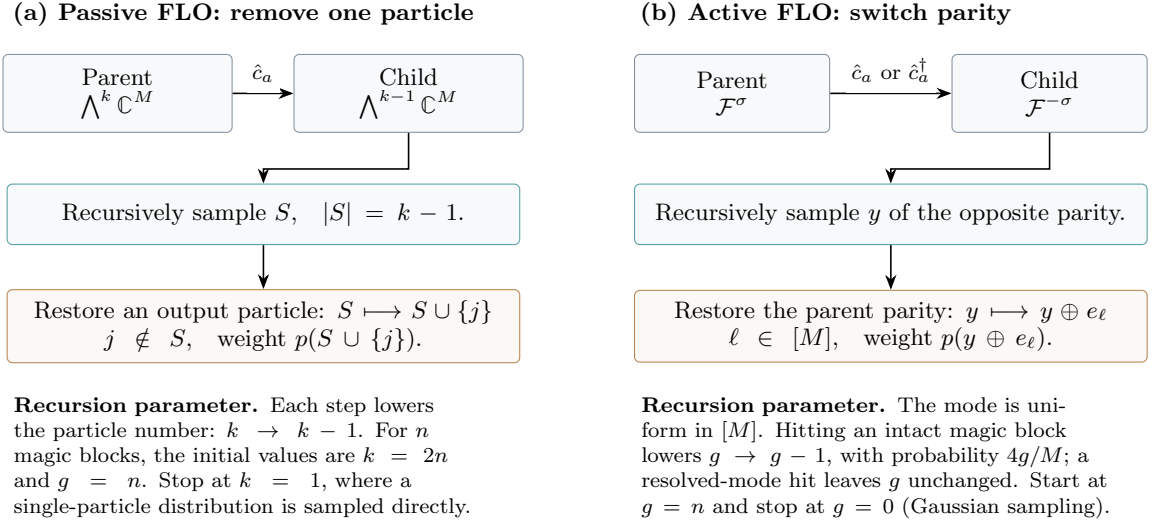

\section{Sampling problem for passive Gaussian dynamics}\label{Sec:Passive sampling}
\subsection{Exact sampling algorithm}
\label{sec:passive-algorithm}
We now apply the recursive construction of Sec.~\ref{sec:instrument-sampling} to the input $|\Psi_n\rangle=|\phi_4\rangle^{\otimes n}$ on $M=4n$ modes. The recovery rule specifies how to reconstruct a parent sample from a child sample. To obtain an efficient algorithm, we must also represent the conditional inputs and evaluate the required parent probabilities. For the present input, the conditional states retain a simple block structure, and all candidate probabilities at a return step can be calculated together.

\begin{theorem}[Exact passive sampling]
\label{thm:sampling-passive}
Given a passive Gaussian unitary specified by $U\in U(4n)$, an exact occupation-number sample from $\hat U|\Psi_n\rangle$ can be generated with per-run arithmetic cost
\begin{align}
O(n^4+2^n)=O(2^n).
\label{eq:sampling-passive-bound}
\end{align}
\end{theorem}
Throughout the sampling results, costs count exact arithmetic, conjugation, comparisons, zero tests, nonnegative square roots, and exact draws from finite represented distributions.

At any stage, the input consists of $g$ intact four-mode blocks and a set $F$ of occupied resolved modes. It has $k=|F|+2g$ particles. An annihilation on an intact block selects one of its two occupation components and leaves a one-particle Fock state; an annihilation on an occupied resolved mode empties that mode. Thus, every conditional input is specified by the remaining intact blocks and resolved occupations. Its occupation expectations are $1/2$ on intact modes and either zero or one on resolved modes, so the instrument probabilities are directly available.

The sampler keeps $U$ fixed and proceeds recursively:
\begin{enumerate}
\item If $k=1$, write the input as $|\psi\rangle=\sum_a v_a|\{a\}\rangle$, draw $j$ with probability $|(Uv)_j|^2$, and return $\{j\}$.
\item Otherwise, save the current input $|\psi\rangle$. Choose $a$ with probability $\langle\psi|\hat n_a|\psi\rangle/k$, and recursively sample the normalized child proportional to $\hat c_a|\psi\rangle$. Denote the returned occupied set by $S$.
\item Calculate all weights
\begin{align}
w_j=|\langle S\cup\{j\}|\hat U|\psi\rangle|^2,
\qquad j\notin S.
\label{eq:sampling-passive-weights}
\end{align}
Draw $j$ with probability $w_j/\sum_{\ell\notin S}w_\ell$, and return $S\cup\{j\}$.
\end{enumerate}
Every downward step removes one particle, so the procedure reaches its base case after $2n-1$ reductions. Thm.~\ref{thm:instrument-recovery}, with the passive return rule~\eqref{eq:passive-return-rule}, proves exactness by induction on $k$. The return calculation uses the saved parent input, and its normalization is positive for every child outcome that can be sampled.

The main computational step is the simultaneous evaluation of the return weights. Write the saved input as
\begin{align}
|\psi\rangle&=2^{-g/2}\sum_{y\in\{0,1\}^g}|J_y\rangle,\label{eq:sampling-residual-input}\\
J_y&=(F,X_1(y_1),\ldots,X_g(y_g)),
\label{eq:sampling-input-lists}
\end{align}
where $X_b(0)$ and $X_b(1)$ are the two occupied pairs of intact block $b$. Each pair is internally ordered, and the blocks follow the fixed input-mode ordering. For the returned set $S$, let $B_y=U_{S,J_y}$. Appending the candidate output row $j$ gives the amplitude
\begin{align}
\mathcal A_j=2^{-g/2}\sum_y
\det\begin{pmatrix}B_y\\ U_{j,J_y}\end{pmatrix}.
\label{eq:sampling-passive-amplitude}
\end{align}
The appended row order may contribute a sign relative to the ordered occupation basis; its squared modulus is $w_j$.

The candidate row enters each determinant linearly. Let $D$ be the ordered list containing $F$ and all modes of the intact blocks, and define the cofactor vector
\begin{align}
z_a(B)=(-1)^{k+a}\det B^{\setminus a},
\label{eq:sampling-cofactor}
\end{align}
where $B^{\setminus a}$ deletes column $a$. If $\iota_y$ embeds coordinates on $J_y$ into those on $D$, then
\begin{align}
h=2^{-g/2}\sum_y\iota_y z(B_y),\qquad
\mathcal A=U_{:,D}h.
\label{eq:sampling-shared-cofactor}
\end{align}
Thus, the sum over input components is performed once, in $h$, and a single matrix-vector multiplication is sufficient to compute every candidate amplitude.

The determinant subroutine in App.~\ref{s1:sec:passive} first eliminates the columns associated with $F$, which are common to all branches. It then shares the remaining calculation across the two pair choices in each intact block. The resulting binary recurrence has $O(r^2)$ work at a node with $r$ remaining blocks:
\begin{align}
T_r\le 2T_{r-1}+O(r^2),\qquad T_g=O(2^g).
\label{eq:sampling-binary-cost}
\end{align}
Including the common elimination and the final multiplication, all return weights are obtained in
\begin{align}
O(k^3+Mk+2^g)
\label{eq:sampling-passive-query}
\end{align}
operations. The appendix gives the cofactor updates and treats rank-deficient matrices.

To bound the total cost, the exponential work must be summed along the sampled path. A first annihilation in an intact block decreases $g$ by one; a subsequent annihilation removes the remaining particle from an already resolved block. There are $n$ first hits. Before the $t$th removal from a resolved block, at least $t$ blocks have already been hit. Counting the final annihilation solely for this bound therefore gives
\begin{align}
\sum_{k=1}^{2n}2^{g_k}
&\le\sum_{j=1}^n2^j+\sum_{t=1}^n2^{n-t}=3\cdot2^n-3,
\label{eq:sampling-passive-path}
\end{align}
where $g_k$ is the intact-block count at the $k$-particle parent. The polynomial work sums to $O(n^4)$ for $M=4n$, proving Thm.~\ref{thm:sampling-passive} on every sampling path.

\begin{figure*}
    \centering
    \includegraphics[width=0.95\linewidth]{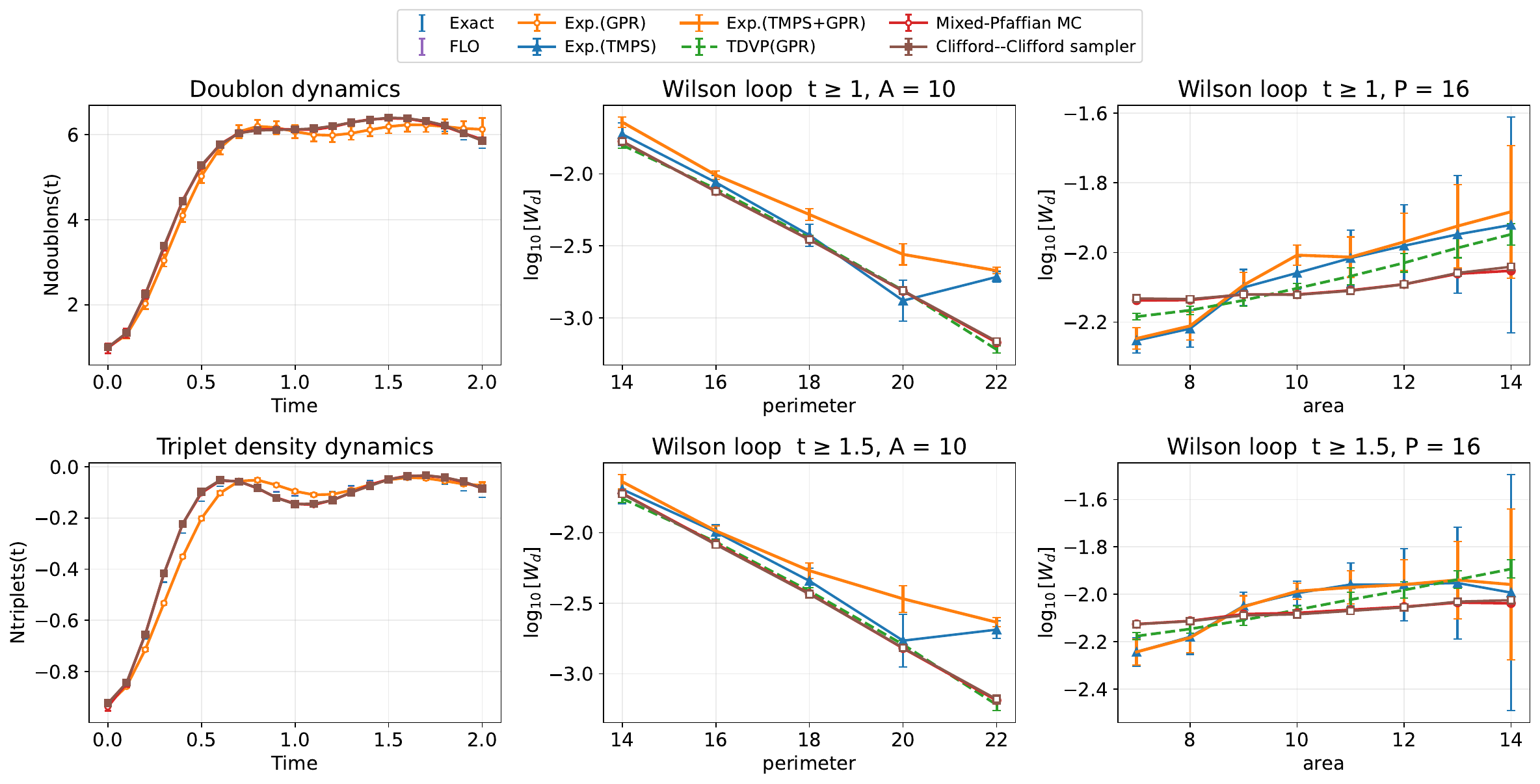}
    \caption{
    Noninteracting trapped-ion benchmark. Error bars for the exact-sampling estimates show central 68\% shot-level bootstrap percentile intervals obtained from 5,000 resamples.
    The exact passive sampler reproduces the exact free-fermion results for the doublon number and triplet density and agrees closely with the independent mixed-Pfaffian additive-error calculation of Ref.~\cite{oh2026classical} for the Wilson-loop observables.
    Both classical calculations show a small systematic deviation from the experimentally inferred Wilson-loop values.
    }
    \label{fig:phasecraft}
\end{figure*}

\subsection{Numerical simulation}
We demonstrate exact classical sampling in the noninteracting regime of the trapped-ion Fermi-Hubbard experiment presented as going beyond exact classical simulation in Ref.~\cite{Phasecraft}. The non-Gaussian input makes this regime nontrivial to simulate despite its passive Gaussian evolution. Reference~\cite{Phasecraft} obtained exact classical benchmarks for low-weight observables, while leaving exact sampling unimplemented because of its computational cost. Our algorithm makes exact sampling feasible for this experimental instance and provides a classical benchmark for its high-weight Wilson-loop observables.

We simulate 28 particles in 56 fermionic modes, with an input consisting of 13 copies of $|\phi_4\rangle$, two additional occupied modes, and two vacuum modes. We apply the passive Gaussian unitary corresponding to the experimental circuit of Ref.~\cite{Phasecraft}. At each time point, we generate $10^4$ samples from the exact output distribution and use the same samples to estimate the doublon number, triplet density, and Wilson-loop observables. The Wilson-loop estimates are further averaged over the time windows and loop geometries shown in Fig.~\ref{fig:phasecraft}, following Ref.~\cite{Phasecraft}. The observable estimates have finite-sample statistical uncertainty, with error bars obtained by bootstrap resampling.

Figure~\ref{fig:phasecraft} shows agreement with the exact free-fermion benchmarks for the doublon number and triplet density. For the high-weight Wilson loops, our results also agree closely with the independent mixed-Pfaffian additive-error calculations of Ref.~\cite{oh2026classical}. This agreement provides a cross-check between two distinct approaches: the present algorithm samples the full output distribution exactly, whereas Ref.~\cite{oh2026classical} estimates the observables directly to additive error. Both classical calculations show a small systematic discrepancy from the experimentally inferred Wilson-loop values. Our exact sampler thus provides an independent classical reference for assessing the experimental and error-mitigation results in the non-interacting regime.

\section{Sampling problem for active Gaussian dynamics}
\label{Sec:Active sampling}

\subsection{Expected-time sampling}
\label{sec:active-expected}
For active Gaussian evolution, we use the balanced gain-and-loss instrument and the conditional bit-flip recovery of Sec.~\ref{sec:instrument-active}. The conditional inputs again consist of intact blocks and resolved occupations. A nonzero ladder operation on an intact block resolves it into a Fock state, whereas an operation on a resolved mode flips its occupation. The number of intact blocks therefore never increases, although a step need not decrease it.

\begin{theorem}[Expected-time active sampling]
\label{thm:sampling-active-expected}
Given an active Gaussian unitary specified by $O\in SO(8n)$, there is an exact occupation-number sampler for $\hat O|\Psi_n\rangle$ that terminates almost surely and has expected arithmetic cost
\begin{align}
O\!\left(n^4\log(n+1)+2^n\right)=O(2^n).
\label{eq:sampling-active-expected-bound}
\end{align}
\end{theorem}

The recursive procedure is as follows:
\begin{enumerate}
\item If $g=0$, sample the Gaussian output of the remaining Fock input using a Gaussian occupation sampler~\cite{Bravyi2005}.
\item Otherwise, save the current input $|\psi\rangle$. Choose a mode $a$ uniformly, and apply loss or gain with probabilities $\langle\hat n_a\rangle_\psi$ and $1-\langle\hat n_a\rangle_\psi$, respectively. Recursively sample the normalized child, obtaining a string $x$.
\item Calculate all neighbor weights
\begin{align}
w_j=|\langle x\oplus e_j|\hat O|\psi\rangle|^2.
\label{eq:sampling-active-weights}
\end{align}
Choose $j$ with probability $w_j/\sum_\ell w_\ell$, and return $x\oplus e_j$.
\end{enumerate}
The Gaussian circuit remains fixed throughout this recursion. At a state with $g>0$ intact blocks, a uniformly chosen mode belongs to an intact block with probability $4g/M$. The waiting time $L_g$ before the next block is resolved is therefore geometric, with
\begin{align}
\mathbb E L_g=\frac{M}{4g}.
\label{eq:sampling-active-waiting}
\end{align}
The recursion reaches $g=0$ almost surely. Exactness follows from Thm.~\ref{thm:instrument-recovery} and Eq.~\eqref{eq:active-return-rule}. More explicitly, truncate the downward procedure at a finite depth and use exact endpoint sampling. The resulting sampler is exact at every truncation depth, and the probability that the truncation is reached tends to zero.

As in the passive case, the return weights share a common amplitude calculation. For a saved parent and a returned string $x$, define the Gaussian state $|\zeta\rangle=\hat O^\dagger|x\rangle$ and the transition vector
\begin{align}
h_a=\langle\zeta|\hat\gamma_a|\psi\rangle.
\label{eq:sampling-active-transition}
\end{align}
The Majorana transformation law gives
\begin{align}
(Oh)_a&=\langle x|\hat\gamma_a\hat O|\psi\rangle,
& w_j&=|(Oh)_{2j-1}|^2.
\label{eq:sampling-active-readout}
\end{align}
The last identity follows because $\hat\gamma_{2j-1}$ flips occupation $j$, up to a sign. Consequently, one transition vector and one matrix--vector multiplication yield all $M$ neighbor weights.

To calculate $h$, introduce the two ladder-transition vectors
\begin{align}
a_b^-&=\langle\zeta|\hat c_b|\psi\rangle,
\qquad a_b^+=\langle\zeta|\hat c_b^\dagger|\psi\rangle,\\
h_{2b-1}&=a_b^-+a_b^+,
\qquad h_{2b}=-i(a_b^--a_b^+).
\label{eq:sampling-ladder-vectors}
\end{align}
Expanding the intact blocks as in Eq.~\eqref{eq:sampling-residual-input} expresses these quantities as sums of Gaussian occupation amplitudes, represented by Pfaffians. Their common structure permits two levels of reuse: all insertion and deletion amplitudes are obtained from a shared Pfaffian vector, and the calculation of that vector is shared across the two occupation choices of each intact block.

Appendix~\ref{s1:sec:active} develops this routine and proves that all neighbor weights, up to a common positive factor, can be calculated in
\begin{align}
O(M^3+2^g)
\label{eq:sampling-active-query}
\end{align}
operations. The same bound applies to an absolute scalar Born probability. The routine includes singular cases and does not require a nonzero vacuum overlap. Its binary calculation obeys the same $O(2^g)$ recurrence as the passive cofactor routine.

There is one return calculation for every downward step. The expected exponential work is
\begin{align}
\mathbb E\!\left[\sum_{g=1}^nL_g2^g\right]
&=\frac M4\sum_{g=1}^n\frac{2^g}{g}=O\!\left(\frac Mn2^n\right).
\label{eq:sampling-active-total}
\end{align}
The polynomial work contributes $O(M^4H_n)$, where $H_n=\sum_{g=1}^n1/g$. Setting $M=4n$ gives Thm.~\ref{thm:sampling-active-expected}. Although the expected number of visits is larger at small $g$, the corresponding amplitude calculations are exponentially cheaper.

This procedure has no finite bound on the length of an individual sampling path: resolved modes may be selected arbitrarily many times before another intact block is visited. In section~\ref{sec:active-worst} show how to modify the outer recursion to obtain a deterministically bounded runtime.

\subsection{Worst-case sampling}
\label{sec:active-worst}
The expected-time sampler can revisit resolved modes arbitrarily many times before reaching an intact block. A worst-case guarantee therefore requires progress after a bounded amount of work. We obtain this by grouping finite batches of instrument paths and replacing selected groups of endpoints by Gaussian mixtures. Every nonterminal recursive call then has strictly fewer intact blocks.

\begin{samepage}
\begin{theorem}[Worst-case active sampling]
\label{thm:sampling-active-worst}
For every $O\in SO(8n)$, an exact occupation-number sample from $\hat O|\Psi_n\rangle$ can be generated with per-run arithmetic cost
\begin{align}
    O(n^5+2^n)&=O(2^n)
    \label{eq:sampling-active-worst-bound}
\end{align}
and at most $n$ nonterminal recursive calls.
\end{theorem}
\end{samepage}

A finite batch of loss and gain operations can still visit only resolved modes, leaving every intact block unchanged. The key is to combine these histories with histories that resolve one block, grouping endpoints that agree everywhere else. Within each group, the selected block is described by an averaged density operator containing both the original non-Gaussian state and a uniform even-parity contribution. We choose the grouping so that this average admits a Gaussian ensemble. Sampling a component of that ensemble then removes one intact block from the recursive input. Progress is therefore guaranteed for the input passed to the recursive call, even when some of the original histories in the group make no progress.

More precisely, the grouping supplies enough of the maximally mixed even-parity state $\hat\Pi_e/8$ for the averaged block state to be a convex combination of $\hat\Pi_e/8$ and
\begin{align}
    \hat\sigma&=\frac37|\phi_4\rangle\langle\phi_4|+\frac47\frac{\hat\Pi_e}{8},
    \label{eq:sampling-mixture-threshold}
\end{align}
which is proved in App.~\ref{s1:app:group-count}. Here, $\hat\Pi_e$ is the even-parity projector on the four-mode block. Both states admit explicit finite ensembles of pure Gaussian states; the ensemble for $\hat\sigma$ is constructed in App.~\ref{s1:app:ensemble}. Sampling such a Gaussian component preserves the group's averaged output distribution and removes one intact block from the recursive input.

A recursive call has the following structure. Here $g$ is the number of intact blocks, and the current Gaussian circuit includes any Gaussian preparations absorbed at earlier calls.
\begin{enumerate}
    \item \textbf{Generate and group a batch.} Sample a bounded sequence of loss and gain operations using Eq.~\eqref{eq:active-instrument-probabilities}. Group the selected paths by their common endpoint outside a chosen block. The remaining paths are retained individually; each of their endpoints already has fewer than $g$ intact blocks.
    \item \textbf{Sample the endpoint output.} For an individual path, recursively sample its endpoint. For a selected group, sample a Gaussian component of the averaged block state, absorb its preparation into the circuit, and recursively sample the resulting input. In either case, the number of intact blocks decreases.
    \item \textbf{Restore the path and recover.} Given the sampled endpoint output, draw an original instrument path from its conditional distribution within the group. Apply the recovery rule in Eq.~\eqref{eq:active-return-rule} backwards along this path, using the circuit and input states saved by the parent call.
\end{enumerate}
If no intact blocks remain, the terminal call uses Gaussian occupation sampling. When only a sufficiently small fraction remains, a bounded gate-by-gate sampler terminates the recursion. Its cost contains a polynomial factor multiplying $2^g$, where $g$ is the remaining block count; the reduction from $n$ to $g$ makes this terminal cost fit within $O(n^5+2^n)$. Appendix.~\ref{s1:sec:upgrade} specifies the batch length, grouping rule, terminal threshold, and complete implementation.

The Gaussian component used by the recursive call is an auxiliary representation of the grouped endpoint state. Recovery still requires a history of the original loss and gain instrument. After the recursive call returns an output, we therefore sample an original history conditional on both its group and that output. This conditioning is necessary because different histories generally assign different probabilities to the returned output. It restores the original joint distribution of histories and outputs, allowing the recovery updates to be applied in reverse order. The posterior weights each history by its prior probability and its endpoint's Born probability for the returned output. The sampled Gaussian component is discarded once the original history is restored.

A restored history may leave all intact blocks unchanged, but it is used only during recovery and does not initiate another recursive call. Thus, restoring the history preserves exactness without undoing the progress already made by the recursion. Appendix.~\ref{s1:app:bridge} implements this conditional draw using a fixed number of endpoint probability queries and a counting recurrence for the paths.

The cost bound uses both this strict progress and the shared probability calculations of App.~\ref{s1:sec:active}. After $k=n-g$ blocks have been resolved, a batch has $O(k+1)$ steps. Its exponential work is therefore $O((k+1)2^{n-k})$. Since $k$ strictly increases between recursive calls, the total exponential contribution obeys
\begin{align}
    \sum_{k=0}^{n-1}(k+1)2^{n-k}&=O(2^n).
    \label{eq:sampling-worst-scaling}
\end{align}
This summation is what avoids a polynomial prefactor multiplying $2^n$. The accumulated polynomial work is $O(n^5)$, and the single terminal call also fits the bound in Thm.~\ref{thm:sampling-active-worst}. The full correctness and cost analysis is given in App.~\ref{s1:app:worst-proof}.

\section{Exact expectation-value problem}\label{Sec:Mean value}

\begin{figure*}
    \centering
    \definecolor{meanTeal}{HTML}{247E83}
    \resizebox{0.8\textwidth}{!}{%
    \begin{tikzpicture}[
        tablebox/.style={draw=meanTeal!75,rounded corners=2pt,fill=meanTeal!4,minimum width=2.2cm,minimum height=0.8cm,align=center,font=\small},
        updatebox/.style={draw=gray!55,rounded corners=2pt,minimum height=1.25cm,align=center,font=\small},
        flowarrow/.style={-{Stealth[length=2mm]},draw=black}
    ]
        \node[anchor=west,font=\small\bfseries] at (0,0.7) {(a) Add one input block at a time};
        \node[tablebox] (f0) at (1.3,0) {$f_0(S)=\mathbf1_{S=\varnothing}$};
        \node[tablebox] (f1) at (5.1,0) {$f_1(S)$};
        \node[font=\large] (dots) at (8.4,0) {$\cdots$};
        \node[tablebox] (fn) at (11.6,0) {$f_n(S)$};
        \draw[flowarrow] (f0)--node[above,font=\footnotesize]{block $1$}(f1);
        \draw[flowarrow] (f1)--(dots);
        \draw[flowarrow] (dots)--node[above,font=\footnotesize]{block $n$}(fn);
        \node[anchor=west,font=\small\bfseries] at (0,-1.15) {(b) Compute one complete table $f_q$};
        \node[updatebox,text width=3.0cm] (init) at (1.6,-2.15) {Initialize one vector per subset\\[5pt]$v_S=f_{q-1}(S)|\phi_4\rangle$};
        \node[updatebox,text width=5.5cm] (upd) at (7.0,-2.15) {For $t=k,\ldots,1$ and $t\in S$:\\[5pt]$v_S\leftarrow v_S+(-1)^{|S\cap[t-1]|}\hat\eta_t^{[q]}v_{S\setminus\{t\}}$};
        \node[updatebox,text width=2.5cm] (out) at (12.2,-2.15) {Contract\\[5pt]$f_q(S)=\langle\phi_4|v_S\rangle$};
        \draw[flowarrow] (init)--(upd);
        \draw[flowarrow] (upd)--(out);
    \end{tikzpicture}%
    }
    \caption{Dynamic programming for exact expectation values. (a) The scalar table $f_q$ stores the auxiliary-state coefficients after the first $q$ input blocks. The final entry gives $\mu=i^{k(k-1)/2}f_n([k])$, with the complementary-string sign restored when applicable. (b) Each block update applies the factors $\hat I+\hat\eta_t^{[q]}\otimes\hat b_t^\dagger$ in the order $t=k,\ldots,1$, using one local vector per subset. Auxiliary creation supplies the displayed sign. The source vector $v_{S\setminus\{t\}}$ is unchanged during pass $t$, allowing in-place updates. For the illustrated four-mode input, every vector has dimension $16$.}
    \label{fig:Dynamic programming}
\end{figure*}

We consider an input $|\Psi_{\mathrm{in}}\rangle$ consisting of $n$ even-parity blocks of constant size, evolved by a Gaussian unitary $\hat O$ with Majorana matrix $O\in\operatorname{SO}(2M)$. The observable is a Hermitian Majorana monomial $\hat\Gamma$ of weight $0\leq k\leq2M$, as defined in Eq.~\eqref{eq:majorana-string}.

Our algorithm adds one input block at a time and stores a table indexed by subsets of the transformed Majorana factors. The main difficulty is to account for their anti-commutation signs while sharing calculations across all subsets. We represent the subset indices as occupation states of auxiliary fermionic modes. Their creation operators supply exactly the required signs, and their factorized action gives an efficient update of the entire table. The auxiliary Fock space is a representation of the table, so this construction introduces no additional physical modes to simulate.

Fermionic parity gives a further reduction~\cite{Bravyi2005,Parity}: odd-weight strings have zero expectation, while even-weight strings can be replaced by their complements with a known sign, specified in Eq.~\eqref{eq:mean-complementary-string}. We therefore use the effective weight
\begin{align}
    K&:=\min\{k,2M-k\}.
\end{align}

\begin{theorem}\label{thm:exact-expectation-value}
Under the assumptions above, there is a classical algorithm that exactly computes
\begin{align}
    \mu&=\langle\Psi_{\mathrm{in}}|\hat O^\dagger\hat\Gamma\hat O|\Psi_{\mathrm{in}}\rangle
\end{align}
in $O(n^2+nK2^K)$ time.
\end{theorem}

\begin{proof}[Proof sketch]
We illustrate the construction for $|\phi_4\rangle^{\otimes n}$, for which $M=4n$. Appendix.~\ref{App:Main} proves the result for general block-product states and also allows different bra and ket states. By the parity reduction above, it suffices to consider even $k$ with $0<k\leq M$; the identity case is immediate. The complementary-string sign is restored at the end when this reduction is used.

Write $\hat\Gamma=i^{k(k-1)/2}\hat\gamma_{a_1}\cdots\hat\gamma_{a_k}$ with $a_1<\cdots<a_k$, and set $\hat\eta_t:=\hat O^\dagger\hat\gamma_{a_t}\hat O$. Gaussian evolution decomposes each factor over input blocks as
\begin{align}
    \hat\eta_t&=\sum_{q=1}^n\hat\eta_t^{[q]},\qquad \hat\eta_t^{[q]}:=\sum_{r=8q-7}^{8q}O_{a_t r}\hat\gamma_r.
\end{align}
For $A=\{t_1<\cdots<t_s\}\subseteq[k]$, define
\begin{align}
    \hat\eta_A^{[q]}&:=\hat\eta_{t_1}^{[q]}\cdots\hat\eta_{t_s}^{[q]},\qquad c_q(A):=\langle\phi_4|\hat\eta_A^{[q]}|\phi_4\rangle,
\end{align}
where the matrix element uses the local representation on block $q$. In particular, $c_q(\varnothing)=1$ and $c_q(A)=0$ for odd $|A|$.

Introduce $k$ auxiliary fermionic modes with creation operators $\hat b_t^\dagger$, attached to the physical space by an ordinary tensor product. Operators belonging to the two spaces therefore commute, while the canonical anti-commutation relations hold within each space. For an increasing subset $A$, write $\hat b_A^\dagger$ for the ordered product of its creation operators and $|A\rangle:=\hat b_A^\dagger|\varnothing\rangle$. The operator
\begin{align}
    \hat M_q&:=\sum_{A\subseteq[k]}c_q(A)\hat b_A^\dagger
\end{align}
records all ways to assign factors to block $q$. Multiplying these operators excludes repeated indices by nilpotence and supplies the fermionic reordering signs automatically. More explicitly, the two minus signs from exchanging physical factors on different blocks and auxiliary creation operators cancel. Together with block factorization, this gives the representation proved in Lem.~\ref{lem:mean-aux-factorization}. Since only even subsets contribute, the $\hat M_q$ commute, and we can apply them in block order.

Starting from the auxiliary vacuum $|w_0\rangle:=|\varnothing\rangle$, define
\begin{align}
    |w_q\rangle&:=\hat M_q|w_{q-1}\rangle,\qquad f_q(S):=\langle S|w_q\rangle.
    \label{eq:mean-dp-recurrence}
\end{align}
Thus $q$ counts the processed blocks, while $S$ labels an entry in the same $2^k$-entry table at every step. The desired coefficient is
\begin{align}
    \mu&=i^{k(k-1)/2}f_n([k]).
    \label{eq:mean-dp-output}
\end{align}
We do not construct $\hat M_q$ as a matrix or evaluate its coefficients separately. The local-vector update below applies it in $O(k2^k)$ time. All $n$ blocks therefore cost $O(nk2^k)$, including construction of the local operators. Reading the supplied dense matrix $O$ adds $O(M^2)=O(n^2)$. Using the shorter complementary string when necessary gives $O(n^2+nK2^K)$.
\end{proof}

\subsection{Table update rule}
\label{Sec:Mean subroutine}

Suppose that the complete preceding table $f_{q-1}$ is available. To apply $\hat M_q$, we keep the physical state of block $q$ until all factors have acted, and contract with its bra only at the end. In the local representation,
\begin{align}
    \hat M_q&=\langle\phi_4|\prod_{t=1}^k\bigl(\hat I+\hat\eta_t^{[q]}\otimes\hat b_t^\dagger\bigr)|\phi_4\rangle,
\end{align}
Here, the bra and ket contract only the physical block, with identity operators on the auxiliary space implicit. The factors are ordered increasingly from left to right, so acting on a vector processes $t=k,\ldots,1$.

\begin{lemma}
\label{lem:mean-table-update}
Given the preceding table and the local operators $\hat\eta_1^{[q]},\ldots,\hat\eta_k^{[q]}$, the next table can be computed in $O(k2^k)$ time.
\end{lemma}

\begin{proof}[Proof sketch]
Use one local vector per subset and perform the following steps:
\begin{enumerate}
    \item Initialize
    \begin{align}
        v_S&:=f_{q-1}(S)|\phi_4\rangle,\qquad S\subseteq[k].
        \label{eq:mean-subroutine-initialization}
    \end{align}
    \item For $t=k,\ldots,1$, update all subsets containing $t$ by
    \begin{align}
        v_S&\leftarrow v_S+(-1)^{|S\cap[t-1]|}\hat\eta_t^{[q]}v_{S\setminus\{t\}},\qquad t\in S,
        \label{eq:mean-subroutine-update}
    \end{align}
    leaving the other vectors unchanged, with $[0]=\varnothing$.
    \item Contract to obtain the complete new table:
    \begin{align}
        f_q(S)&=\langle\phi_4|v_S\rangle.
        \label{eq:mean-subroutine-output}
    \end{align}
\end{enumerate}
The sign in step 2 is precisely the sign of $\hat b_t^\dagger|S\setminus\{t\}\rangle$. The source vector has no $t$ in its index and remains unchanged during that pass, so the update is performed in place. Lem.~\ref{lem:mean-one-block-update} proves that these steps apply $\hat M_q$ exactly. Each vector has dimension $16$, and each update uses one local matrix--vector multiplication. The $k$ passes cost $O(k2^k)$; initialization and contraction cost $O(2^k)$.
\end{proof}

\section{Approximate expectation-value problem}
\label{Sec:Approximate mean value}

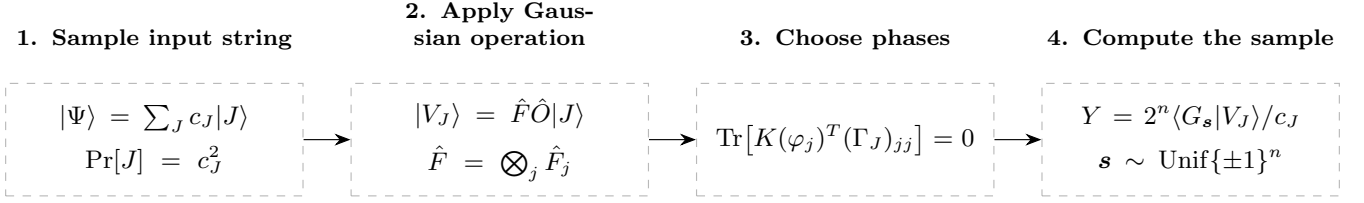
\begin{figure*}
    \centering
    \resizebox{\linewidth}{!}{%
    \begin{tikzpicture}[
        flowbox/.style={draw=gray!55,dashed,align=center,text width=3.35cm,minimum height=1.4cm,inner sep=5pt,font=\small},
        flowtitle/.style={font=\footnotesize\bfseries,align=center,text width=3.75cm},
        flowarrow/.style={-{Stealth[length=2mm]},draw=black}
    ]
        \node[flowbox] (input) at (0,0) {$|\Psi\rangle=\sum_Jc_J|J\rangle$\\[5pt]$\Pr[J]=c_J^2$};
        \node[flowbox] (filter) at (4.3,0) {$|V_J\rangle=\hat F\hat O|J\rangle$\\[5pt]$\hat F=\bigotimes_j\hat F_j$};
        \node[flowbox] (phase) at (8.6,0) {$\Tr[K(\varphi_j)^T(\Gamma_J)_{jj}]=0$};
        \node[flowbox] (sample) at (12.9,0) {$Y=2^n\langle G_{\bm s}|V_J\rangle/c_J$\\[5pt]$\bm s\sim\operatorname{Unif}\{\pm1\}^n$};
        \node[flowtitle,above=3mm of input] {1. Sample input string};
        \node[flowtitle,above=3mm of filter] {2. Apply Gaussian operation};
        \node[flowtitle,above=3mm of phase] {3. Choose phases};
        \node[flowtitle,above=3mm of sample] {4. Compute the sample};
        \draw[flowarrow] (input.east)--(filter.west);
        \draw[flowarrow] (filter.east)--(phase.west);
        \draw[flowarrow] (phase.east)--(sample.west);
    \end{tikzpicture}%
    }
    \caption{Schematic illustration of sample generation for our approximate expectation-value algorithm. We first sample an occupation string $J$ from the block-product input and construct the Gaussian vector $|V_J\rangle$ by applying the Gaussian circuit and a generally non-unitary Gaussian operation. We then choose the phases based on $|V_J\rangle$ and sample a Gaussian state from the corresponding decomposition. Evaluating and rescaling their overlap yields an unbiased estimator $Y$ of the target amplitude $\langle\Phi|\hat O|\Psi\rangle$, with second moment bounded by one. Here $\Gamma_J$ is the covariance matrix of the normalized $|V_J\rangle$, and $K(\varphi)$ is the covariance matrix of $|g_+(\varphi)\rangle$.}
    \label{fig:Approximate expectation value}
\end{figure*}

We extend the additive-error estimation results of Ref.~\cite{oh2026classical} from passive to active Gaussian evolution for products of four-mode two-particle states with arbitrary coefficients. The quantities of interest include multipoint number correlators of arbitrary weight, individual output probabilities, marginal probabilities for specified occupations on selected modes, and binned probabilities for integer-weighted occupation sums with polynomially bounded range. These quantities are expectation values of bounded observables, and estimates obtained from $N$ experimental samples have statistical uncertainty of order $N^{-1/2}$.

The observable reductions used in the passive setting remain valid for active evolution: number correlators and occupation projectors reduce to averages of Gaussian transition amplitudes through parity operators, while binned probabilities are obtained from Fourier components generated by occupation-dependent phases. The main task is therefore to estimate these amplitudes with a controlled second moment. A direct extension of the passive overlap estimator remains unbiased, but has a variance that can grow exponentially under active evolution. Our construction samples an occupation string from one block-product state and transfers the coefficients of the other into a positive non-unitary Gaussian operator acting on the resulting Gaussian ket. The remaining bra is a product of equal-weight two-particle blocks, whose Gaussian decomposition is chosen using the ket's covariance matrix. This last choice controls each overlap contribution. We prove a Gaussian compression bound which controls the combined second moment. 

\begin{samepage}
\begin{theorem}\label{thm:active-flo-inner-product-estimation}\label{thm:active-flo-expectation-value-estimation}
Let $\hat G$ be a fermionic Gaussian unitary on $4n$ modes, given by a polynomial-size Gaussian circuit including its overall phase, and let
\begin{align}
    |\Psi\rangle &= \bigotimes_{j=1}^n |\psi_j\rangle,
    & |\Phi\rangle &= \bigotimes_{j=1}^n |\phi_j\rangle,
\end{align}
where all blocks are normalized four-mode even-parity states. Given $0<\epsilon,\delta<1$, a randomized classical algorithm returns an estimate $\widetilde a$ of $a=\langle\Phi|\hat G|\Psi\rangle$ satisfying
\begin{align}
    \Pr\!\left[|\widetilde a-a|\leq\epsilon\right] &\geq 1-\delta,
\end{align}
in $O(n^3\epsilon^{-2}\log(2/\delta))$ arithmetic operations after a one-time polynomial-time preprocessing of the Gaussian circuit that retains its overall phase. The same guarantee holds for estimating $\langle x|\hat G|\Psi\rangle$ for any occupation string $x\in\{0,1\}^{4n}$.
\end{theorem}
\end{samepage}

We describe the construction in the same conventions as App.~\ref{app:approx}. Write $|a\rangle:=|1100\rangle$ and $|b\rangle:=|0011\rangle$. Local Gaussian transformations bring the two sets of input blocks to the canonical forms
\begin{align}
    |\phi_j\rangle&=u_j|a\rangle+v_j|b\rangle,\qquad |\psi_j\rangle=c_{j,a}|a\rangle+c_{j,b}|b\rangle,
\end{align}
with nonnegative normalized coefficients~\cite{Oszmaniec2014}. Absorb these transformations into the circuit and denote the remaining Gaussian unitary by $\hat O$. The target amplitude retains the form $a=\langle\Phi|\hat O|\Psi\rangle$ in this canonical representation.

The bra state is represented using the positive Gaussian operation $\hat F_j$ of Eq.~\eqref{eq:app-bounded-filter}. They satisfy $\hat F_j|a\rangle=u_j|a\rangle$ and $\hat F_j|b\rangle=v_j|b\rangle$, so
\begin{align}
    \langle\phi_j|&=\sqrt2\langle\psi_{\rm eq}|\hat F_j,\qquad |\psi_{\rm eq}\rangle:=\frac{|a\rangle+|b\rangle}{\sqrt2}.
\end{align}
These general Gaussian operations also cover zero coefficients; they need not be invertible. Expanding the ket as $|\Psi\rangle=\sum_{J\in\{a,b\}^n}c_J|J\rangle$ gives the product distribution $\Pr[J]=c_J^2$. We sample $J$ from this distribution and form the unnormalized Gaussian vector
\begin{align}
    |V_J\rangle&:=\hat F\hat O|J\rangle,\qquad \hat F:=\bigotimes_{j=1}^n\hat F_j.
\end{align}
A zero vector contributes zero and requires no further calculation.

For a nonzero $|V_J\rangle$, we use the Gaussian states
\begin{align}
    |g_s(\varphi)\rangle&:=\frac{|\psi_{\rm eq}\rangle+s|e_\varphi\rangle}{\sqrt2},\qquad s\in\{\pm1\},\\
    |e_\varphi\rangle&:=\frac{e^{-i\varphi/2}|0000\rangle+e^{i\varphi/2}|1111\rangle}{\sqrt2}.
\end{align}
The average over $s$ is $|\psi_{\rm eq}\rangle/\sqrt2$, independently of $\varphi$. This freedom lets us choose each phase from the covariance matrix of the normalized $|V_J\rangle$, before drawing any signs. Lem.~\ref{lem:branch-bound} shows that these choices bound the overlap for every sign string simultaneously. We then draw independent uniform signs $s_j$ and evaluate
\begin{align}
    Y&:=\frac{2^n}{c_J}\langle G_{\bm s}|V_J\rangle,\qquad |G_{\bm s}\rangle:=\bigotimes_{j=1}^n|g_{s_j}(\varphi_j)\rangle.
\end{align}
Only strings with $c_J>0$ can be sampled. Gaussian overlap evaluation, including its phase, takes polynomial time.

Averaging over the signs and the occupation string gives the target amplitude. To control the sampling cost, the factors $c_J^2$ cancel the importance weights in the second moment, leaving a trace over the canonical occupation subspace. The Gaussian compression bound in App.~\ref{app:approx} bounds this trace by the fourth Schatten norm of $\hat F^2$, which equals one. Consequently,
\begin{align}
    \mathbb E[Y]&=a,\qquad \mathbb E[|Y|^2]\leq1.
    \label{eq:active-amplitude-second-moment}
\end{align}
Applying median-of-means estimation to the real and imaginary parts yields Thm.~\ref{thm:active-flo-expectation-value-estimation}. The occupation-string draw uses the same importance-sampling principle as Thm.~4 of Ref.~\cite{oh2026classical}; the phase choice and Gaussian compression supply the bound for active evolution. After the circuit is compiled once with its phase retained, each estimator sample costs $O(n^3)$ arithmetic operations; App.~\ref{app:approx} gives the implementation and cost analysis.

When one input is Gaussian, the construction omits the occupation-string draw. The resulting unbiased estimator $Z_x$ of $\langle x|\hat G|\Psi\rangle$ satisfies the stronger bound
\begin{align}
    \mathbb E[|Z_x|^2]&\leq\prod_{j=1}^n\lambda_j\leq1,\qquad \lambda_j:=\max\{c_{j,a}^2,c_{j,b}^2\},
\end{align}
as shown in Eq.~\eqref{eq:app-gaussian-block-bound}. Here the non-unitary Gaussian operation is associated with the block-product state $|\Psi\rangle$; equivalently, apply the construction to the conjugate amplitude. For maximally non-Gaussian blocks, each $\lambda_j=1/2$.

The transition-amplitude estimator gives additive-error estimates of number correlators through the same reduction used for passive evolution in Ref.~\cite{oh2026classical}.

\begin{samepage}
\begin{theorem}[Number correlators]\label{thm:active-number-correlators}
Let $|\Psi\rangle$ be a product of normalized four-mode even-parity states and $\hat G$ a fermionic Gaussian unitary on $4n$ modes. For any $S\subseteq[4n]$ and $0<\epsilon,\delta<1$, there is a randomized classical algorithm that estimates
\begin{align}
    \mu_S &= \langle\Psi|\hat G^\dagger\Bigl(\prod_{j\in S}\hat n_j\Bigr)\hat G|\Psi\rangle
\end{align}
by a real number $\widetilde\mu_S$ satisfying
\begin{align}
    \Pr\!\left[|\widetilde\mu_S-\mu_S|\leq\epsilon\right] &\geq 1-\delta,
\end{align}
in $O(n^3\epsilon^{-2}\log(2/\delta))$ arithmetic operations after a one-time polynomial-time preprocessing of the Gaussian circuit, with no restriction on $|S|$.
\end{theorem}
\end{samepage}

\begin{proof}
For $T\subseteq S$, define the Gaussian parity operator $\hat P_T:=\prod_{j\in T}(\hat I-2\hat n_j)$. The passive reduction of Ref.~\cite{oh2026classical} gives
\begin{align}
    \mu_S &= 2^{-|S|}\sum_{T\subseteq S}(-1)^{|T|}\langle\Psi|\hat G^\dagger\hat P_T\hat G|\Psi\rangle.
\end{align}
Sample $T\subseteq S$ uniformly, use the transition-amplitude estimator above for the Gaussian unitary $\hat G^\dagger\hat P_T\hat G$, and multiply the real part of its estimator by $(-1)^{|T|}$. This gives an unbiased real estimator of $\mu_S$ with second moment at most one, so median-of-means estimation using $O(\epsilon^{-2}\log(2/\delta))$ independent samples proves the claim. The compiled representations of $\hat G$ and $\hat G^\dagger$ are reused for every draw of $T$. Since $\hat P_T$ is a product of single-mode parity operators, applying $\hat G^\dagger\hat P_T\hat G$ to a Gaussian state costs $O(n^3)$ per draw, as in App.~\ref{app:approx}.
\end{proof}

This includes the four-mode two-particle inputs with arbitrary coefficients in Eq.~\eqref{eq:paired-block}. Replacing any selected $\hat n_j$ by $\hat I-\hat n_j$ only changes the signs in the parity expansion. The same error and runtime guarantees therefore apply to individual output probabilities and marginal probabilities for specified occupations on a subset of modes.

For integer-weighted occupation sums with polynomially bounded range, the Fourier approach of Refs.~\cite{oh2026classical,oh2024quantum} expresses binned probabilities in terms of Gaussian transition amplitudes. The same reduction applies to active evolution, so the estimator above reconstructs these probabilities with polynomial overhead.

\section{Discussion}\label{Sec:Discussion}

We have developed classical algorithms for simulating free-fermion circuits with non-Gaussian block-product inputs. For $n$ copies of a four-mode magic state, our algorithms generate exact samples under both passive and active evolution in $O(2^n)$ time. We have also established polynomial-time exact computation of logarithmic-weight Majorana expectations for arbitrary constant-size even-parity blocks, and polynomial-time additive-error estimation of overlaps and number correlators for general even-parity four-mode blocks under active evolution. The latter result also covers individual output probabilities but does not, by itself, yield an efficient sampler for the full output distribution. In particular, because inverse-polynomial additive accuracy does not provide relative-error control over exponentially small probabilities. These results provide classical benchmarks for fermionic quantum experiments and clarify how the available simulation guarantees depend on the input structure, the observable, and the required accuracy.

Our expectation-value algorithms support classical training and evaluation of fermionic quantum machine-learning models. Fermionic Born Machines use Pauli-$Z$ correlators to evaluate their training losses~\cite{bako2025fermionic}. Our algorithms may also be useful for studying extensions of the fermionic learning architectures in Ref.~\cite{kerenidis2026scalable} to arbitrary pure states of definite parity on constant-size input blocks and higher-order occupation-correlator readouts. Each such correlator can be evaluated by our exact algorithm in polynomial time when its weight is logarithmic in the number of blocks.

Several open questions remain. First, what is the optimal classical complexity of the tasks considered here? In particular, for the exact expectation-value problem, our algorithm runs in polynomial time when the effective weight $K=\min\{k,2M-k\}$ is $O(\log n)$. Establishing hardness results for larger effective weights would help clarify the limits of efficient exact computation for block-product inputs.

Another question is whether the states we have developed our algorithms for are the most general possible. In the case of passive bosonic linear optical circuits, polynomial-time additive-error estimation extends to larger constant-size blocks. It is tempting to conjecture that the additive error algorithm we have developed here can also be generalised to larger blocks. Additionally it is possible that the sampling algorithms we have developed can apply more broadly, for example to arbitrary tensor products of $4$-mode parity eigenstates.

\begin{acknowledgements}
Generative AI tools were used to assist with exploratory calculations, proofs, literature searches, and language editing. All results, proofs, and references were independently verified by the authors, who take full responsibility for the content of the manuscript. 
J.H. and C.O. were supported by the National Research Foundation of Korea Grants (No. RS-2024-00431768 and No. RS-2025-00515456) funded by the Korean government (Ministry of Science and ICT (MSIT)) and the Institute of Information \& Communications Technology Planning \& Evaluation (IITP) Grants funded by the Korean government (MSIT) (No. RS-2024-00437284, No. IITP-2025-RS-2025-02283189 and No. IITP-2025-RS-2025-02263264) by Global Partnership Program of Leading Universities in Quantum Science and Technology (RS-2025-08542968) through the National Research Foundation of Korea~(NRF) funded by the Korean government (Ministry of Science and ICT(MSIT)). MO and ORS acknowledge the support from National Science Center, Poland within the QuantERA III Programme (No 2023/05/Y/ST2/00140 acronym Tuquan). This work was also supported in part by the Horizon Europe project Quantum Excellence Centre for Quantum-Enhanced Applications (QEC4QEA).
    The C4QEC project is carried out within the IRAP of the Foundation for Polish Science co-financed by the European Union.
\end{acknowledgements}

\clearpage
\appendix
\onecolumngrid

\begingroup
\renewcommand{\tocname}{Contents of the appendices}
\makeatletter
\let\l@subsubsection\@gobbletwo
\makeatother
\let\AppendixTOCSavedContentsline\contentsline
\AppendixTOCHide
\tableofcontents
\endgroup
\addtocontents{toc}{\protect\AppendixTOCShow}

\section{Bosonic Clifford--Clifford sampling from the loss instrument}
\label{app:bosonic-cc}

The instrument-and-recovery construction also gives a direct derivation of the original bosonic Clifford--Clifford algorithm~\cite{clifford2018classical, clifford2024faster} (and its extensions to input with collisions~\cite{Brod2020classicalsimulation}). The essential change from passive fermions is that an occupied output mode can contain several particles, so the classical deletion kernel is weighted by occupation number.

Let $\hat a_j$ be bosonic annihilation operators, with $[\hat a_i,\hat a_j^\dagger]=\delta_{ij}\hat I$, and let $\hat\Gamma_k(U)$ denote the restriction of a passive interferometer to the $k$-boson sector. The parent and child spaces are
\begin{align}
 \mathcal H_{\mathrm p}=\operatorname{Sym}^{k}(\mathbb C^M),
 \qquad \mathcal H_{\mathrm c} =\operatorname{Sym}^{k-1}(\mathbb C^M).
 \label{eq:boson-spaces}
\end{align}
On these spaces, define the normalized one-particle loss channel and its mode-resolved instrument by
\begin{align}
 \mathcal D_k^{\mathrm B}(\hat\rho)=\frac1k\sum_{a=1}^{M}\hat a_a\hat\rho\hat a_a^\dagger,
 \qquad \mathcal I_a^{\mathrm B}(\hat\rho)&=\frac1k\hat a_a\hat\rho\hat a_a^\dagger.
 \label{eq:boson-channel}
\end{align}
The channel is trace preserving because $\sum_a\hat a_a^\dagger\hat a_a=k\hat I$ on the parent space. Unitary mixing of the annihilation operators yields
\begin{align}
    \mathcal D_k^{\mathrm B}\!\left(\hat\Gamma_k(U)\hat\rho\hat\Gamma_k(U)^\dagger\right)  =\hat\Gamma_{k-1}(U)\mathcal D_k^{\mathrm B}(\hat\rho) \hat\Gamma_{k-1}(U)^\dagger.
 \label{eq:boson-covariance}
\end{align}
Write $|\mathbf t\rangle$ for a normalized occupation vector with $\sum_j t_j=k$, and let $\mathbf e_j$ be the $j$th unit vector. For a child occupation $\mathbf s$, the adjoint channel satisfies
\begin{align}
 (\mathcal D_k^{\mathrm B})^\dagger(|\mathbf s\rangle\langle\mathbf s|) =\frac1k\sum_{j=1}^{M}(s_j+1)  |\mathbf s+\mathbf e_j\rangle\langle\mathbf s+\mathbf e_j|.
 \label{eq:boson-adjoint}
\end{align}
Thus, its classical action removes one uniformly chosen particle, rather than choosing an occupied mode uniformly:
\begin{align}
 K_k^{\mathrm B}(\mathbf s\mid\mathbf t)
 =\frac1k\sum_{j:t_j>0}t_j\,
  \mathbf 1_{\{\mathbf s=\mathbf t-\mathbf e_j\}}.
 \label{eq:boson-kernel}
\end{align}
For a parent output law $p$, Thm.~\ref{thm:instrument-recovery} therefore gives
\begin{align}
 q(\mathbf s)=\frac1k\sum_{j=1}^{M}(s_j+1)p(\mathbf s+\mathbf e_j),
 \qquad
 R_p(\mathbf s+\mathbf e_j\mid\mathbf s)
 =\frac{(s_j+1)p(\mathbf s+\mathbf e_j)}{kq(\mathbf s)}.
 \label{eq:boson-recovery}
\end{align}
We emphasize that the factors $s_j+1$ are required even for a collision-free input, since collisions are possible in the output distribution.

To recover the original algorithm, let $S\subseteq[M]$ contain $k$ distinct occupied input modes and write $|1_S\rangle$ for the corresponding input state. The loss instrument chooses $a\in S$ with probability $1/k$ and leaves $|1_{S\setminus\{a\}}\rangle$. Hence a recursive call first chooses a uniformly random deleted input mode, samples the smaller interferometer problem, and then applies Eq.~\eqref{eq:boson-recovery}. Averaging over the deleted mode produces precisely the law $q$ in Eq.~\eqref{eq:boson-recovery}; averaging the return step then gives $p$. 

In occupation notation, the parent probabilities are
\begin{align}
 p_S(\mathbf t)=\left|\langle\mathbf t|\hat\Gamma_k(U)|1_S\rangle\right|^2 =\frac{|\per U[\mathbf t,S]|^2}{\prod_{i=1}^{M}t_i!},
 \label{eq:boson-permanent-law}
\end{align}
where $U[\mathbf t,S]$ uses the columns in $S$ and repeats row $i$ exactly $t_i$ times. For a fixed child occupation $\mathbf s$, the recovery weight becomes
\begin{align}
 (s_j+1)p_S(\mathbf s+\mathbf e_j)
 =\frac{|\per U[\mathbf s+\mathbf e_j,S]|^2}
 {\prod_{i=1}^{M}s_i!}.
 \label{eq:boson-factorial-cancellation}
\end{align}
The common factorial cancels when normalizing over $j$. Thus the bosonic occupation factors reproduce, rather than modify, the permanent-square weights of Clifford and Clifford.

More explicitly, for $N$ distinct occupied input modes choose a uniformly random ordering $(\alpha_1,\ldots,\alpha_N)$ and set $S_k=\{\alpha_1,\ldots,\alpha_k\}$. This is the reverse of the deletion order. Starting with the empty output list, append $r_k=j$ with probability proportional to
\begin{align}
 \left|\per U[(r_1,\ldots,r_{k-1},j),S_k]\right|^2,
 \qquad j\in[M].
 \label{eq:boson-cc-weights}
\end{align}
At $k=1$ these weights are $|U_{j\alpha_1}|^2$. Converting the final list into occupations gives the desired $N$-boson sample. The correctness statement averages over the auxiliary input ordering; it need not hold for each fixed ordering. For a fixed parent set $S_k$, averaging over its uniformly chosen deleted element and the preceding recursive choices is exactly the instrument marginal used above.

For completeness, let $B=U[(r_1,\ldots,r_{k-1}),S_k]$, and denote by $B_{\hat\ell}$ the square matrix obtained by deleting its $\ell$th column. Expansion along the appended row yields
\begin{align}
 \per U[(r_1,\ldots,r_{k-1},j),S_k]
 =\sum_{\ell=1}^{k}U_{j\alpha_\ell}\per B_{\hat\ell}.
 \label{eq:boson-shared-minors}
\end{align}
The shared-minor method of Ref.~\cite{clifford2018classical} evaluates these $k$ permanents jointly in $O(k2^k)$ arithmetic operations and all $M$ return weights in a further $O(Mk)$ operations. Summing over $k$ gives the original $O(N2^N+MN^2)$ time and $O(M)$ additional-space bounds, excluding storage of the interferometer. 

\section{Shared probability calculations for exact sampling}
\label{s1:sec:s1-sampling}

This appendix develops the shared probability calculations used by the exact samplers: determinant sums for passive recovery in App.~\ref{s1:sec:passive}, and Pfaffian sums for active recovery and individual Born probabilities in App.~\ref{s1:sec:active}. We prove their correctness, arithmetic cost, and storage bounds, including singular cases. The complete worst-case active sampler is given separately in App.~\ref{s1:sec:upgrade}.

We use the input block $\ket{\phi_4}:=(\ket{1100}+\ket{0011})/\sqrt2$ on $M\ge4n$ modes. An intact block is one still in $\ket{\phi_4}$; a resolved mode has a specified Fock occupation. At a given stage, let $g$ be the number of intact blocks and $F$ the ordered list of occupied resolved modes. After absorbing any local Gaussian preparation into the circuit, the input has the form
\begin{align}
    \ket\psi
    &=2^{-g/2}\sum_{y\in\{0,1\}^g}\ket{J_y},
    \qquad
    J_y:=(F,X_1(y_1),\ldots,X_g(y_g)),
    \qquad
    k:=|F|+2g.
    \label{s1:eq:residual-input}
\end{align}
Here, $X_b(0)$ and $X_b(1)$ are the two occupied pairs of intact block $b$, and $k$ is the input particle number. Lists use the displayed order, with each pair internally ordered. Since moving a two-particle pair past other blocks contributes no choice-dependent sign, the branch coefficients have the same phase; we suppress this common phase. The initial input is $\ket{\Psi_n}:=\ket{\phi_4}^{\otimes n}\otimes\ket0^{\otimes(M-4n)}$.

All costs below count exact arithmetic, conjugation, comparisons, zero tests, nonnegative square roots, and sampling from finite represented distributions. The calculations include singular matrices.

\subsection{Shared determinant calculations for passive sampling}
\label{s1:sec:passive}

In Sec.~\ref{sec:passive-algorithm}, each recovery step adds one occupied output mode to the set returned by the recursive call. The new mode is selected using the probabilities of the candidate outcomes under the saved parent input, as specified in Eq.~\eqref{eq:sampling-passive-weights}. The runtime bound in the main text relies on computing all these probabilities together. Here, we construct the shared determinant calculation used for this purpose and prove its correctness, arithmetic cost, and working-memory bound, including when the matrices involved are rank deficient.

We use the convention $\hat U^\dagger\hat c_j\hat U=\sum_{a=1}^M U_{ja}\hat c_a$ and $\hat U|0\rangle=|0\rangle$, with $U\in U(M)$. For an input state $|\psi\rangle$, define
\begin{align}
    p_\psi(T)
    &:=|\langle T|\hat U|\psi\rangle|^2,
\end{align}
where $T$ is an occupied output set and its modes are listed in increasing order.

\begin{theorem}[Probability queries for passive sampling]
\label{s1:thm:passive-probability-query}
Let $|\psi\rangle$ be an input of the form in Eq.~\eqref{s1:eq:residual-input}, with $g$ intact blocks and $k=|F|+2g\ge1$ particles on $M$ modes. Given $U\in U(M)$ and an occupied output set $S\subseteq[M]$ with $|S|=k-1$, a classical algorithm can compute all probabilities
\begin{align}
    p_\psi(S\cup\{j\}),
    \qquad j\notin S,
\end{align}
in $O(k^3+Mk+2^g)$ arithmetic operations and $O(k^2+M)$ extra working memory, excluding the supplied matrix $U$. The algorithm includes rank-deficient cases.
\end{theorem}

The proof is organized as follows. In Sec.~\ref{s1:sec:passive-amplitudes}, we express all candidate probabilities through one cofactor vector and identify the sum that must be computed. In Sec.~\ref{s1:app:passive}, we eliminate the columns shared by every input occupation list, reducing the remaining calculation to either a cofactor-vector sum or a scalar determinant sum. Section~\ref{s1:app:passive-recursions} describes a recursive subroutine for evaluating these sums, and Sec.~\ref{s1:app:passive-correctness} proves its correctness. Finally, Sec.~\ref{s1:sec:passive-cost} combines these results and analyzes their arithmetic cost and working memory to prove Thm.~\ref{s1:thm:passive-probability-query}. The accumulation of these costs along the sampling path is handled in Sec.~\ref{sec:passive-algorithm}.

\subsubsection{Obtaining the recovery weights from a cofactor vector}
\label{s1:sec:passive-amplitudes}

The main text obtains all candidate amplitudes by multiplying $U$ by a single vector that depends on the saved input and the returned set $S$. We first specify this vector and the occupation lists used to construct it. This identifies the sum that the subsequent elimination procedure must evaluate.

For the saved parent, $k=|F|+2g$ and $|S|=k-1$. Define the ordered list
\begin{align}
    Q&:=\bigl(X_1(0),X_1(1),\ldots,X_g(0),X_g(1)\bigr).
\end{align}
For each $y\in\{0,1\}^g$, let $Q_y:=(X_1(y_1),\ldots,X_g(y_g))$ contain the selected pair from each intact block, and set $D:=(F,Q)$. All lists are formed by concatenation, with each pair internally ordered. Then, $J_y=(F,Q_y)$ and $|D|=|F|+4g\le2k$. 

Let $B_y:=U_{S,J_y}$, with the rows indexed by $S$ in increasing order. For any $(m-1)\times m$ matrix $B$, define its cofactor vector by
\begin{align}
    z_a(B)&:=(-1)^{m+a}\det B^{\setminus a},
    \qquad a=1,\ldots,m,
    \label{s1:eq:passive-cofactors}
\end{align}
where $B^{\setminus a}$ is obtained by deleting column $a$. We use the convention that the determinant of the empty matrix is one. Let $\iota_y$ place the entries indexed by $J_y$ into the corresponding coordinates of $D$ and insert zeros elsewhere.

The cofactor representation derived in Eq.~\eqref{eq:sampling-shared-cofactor} gives
\begin{align}
    h&:=2^{-g/2}\sum_{y\in\{0,1\}^g}\iota_y z(B_y),
    \qquad
    \mathcal A=U_{:,D}h,
    \label{s1:eq:passive-h}
\end{align}
where $\mathcal A=(\mathcal A_1,\ldots,\mathcal A_M)^T$. Here, $\mathcal A_j$ uses the output-row order obtained by appending $j$ after $S$. Reordering these rows changes only a sign independent of $y$, so $|\mathcal A_j|^2=p_\psi(S\cup\{j\})$ for $j\notin S$. For $j\in S$, the repeated row gives $\mathcal A_j=0$. This algebraic recursion sums both choices for each intact block; the outer sampling recursion selects only one child state.

\subsubsection{Eliminating the common columns}
\label{s1:app:passive}

The columns indexed by $F$ occur in every branch matrix $B_y$. We can therefore eliminate them once and apply the resulting reduction to all terms defining $h$. Depending on the rank of these columns, the remaining calculation is either a cofactor-vector sum, a scalar determinant sum, or zero.

Write $L:=U_{S,F}$ and $r_F:=\operatorname{rank}L$. Order the independent columns of $F$ first and set $\delta:=|F|-r_F$. For the calculation below, use this reordered list in $J_y$ and $D$, recording its permutation and sign. One common row elimination gives
\begin{align}
    E[L\ \ U_{S,Q}]
    &=\begin{pmatrix}
        I_{r_F}&A&V\\
        0&0&W
    \end{pmatrix},
    \qquad
    W\in\mathbb C^{(\delta+2g-1)\times4g}.
    \label{s1:eq:passive-common}
\end{align}
Here, $E$ is an invertible row-transformation matrix, and the first two column groups have widths $r_F$ and $\delta$. The fixed column reordering changes all branch amplitudes by one common sign. The subscript $y$ on $V_y,W_y$ selects the columns indexed by $Q_y$; $\iota_y^Q$ embeds a vector on those columns into all $4g$ coordinates of $Q$, with zeros elsewhere.

For $g\ge1$, define the cofactor sum for $(2g-1)\times4g$ matrices and the determinant sum for $2g\times4g$ matrices by
\begin{align}
    \mathcal C_g(W)&:=\sum_{y\in\{0,1\}^g}\iota_y^Q z(W_y),
    \qquad
    \mathcal D_g(W):=\sum_{y\in\{0,1\}^g}\det W_y.
    \label{s1:eq:passive-CD}
\end{align}
We also set $\mathcal D_0:=1$. The following lemma specifies which sum is required and how its result determines $h$.

\begin{lemma}
\label{s1:lem:passive-common-reduction}
With the row reduction in Eq.~\eqref{s1:eq:passive-common}, we have $h=0$ if $\delta\ge2$. In the remaining cases,
\begin{align}
    h&=
    \begin{cases}
        \displaystyle
        \frac{2^{-g/2}}{\det E} \binom{-V\mathcal C_g(W)}{\mathcal C_g(W)},
        &\delta=0,
        \\[8pt]
        \displaystyle
        \frac{2^{-g/2}}{\det E} \begin{pmatrix}
            -a\\1\\0_{4g}
        \end{pmatrix}\mathcal D_g(W),
        &\delta=1.
    \end{cases}
    \label{s1:eq:passive-lift}
\end{align}
Here, $A=a$ is one column when $\delta=1$. These formulas use the reordered coordinates; the original column order and its permutation sign are restored afterwards.
\end{lemma}

\begin{proof}
Split $F=(F_{\rm ind},F_{\rm dep})$ into $r_F$ independent columns and $\delta$ dependent columns. For any row $v$ and branch $y$, multiply the upper $k-1$ rows by $E$, then subtract $v_{F_{\rm ind}}$ times its identity rows from the last row. Expanding along that identity block yields
\begin{align}
    \det(E)\det\begin{pmatrix}
        U_{S,J_y}\\v_{J_y}
    \end{pmatrix}
    &=
    \det\begin{pmatrix}
        0_{(\delta+2g-1)\times\delta}&W_y\\
        v_{F_{\rm dep}}-v_{F_{\rm ind}}A&
        v_{Q_y}-v_{F_{\rm ind}}V_y
    \end{pmatrix}.
    \label{s1:eq:passive-reduction-proof}
\end{align}
If $\delta\ge2$, the first $\delta$ columns on the right-hand side are supported only in the last row, so the determinant vanishes for every $v$ and hence $h=0$. For $\delta=0$ or $\delta=1$, expand along the last row or the first column, respectively. The latter expansion has positive sign because $2g$ is even. Comparing coefficients of the arbitrary row $v$ and summing over $y$ with the factor $2^{-g/2}$ gives Eq.~\eqref{s1:eq:passive-lift}, without any rank assumption on $W_y$.
\end{proof}

For $g=0$, the matrix $L$ has $k-1$ rows and $k$ columns, so $\delta\ge1$. The case $\delta=1$ uses the empty determinant $\mathcal D_0=1$, while $\delta\ge2$ gives zero. Thus, no cofactor sum with a negative row count is needed. The permutation and sign restoration used above follows the identity in Eq.~\eqref{s1:eq:cofactor-permutation} below.

\subsubsection{Subroutine: Computing the cofactor and determinant sums}
\label{s1:app:passive-recursions}

The preceding reduction expresses $h$ through either the cofactor sum $\mathcal C_g(W)$ or the determinant sum $\mathcal D_g(W)$. Both are needed in the recursion: eliminating a rank-one selected pair in a cofactor call leaves a determinant problem on the remaining blocks. We therefore describe the two updates together. Its input is a matrix $W$ with the dimensions in Eq.~\eqref{s1:eq:passive-CD}, together with its columns grouped into two alternative pairs for each intact block. Its output is the vector $\mathcal C_g(W)$ or the scalar $\mathcal D_g(W)$.

At each call, we split the sum according to the pair selected from the first remaining block. Eliminating that pair reduces the calculation to a sum over one fewer block. The rank of the selected pair determines which reduction applies and whether the contribution vanishes. We state the update rules below and prove their correctness in Sec.~\ref{s1:app:passive-correctness}.

\begin{enumerate}
    \item \textbf{Handle the base case.}
    A determinant call with $g=0$ returns $\mathcal D_0=1$. Cofactor calls occur only for $g\ge1$. For $g\ge1$, consider the two choices for the first remaining block.

    \item \textbf{Eliminate the selected pair.}
    For one choice, let $Y$ be the two selected columns and let $X$ contain all columns from the remaining blocks. Discard the unselected alternative pair for this branch. The cofactor signs use the number of columns in each selected branch matrix at the current recursive level.

    For a rank-two pair, a row transformation gives
    \begin{align}
        E_Y[Y\ \ X]
        &=\begin{pmatrix}
            I_2&V_Y\\
            0&W'_Y
        \end{pmatrix}.
        \label{s1:eq:passive-pair-two}
    \end{align}
    In a determinant call, $W'_Y$ has $2(g-1)$ rows; in a cofactor call, it has $2g-3$ rows. It has $4(g-1)$ columns in both cases.

    In a rank-one cofactor call, place an independent selected column first and reduce to
    \begin{align}
        E_Y[Y\ \ X]
        &=\begin{pmatrix}
            1&a_Y&V_Y\\
            0&0&W'_Y
        \end{pmatrix},
        \qquad
        W'_Y\in\mathbb C^{2(g-1)\times4(g-1)}.
        \label{s1:eq:passive-pair-one}
    \end{align}
    The row transformations are implemented by row swaps, scalings, and elimination on the pivot columns, while recording their determinant factors. A determinant branch with a rank-deficient selected pair contributes zero. A cofactor branch with a rank-zero selected pair also contributes zero.

    \item \textbf{Recurse, restore, and add.}
    Let $Y_0,Y_1$ denote the two alternative selected-column matrices. For the determinant sum, recurse only on rank-two choices and return
    \begin{align}
        \mathcal D_g(W)
        &=
        \sum_{\substack{Y\in\{Y_0,Y_1\}\\
                        \operatorname{rank}Y=2}}
        \frac{\mathcal D_{g-1}(W'_Y)}{\det E_Y}.
        \label{s1:eq:passive-D-rec}
    \end{align}

    For the cofactor sum, let $\jmath_Y$ embed the selected pair and remaining-block coordinates into the parent coordinate order, placing zeros at the discarded alternative pair. Return the sum of $\jmath_Yt_Y$ over the two choices, where
    \begin{align}
    t_Y&:=
    \begin{cases}
        \displaystyle
        \frac{1}{\det E_Y}
        \binom{-V_Y\mathcal C_{g-1}(W'_Y)}
              {\mathcal C_{g-1}(W'_Y)},
        &\operatorname{rank}Y=2,
        \\[8pt]
        \displaystyle
        \frac{\mathcal D_{g-1}(W'_Y)}{\det E_Y}
        \begin{pmatrix}
            -a_Y\\1\\0_{4(g-1)}
        \end{pmatrix},
        &\operatorname{rank}Y=1.
    \end{cases}
    \label{s1:eq:passive-C-rec}
\end{align}
    At rank zero, set $t_Y=0$.

    If the independent rank-one column was second, restore the selected-column swap and its sign before embedding and adding the branch result. For any column permutation matrix $\Pi$, the required restoration follows from
    \begin{align}
        z(B\Pi)
        &=\det(\Pi)\Pi^Tz(B),
        \qquad
        z(B)=\det(\Pi)\Pi z(B\Pi).
        \label{s1:eq:cofactor-permutation}
    \end{align}
    This identity is justified in the correctness proof below. Row-operation signs and scalings are already included in $\det E_Y$, and all divisions use pivots explicitly found to be nonzero.
\end{enumerate}

For $g=1$, the cofactor matrix has one row, so only rank-zero or rank-one branches occur. The latter terminate at $\mathcal D_0=1$. Thus, the recursion never requires a cofactor call with a negative row count.

\subsubsection{Correctness of the subroutine}
\label{s1:app:passive-correctness}

We now show that each branch update preserves its contribution to the required sum. A rank-two selected pair reduces either type of sum to the same type with one fewer block. A rank-one pair contributes only to the cofactor sum and reduces it to a determinant sum. These identities also explain why the remaining rank-deficient cases vanish.

\begin{lemma}
\label{s1:lem:passive-recursion-correctness}
The subroutine in Sec.~\ref{s1:app:passive-recursions} computes the sums $\mathcal C_g(W)$ and $\mathcal D_g(W)$ defined in Eq.~\eqref{s1:eq:passive-CD}, including when the selected branch matrices are rank deficient.
\end{lemma}

\begin{proof}
For a determinant call, expansion along the identity block in Eq.~\eqref{s1:eq:passive-pair-two} gives the reduced determinant divided by $\det E_Y$. If the selected pair has rank less than two, every completed square matrix has dependent columns and contributes zero. Summing over the remaining choices proves Eq.~\eqref{s1:eq:passive-D-rec}.

For a cofactor call, apply the algebraic reduction in Eq.~\eqref{s1:eq:passive-reduction-proof} with the selected pair as the common columns and $g-1$ remaining blocks. Here, $\delta=2-\operatorname{rank}Y$. Comparing coefficients of the appended row and summing over the remaining choices gives the two formulas in Eq.~\eqref{s1:eq:passive-C-rec} when $\delta=0$ or $\delta=1$, and the zero contribution when $\delta=2$.

Equation~\eqref{s1:eq:cofactor-permutation} follows by appending an arbitrary row to $B$, applying the column permutation, and comparing coefficients in the determinant expansion. Thus, restoring each permutation and its sign before embedding preserves the branch contribution. Every nonzero recursive branch has one fewer block, so induction from the empty determinant $\mathcal D_0=1$ proves both recursions.
\end{proof}

\subsubsection{Proof of the probability-query theorem}
\label{s1:sec:passive-cost}
\label{s1:app:passive-storage}

We now combine the cofactor representation, the common-column reduction, and the recursive subroutine to prove Thm.~\ref{s1:thm:passive-probability-query}. We also establish the time and memory bounds for one probability calculation; the main text analyzes how these costs accumulate along the sampling path.

\begin{proof}[Proof of Thm.~\ref{s1:thm:passive-probability-query}]
Lem.~\ref{s1:lem:passive-common-reduction} reduces the calculation of $h$ to either $\mathcal C_g(W)$ or $\mathcal D_g(W)$, or returns $h=0$. Lem.~\ref{s1:lem:passive-recursion-correctness} computes the required sum, including rank-deficient cases. After restoring the common columns and their ordering, Eq.~\eqref{s1:eq:passive-h} gives all candidate probabilities as $|\mathcal A_j|^2$, where $\mathcal A=U_{:,D}h$.

Each algebraic node makes at most two calls with one fewer block. Its row updates, vector restoration, and embedding use $O(g^2)$ operations. Hence, a common upper bound for the costs of the two recursive sums satisfies $T_g \le2T_{g-1}+Cg^2$ with $T_0 = O(1)$, so that
\begin{align}
    T_g &\le2^g\left(C_0+C\sum_{j=1}^g\frac{j^2}{2^j}\right) =O(2^g).
    \label{s1:eq:binary-cost}
\end{align}
Initial common-column elimination costs $O(k^3)$, restoring $h$ costs $O(k^2)$, and multiplying by $U_{:,D}$ costs $O(Mk)$ because $|D|\le2k$. Taking the squared absolute values costs $O(M)$, which is absorbed in this bound. Thus, all candidate probabilities are obtained in $O(k^3+Mk+2^g)$ arithmetic operations.

For storage, a pair reduction retains at most two pivot rows, their multipliers and scalars, and the permutation data. This takes $O(g)$ words. We process the two branches successively, work in one active table, keep discarded-pair columns inactive, and undo the updates on return. A first-branch result retained while evaluating the second branch also uses $O(g)$ words. The stack therefore uses $\sum_{j=1}^gO(j)=O(g^2)$ space, in addition to an $O(k^2)$ table and an $O(M)$ candidate vector. The initial common-column reduction has the same $O(k^2)$ storage scale. Since $g\le k/2$, the total extra working memory is $O(k^2+M)$, excluding the supplied circuit matrix $U$.
\end{proof}

This establishes the per-query bound used in Eq.~\eqref{eq:sampling-passive-query}. Section~\ref{sec:passive-algorithm} sums these costs along the sampling path using Eq.~\eqref{eq:sampling-passive-path} and obtains the total runtime in Thm.~\ref{thm:sampling-passive}. Since $k\le2n$, the extra working-memory bound above is $O(n^2+M)$ for the inputs considered there.

\subsection{Shared Pfaffian calculations for active sampling}\label{s1:sec:active}
In Sec.~\ref{Sec:Active sampling}, both active sampling algorithms use a recovery step to obtain a sample from the parent output distribution. At each recovery step, we choose which bit of the current occupation string to flip, using the probabilities of the $M$ candidate strings under the saved parent input as weights. The runtime bounds stated in the main text rely on computing all these weights together in $O(M^3+2^g)$ time, where $g$ is the number of intact input blocks. In this subsection, we describe the shared Pfaffian calculation that achieves this bound and prove its correctness and cost per call; the main text analyzes how often this calculation is invoked by each sampling algorithm.

More precisely, let $|\psi\rangle$ denote the saved parent input and $\hat O$ the Gaussian unitary with Majorana matrix $O$. For the current occupation string $x$, the recovery step requires the weights $p_\psi^O(x\oplus e_j)$ for $j\in[M]$, where
\begin{align}
    p_\psi^O(z) &:= |\langle z|\hat O|\psi\rangle|^2,
    \qquad z\in\{0,1\}^M.
\end{align}
Here, $e_j$ denotes the bit string whose only nonzero entry is at position $j$.

\begin{theorem}[Probability queries for active sampling]
\label{s1:thm:active-probability-query}
Let $|\psi\rangle$ be an input of the form in Eq.~\eqref{s1:eq:residual-input}, with $g$ intact blocks on $M$ modes. Given $O\in\operatorname{SO}(2M)$ and an occupation string $x\in\{0,1\}^M$, a classical algorithm can compute either of the following:
\begin{enumerate}
    \item all $M$ neighboring probabilities $p_\psi^O(x\oplus e_j)$, up to a common positive factor;
    \item the absolute Born probability $p_\psi^O(x)$.
\end{enumerate}
Each task requires $O(M^3+2^g)$ arithmetic operations and $O(M^2)$ working memory. The algorithm includes singular cases and does not assume a nonzero vacuum overlap.
\end{theorem}

The proof is organized as follows. In App.~\ref{s1:sec:active-amplitudes}, we construct a common Pfaffian representation of Gaussian occupation amplitudes. Using this representation, App.~\ref{s1:sec:pf-vector} reduces the calculation of all neighboring probabilities to a single sum of Pfaffian vectors. We describe a recursive subroutine for computing this sum in App.~\ref{s1:sec:pf-routine} and prove its correctness in App.~\ref{s1:app:pf-identities}. Next, App.~\ref{s1:app:scalar} shows how the same subroutine computes an individual Born probability with its absolute normalization. Finally, in App.~\ref{s1:sec:active-cost}, we combine these constructions and analyze their arithmetic cost and storage requirements to prove Thm.~\ref{s1:thm:active-probability-query}.

\subsubsection{Expressing Gaussian amplitudes as Pfaffians}\label{s1:sec:active-amplitudes}
Fix a saved parent input $|\psi\rangle$ of the form in Eq.~\eqref{s1:eq:residual-input} and a returned occupation string $x$, and let $|\zeta\rangle:=\hat O^\dagger|x\rangle$. Expanding the intact input blocks yields the ladder-transition vectors of Eq.~\eqref{eq:sampling-ladder-vectors} as
\begin{align}
 a^-_a=2^{-g/2}\sum_{y\in\{0,1\}^g}\langle\zeta|\hat c_a|J_y\rangle,\qquad
 a^+_a=2^{-g/2}\sum_{y\in\{0,1\}^g}\langle\zeta|\hat c_a^\dagger|J_y\rangle.
 \label{s1:eq:ladder-sums}
\end{align}
Each vector $\hat c_a|J_y\rangle$ or $\hat c_a^\dagger|J_y\rangle$ is either zero or another occupation-basis state up to a sign. Therefore, evaluating Eq.~\eqref{s1:eq:ladder-sums} requires occupation amplitudes $\langle\zeta|J\rangle$ for different occupation lists $J$. We first express these amplitudes using matrices that depend only on $|\zeta\rangle$. This common representation will allow us to share the calculation across the different choices of $J$.

For an ordered list $J=(j_1,\ldots,j_s)$ of distinct mode indices, write
\begin{align}
    |J\rangle
    &:=\hat c_{j_1}^\dagger\cdots\hat c_{j_s}^\dagger|0\rangle,
    \qquad |J|:=s.
\end{align}
For the matrices $H$ and $Z$ introduced below, $H_{:,J}$ denotes the submatrix of $H$ obtained by retaining all rows and the columns indexed by $J$, and $Z_{J,J}$ denotes the principal submatrix of $Z$ indexed by $J$. The selected columns of $H$ and the selected rows and columns of $Z$ follow the order specified by $J$. 

\begin{lemma}\label{s1:lem:gaussian-coefficients}
For $|\zeta\rangle=\hat O^\dagger|x\rangle$, there exist an integer $0\le d\le M$, a matrix $H\in\mathbb C^{d\times M}$, a skew-symmetric matrix $Z\in\mathbb C^{M\times M}$, and a scalar $\kappa\ne0$ such that
\begin{align}
    \langle\zeta|J\rangle=\kappa\,\pf C(J),
    \qquad C(J) := \begin{pmatrix}
        0_d&H_{:,J}\\
        -H_{:,J}^{T}&Z_{J,J}
    \end{pmatrix}
    \label{s1:eq:gaussian-coefficient}
\end{align}
whenever $d+|J|$ is even; otherwise, the amplitude is zero. Here, $0_d$ is the $d\times d$ zero matrix and $\pf$ denotes the Pfaffian. The same $H$, $Z$, and $\kappa$ apply to every $J$, and $H$ and $Z$ can be constructed from $O$ and $x$ in $O(M^3)$ operations. The construction also ensures $HH^\dagger=I_d$ and $ZH^\dagger=0$.
\end{lemma}

\begin{proof}
The Gaussian normal form~\cite{AvezBender2012} gives
\begin{align}
    |\zeta\rangle &=\lambda_\zeta\,\hat d_1^\dagger\cdots\hat d_d^\dagger \exp\left(\frac12\sum_{a,b=1}^{M}Z^{\rm ket}_{ab}\hat c_a^\dagger\hat c_b^\dagger\right)|0\rangle,
    \qquad \hat d_i^\dagger:=\sum_{a=1}^{M}D_{ai}\hat c_a^\dagger.
    \label{s1:eq:gaussian-normal-form}
\end{align}
Here, the columns of $D$ describe orthonormal always-occupied orbitals, $Z^{\rm ket}$ is a skew-symmetric pairing matrix supported on their orthogonal complement, and $\lambda_\zeta\ne0$. Thus, $D^\dagger D=I_d$ and $D^\dagger Z^{\rm ket}=0$.

Expanding the creation operators and comparing with the Pfaffian expansion gives Eq.~\eqref{s1:eq:gaussian-coefficient} with
\begin{align}
    H&=D^\dagger,
    \qquad Z=(Z^{\rm ket})^*,
    \qquad \kappa=\lambda_\zeta^*(-1)^{d(d-1)/2}.
\end{align}
The sign accounts for the ordering of the auxiliary coordinates. Coefficients with $d+|J|$ odd vanish because the remaining particles are created in pairs. These choices also give $HH^\dagger=I_d$ and $ZH^\dagger=0$.

To construct these matrices, use the operators $\hat O^\dagger\hat c_j\hat O$ for $x_j=0$ and $\hat O^\dagger\hat c_j^\dagger\hat O$ for $x_j=1$, all of which annihilate $|\zeta\rangle$. Their coefficients are obtained directly from $O$ and $x$. Let $A$ collect the coefficients multiplying the annihilation operators. Find an orthonormal basis $D$ of $\ker A$ and solve the pairing equations on its orthogonal complement to obtain $Z^{\rm ket}$. These operations construct $H$ and $Z$ in $O(M^3)$ time, including when $A$ is singular.
\end{proof}

\subsubsection{Obtaining neighboring probabilities from a vector sum}
\label{s1:sec:pf-vector}

In the preceding subsection, we expressed Gaussian occupation amplitudes as Pfaffians. We now use this representation to compute both ladder-transition vectors in Eq.~\eqref{s1:eq:ladder-sums} together. For a fixed occupation list, annihilating an occupied mode deletes one row and column, while expanding the Pfaffian enlarged by creation produces the same collection of minors. We collect these minors, with their corresponding signs, into a vector and show that summing these vectors in common coordinates determines both ladder-transition vectors and hence all neighboring probabilities up to a common positive factor.

We first consider the case where $x$ has parity opposite to the input particle number $k$. Since $|\zeta\rangle=\hat O^\dagger|x\rangle$ has the same parity as $x$, each matrix $C(J_y)$ has odd order $d+k$. Deleting one row and column therefore gives an even-order matrix whose Pfaffian is defined. If $x$ has the same parity as the parent input, all neighboring probabilities are zero.

For an odd skew-symmetric matrix $C$ of order $N$, define the vector
\begin{align}
    \qv_i(C)
    &:=(-1)^{i+1}\pf C^{\hat i},
    \qquad i=1,\ldots,N,
    \label{s1:eq:pf-vector}
\end{align}
where $C^{\hat i}$ is obtained by deleting row and column $i$ from $C$. We use the convention $\pf(0_{0\times0})=1$.

For $C_y:=C(J_y)$, the first $d$ entries of $\qv(C_y)$ correspond to the auxiliary coordinates, while the remaining entries correspond to the occupied modes listed in $J_y$. These physical mode labels depend on $y$, so we first place the vectors in a common coordinate order before adding them. Specifically, let $\iota_y$ retain the first $d$ entries, place each physical entry in the position corresponding to its mode label among the $M$ physical coordinates, and fill the unselected positions with zeros. Define
\begin{align}
    \binom{u}{v}
    &:=\sum_{y\in\{0,1\}^g}\iota_y\qv(C_y),
    \qquad
    u\in\mathbb C^d,\quad v\in\mathbb C^M.
    \label{s1:eq:grouped-vector}
\end{align}
The following lemma shows that this single vector sum determines both ladder-transition vectors.

\begin{lemma}
\label{s1:lem:ladder-readout}
Suppose that $x$ has parity opposite to the input particle number $k$. Then, the vectors $u$ and $v$ in Eq.~\eqref{s1:eq:grouped-vector} satisfy
\begin{align}
    a^- =2^{-g/2}\kappa(-1)^d v, 
    \qquad a^+=2^{-g/2}\kappa(-1)^k \left(H^{T}u+Z^{T}v\right).
    \label{s1:eq:grouped-readout}
\end{align}
\end{lemma}

\begin{proof}
Fix an occupation list $J=(j_1,\ldots,j_k)$. We use the matrix $C(J)$ before applying annihilation or creation.

First, consider annihilation. If $a=j_\ell$ is occupied, applying $\hat c_a$ removes that mode and contributes the sign $(-1)^{\ell-1}$. By Lem.~\ref{s1:lem:gaussian-coefficients},
\begin{align}
    \langle\zeta|\hat c_a|J\rangle=\kappa(-1)^{\ell-1} \pf C(J)^{\hat{d+\ell}}=\kappa(-1)^d\qv_{d+\ell}(C(J)).
    \label{s1:eq:annihilation-pf}
\end{align}
Here, $C(J)^{\hat{d+\ell}}$ denotes the submatrix obtained by deleting row and column $d+\ell$ from $C(J)$, corresponding to the occupied mode $j_\ell$. If $a\notin J$, the amplitude is zero. Thus, all annihilation amplitudes are obtained from the physical components of $\qv(C(J))$.

Next, consider creation. Appending mode $a$ to the occupation list adds the column $w_a(J):=\binom{H_{:,a}}{Z_{J,a}}$ and its corresponding row to $C(J)$. Expanding the Pfaffian along this column gives $\qv(C(J))^T w_a(J)$. Moving $\hat c_a^\dagger$ past the $k$ existing creation operators contributes the sign $(-1)^k$, so
\begin{align}
    \langle\zeta|\hat c_a^\dagger|J\rangle
    &=\kappa(-1)^k\qv(C(J))^T w_a(J).
    \label{s1:eq:creation-pf}
\end{align}
If $a\in J$, the appended coordinate repeats an existing one, making the Pfaffian zero, consistently with $\hat c_a^\dagger|J\rangle=0$.

Finally, sum these expressions over the occupation lists $J_y$. The embedding $\iota_y$ places each physical component at its corresponding mode label, so the annihilation amplitudes sum to $\kappa(-1)^d v$. For creation, the auxiliary components contribute $H^{T}u$, and the physical components contribute $Z^{T}v$. Including the input normalization $2^{-g/2}$ proves Eq.~\eqref{s1:eq:grouped-readout}.
\end{proof}

Once $u$ and $v$ are known, Eqs.~\eqref{eq:sampling-ladder-vectors} and~\eqref{eq:sampling-active-readout} determine all neighboring output probabilities up to a common positive factor. Both ladder-transition vectors contain the common factor $2^{-g/2}\kappa$, so all neighboring probabilities contain the factor $2^{-g}|\kappa|^2$. This factor cancels when the weights are normalized to choose which bit to flip. Thus, the recovery step does not require us to compute $\kappa$. When an absolute output probability is required, we determine $|\kappa|^2$ using the procedure in Sec.~\ref{s1:app:scalar}.

\subsubsection{Subroutine: Computing the Pfaffian vector sum}\label{s1:sec:pf-routine}

In the preceding subsection, we showed that all neighboring output probabilities can be obtained, up to a common positive factor, from the vector sum in Eq.~\eqref{s1:eq:grouped-vector}. We now describe a recursive algorithm for computing this sum. Its input consists of the matrices $H$ and $Z$ from Lem.~\ref{s1:lem:gaussian-coefficients}, the ordered list $F$ of occupied resolved modes, and the two ordered pairs $X_b(0),X_b(1)$ for each intact block $b=1,\ldots,g$. These data specify every matrix $C(J_y)$ without listing the $2^g$ block choices individually. The algorithm returns the vectors $u\in\mathbb C^d$ and $v\in\mathbb C^M$ defined in Eq.~\eqref{s1:eq:grouped-vector}: the sum of the corresponding Pfaffian vectors after placing their entries in the common auxiliary and physical coordinate order.

The branch matrices share their auxiliary coordinates and resolved occupations, so common elimination steps can be performed before branching on the next intact block. We then restore the eliminated coordinates and add the branch vectors in common coordinates. The identities used below are proved in Sec.~\ref{s1:app:pf-identities}.

\begin{enumerate}
    \item \textbf{Form the shared matrix.}
    
    Form $C([M])$ from Eq.~\eqref{s1:eq:gaussian-coefficient}, on the auxiliary coordinates and all physical modes. Retain the $d$ auxiliary coordinates and the occupied modes $F$ as common coordinates, together with both pairs $X_b(0),X_b(1)$ from each intact block. Order the retained coordinates as the auxiliary list, $F$, and then $X_b(0),X_b(1)$ for $b=1,\ldots,g$, with each pair internally ordered. Every $C_y$ is the principal submatrix selecting one pair from each block, in the order specified by $J_y$ in Eq.~\eqref{s1:eq:residual-input}. The expanded matrix has $d+|F|+4g$ coordinates, whereas one selected branch has $d+|F|+2g$.

    \item \textbf{Sweep through common coordinates.}
    
    A nonzero common pair can be eliminated using Pfaffian factorization. Move it to the front, recording the permutation, and write
    \begin{align}
        \mathcal A=\begin{pmatrix}P&B\\-B^{T}&T\end{pmatrix},\qquad
        P\coloneqq\begin{pmatrix}0&p\\-p&0\end{pmatrix},\qquad
        S\coloneqq T+B^{T}P^{-1}B,\qquad p\ne0.
        \label{s1:eq:pf-pivot}
    \end{align}
    For even matrices $\pf(\mathcal A)=p\,\pf(S)$. For our odd branch matrices, the corresponding vector lift appears in Step 4. Because the pair is common, this single Schur update reduces every branch at once. Store $P,B,p$ and the coordinate permutation, and continue with $S$.

    Repeat until the entire common submatrix is zero, not merely until one tested pair is zero. At that point $f$ common coordinates remain and the expanded residual has the form
    \begin{align}
        \mathcal A_{\rm rem}=\begin{pmatrix}0_f&B_{\rm rem}\\-B_{\rm rem}^{T}&T_{\rm rem}\end{pmatrix},\qquad
        B_{\rm rem}\in\C^{f\times4g},\quad T_{\rm rem}\in\C^{4g\times4g}.
        \label{s1:eq:pf-initial-residual}
    \end{align}
    None of the block choices has yet been made. A fixed $y$ retains only $2g$ of these $4g$ variable coordinates. The zero common block is retained: its coordinates can still couple to the variable ones.

    \item \textbf{Select the next pair and recurse on the remaining sum.}
    
    At each recursive node, first complete the common sweep so that the remaining common submatrix is zero. Let $f$ be its number of coordinates and $r$ the number of block choices still unresolved; initially $r=g$. If $y_1,\ldots,y_{b-1}$ have already been fixed, then $r=g-b+1$ and the remaining sum runs over $y_b,\ldots,y_g$. Each recursive call returns the sum of the Pfaffian vectors over its remaining block choices, embedded in the coordinate order at the start of that call.

    Each selected branch has order $f+2r$, which is odd, so $f$ is odd. First handle the terminal cases. If $f>2r+1$, return zero: after deleting any one coordinate, there are more surviving common coordinates than selected variable coordinates available to pair with them. If $r=0$, the odd residual is $0_f$; return $(1)$ when $f=1$ and zero when $f>1$. These are exactly its maximal-Pfaffian vectors. Saved pivots are restored on return.

    Otherwise split into $y_b=0$ and $y_b=1$, each retaining its selected pair and summing over $y_{b+1},\ldots,y_g$. Let the internal pair entries be $p_0,p_1$. If either is zero, choose $\eta$ avoiding $-p_0$ and $p_1$ and replace
    \begin{align}
        p_0\longmapsto p_0+\eta,\qquad p_1\longmapsto p_1-\eta.
        \label{s1:eq:balanced-pivot}
    \end{align}
    Any three distinct field elements suffice for this choice; $0,1,2$ are convenient, not special. For every fixed choice of the remaining blocks, these opposite changes preserve the sum of the two Pfaffian vectors after embedding them in common coordinates. This cancellation is proved in Lem.~\ref{s1:lem:balanced-pivots}. Both selected pairs are now valid, nonzero pivots. Eliminate the selected pair in each branch using Eq.~\eqref{s1:eq:pf-pivot}, then recurse with $r-1$.

    Since the Schur update can create a nonzero common block, each branch begins again with the common sweep. Its new common block has rank at most two, and at most one additional common-pair pivot restores it to zero. Thus common elimination is used both before any block choice and after each selected-pair elimination.

    \item \textbf{Lift, embed, and add.}
    
    For each selected branch, let $t$ be the Pfaffian vector after eliminating a pivot. The vector before elimination is recovered by
    \begin{align}
        t\longmapsto p\binom{-P^{-1}Bt}{t}.
        \label{s1:eq:pf-lift}
    \end{align}
    This identity is proved in Lem.~\ref{s1:lem:pf-vector-lift}. The lift restores the two eliminated coordinates and multiplies the remaining entries by the pivot scalar. Since the lift is linear, it also applies directly to the embedded sum of the residual vectors.

    At a binary split, let $\widetilde t_s$ denote the vector returned from branch $y_b=s$ after restoring its selected pair $X_b(s)$, for $s\in\{0,1\}$. The two vectors share the remaining coordinates but contain different selected pairs. Define $\operatorname{embed}_{X_b(s)}$ to place the entries of $\widetilde t_s$ in the parent coordinate order and insert zeros at the coordinates of the unselected pair $X_b(1-s)$. We then add the embedded vectors,
    \begin{align}
        t_{\rm parent}
        &:=
        \operatorname{embed}_{X_b(0)}(\widetilde t_0)
        +\operatorname{embed}_{X_b(1)}(\widetilde t_1).
        \label{s1:eq:pf-embed-add}
    \end{align}
    Undo preceding common pivots in reverse order. If a reordering was $C'=\Pi C\Pi^{T}$, restore its vector by $\qv(C)=\det(\Pi)\Pi^{T}\qv(C')$. Keep physical mode labels throughout. At the top, the returned vector is exactly $(u,v)$ in Eq.~\eqref{s1:eq:grouped-vector}, with zero entries on unused physical modes.
\end{enumerate}

\subsubsection{Correctness of the subroutine}
\label{s1:app:pf-identities}
\label{s1:app:balance}

We now prove the correctness of the subroutine described in Sec.~\ref{s1:sec:pf-routine}. We first prove the vector-lifting identity and show that it applies to sums of residual vectors. Next, we show that the opposite corrections used to make both candidate pivots nonzero preserve their embedded sum. Using these two identities, we prove by induction that the recursive procedure returns the required vector sum.

\begin{lemma}
\label{s1:lem:pf-vector-lift}
Let $C$ be an odd skew-symmetric matrix with the block form in Eq.~\eqref{s1:eq:pf-pivot}, where $p\ne0$. Then
\begin{align}
 \qv(C)=p\binom{-P^{-1}B\qv(S)}{\qv(S)},\qquad S=T+B^{T}P^{-1}B.
 \label{s1:eq:pf-vector-lift-proof}
\end{align}
No invertibility assumption is required for $S$.
\end{lemma}

\begin{proof}
The required identity follows from standard Pfaffian congruence and appended-column expansion~\cite{Wimmer2011}. For every invertible matrix $L$, these formulas give
\begin{align}
 \qv(LCL^{T})=\det(L)L^{-T}\qv(C).
 \label{s1:eq:pf-congruence}
\end{align}
Indeed, append a column $Lw$ on the transformed side and $w$ on the original side; Pfaffian congruence equates their expansions for every $w$. Choose
\begin{align}
 L\coloneqq\begin{pmatrix}I&0\\B^{T}P^{-1}&I\end{pmatrix},\qquad
 LCL^{T}=P\oplus S,\qquad \det L=1.
\end{align}
The vector of $P\oplus S$ is $(0,0,p\qv(S)^{T})^{T}$: deleting either pivot coordinate leaves its partner isolated. Equation~\eqref{s1:eq:pf-congruence} now proves Eq.~\eqref{s1:eq:pf-vector-lift-proof}.

For a grouped call, each branch cross block $B_y$ is a restriction of the same expanded $B$. If $t$ is the embedded sum of its residual vectors, then $\sum_{y\in\{0,1\}^g}B_y\qv(S_y)=Bt$. Linearity therefore proves the lift in Eq.~\eqref{s1:eq:pf-lift} for the entire sum. A coordinate permutation is the special case $L=\Pi$, with restoration $\det(\Pi)\Pi^{T}$.

\end{proof}

The next lemma supplies the additional identity needed by the shared recursion. Opposite corrections make both candidate pivots nonzero while preserving the embedded sum of their branch vectors. This lets us use the same elimination step even when an unmodified pivot vanishes.

\begin{lemma}
\label{s1:lem:balanced-pivots}
Fix the later block choices at a binary split, and let $C_0$ and $C_1$ be the two odd branch matrices. Their selected pairs occupy the same two positions in their respective coordinate lists, and deleting these pairs leaves the same ordered residual matrix. Let $C_0(\eta)$ and $C_1(-\eta)$ be obtained by changing the corresponding internal pair entries by $+\eta$ and $-\eta$, respectively, together with the changes required by skew symmetry. Then
\begin{align}
 \iota_0\qv(C_0(\eta))+\iota_1\qv(C_1(-\eta))
 &=\iota_0\qv(C_0)+\iota_1\qv(C_1),
 \label{s1:eq:balanced-sum-identity}
\end{align}
where $\iota_0$ and $\iota_1$ embed the branch vectors into their common expanded coordinate list.
\end{lemma}

\begin{proof}
Fix all later block choices and compare the alternatives of the next pair under Eq.~\eqref{s1:eq:balanced-pivot}. Since deleting a coordinate on the selected pair removes the corrected edge, the component at that coordinate does not change. At any shared coordinate, the Pfaffian is affine in that pair's internal entry. Deleting the two endpoints of that edge leaves the same ordered residual matrix for both alternatives. Since each alternative occupies the same two positions in its own branch list, the signed coefficients agree. The changes $+\eta$ and $-\eta$ therefore cancel in the embedded sum. Summing over later choices proves the claimed invariance.

\end{proof}

These identities allow us to prove that the recursive procedure returns the required sum, including when some branch matrices are singular.

\begin{lemma}
\label{s1:lem:grouped-pf-correctness}
Consider a complex skew-symmetric matrix whose coordinates consist of a fixed common list and two disjoint ordered pairs for each of $g$ blocks, as in Step 1 of Sec.~\ref{s1:sec:pf-routine}. Suppose that every branch obtained by selecting one pair from each block has odd order. The procedure computes the unweighted embedded sum of their maximal-Pfaffian vectors. In particular, for the matrices $C(J_y)$ it returns $(u,v)$ in Eq.~\eqref{s1:eq:grouped-vector}.
\end{lemma}

\begin{proof}
We prove correctness by induction on $(r,f)$ in lexicographic order. A common pivot lowers $f$ by two and has the exact linear lift just proved. Pruning and the $r=0$ base return the correct vectors by the pairing argument in Step 3. A binary split lowers $r$, balanced correction preserves its embedded sum, and lifting, embedding, and adding its two exact results reconstruct that sum. Restoring the initial common sweep gives the required embedded sum, and in particular Eq.~\eqref{s1:eq:grouped-vector}.
\end{proof}

After a selected-pair pivot, write $B_F$ for the columns of the pivot cross block corresponding to the surviving common coordinates. The new common block is $B_F^{T}P^{-1}B_F$, with rank at most two. Since a nonzero $2\times2$ principal pivot removes that rank by the Schur rank identity, one additional common pivot suffices. This bounds the additional common elimination at each node.

\subsubsection{Computing individual output probabilities}
\label{s1:app:scalar}

The preceding construction computes all neighboring output probabilities up to a common positive factor. We now use the same subroutine to compute an individual Born probability $p_\psi^O(x)$ with its absolute normalization. This requires a scalar sum of Gaussian occupation amplitudes. We show that the corresponding Pfaffian sum can be obtained as one component of the vector returned by our subroutine, after adding one isolated common coordinate. We then determine the normalization factor needed to recover the probability. These individual probabilities are used for the endpoint tags in Eq.~\eqref{eq:sampling-posterior-tags} and for the bounded fallback in Sec.~\ref{s1:app:fallback}.

\begin{corollary}
\label{s1:lem:scalar-born-query}
Suppose that $x$ has the same parity as the input particle number $k$, so that $d+k$ is even. The scalar probability is given by
\begin{align}
 S_\psi&:=\sum_{y\in\{0,1\}^g}\pf C(J_y),\qquad
 p_\psi^O(x)=2^{-g}|\kappa|^2|S_\psi|^2.
 \label{s1:eq:scalar-query}
\end{align}
The sum $S_\psi$ can be computed by the shared elimination after placing one isolated common coordinate first in each branch matrix. Moreover, $|\kappa|^2$ can be determined in $O(M^3)$ operations. At incompatible parity, $p_\psi^O(x)=0$.
\end{corollary}

\begin{proof}
The scalar probability formula follows by expanding $|\psi\rangle$ and applying Lem.~\ref{s1:lem:gaussian-coefficients}. For the scalar sum, observe that
\begin{align}
 \qv_\star(0_1\oplus C(J_y))=\pf C(J_y).
 \label{s1:eq:scalar-isolated-coordinate}
\end{align}
Here the isolated coordinate is placed \emph{first}. Every other component of that odd matrix's vector vanishes because the isolated coordinate remains after its deletion. Declare $\star$ common and run Sec.~\ref{s1:sec:pf-routine} unchanged; its $\star$ component is $S_\psi$. Thus one invocation of the shared elimination evaluates the scalar sum.

For absolute normalization, use the coefficient construction in Lem.~\ref{s1:lem:gaussian-coefficients}. The occupied orbitals are orthonormal, and the pairing matrix is supported on their orthogonal complement. In a paired orbital basis, the normal-form vector before multiplication by $\lambda_\zeta$ has squared norm $\prod_\nu(1+|z_\nu|^2)=\det(I+Z^\dagger Z)^{1/2}$, where $z_\nu$ are the pair-creation coefficients. Since $|\zeta\rangle$ is normalized and $|\kappa|=|\lambda_\zeta|$, we obtain
\begin{align}
    |\kappa|^2
    &=\det(I+Z^\dagger Z)^{-1/2}.
    \label{s1:eq:scalar-normalization}
\end{align}
The matrix $I+Z^\dagger Z$ is positive definite, and its determinant can be computed in $O(M^3)$ operations. This formula also applies when the vacuum overlap vanishes.
\end{proof}

When comparing probabilities for different input states at a fixed $x$ and a fixed circuit, the same Gaussian coefficient representation applies, so $|\kappa|^2$ is common and may be omitted. The input normalization factor $2^{-g}$ must still be retained, with $g$ equal to the number of intact blocks in the corresponding input. The bounded fallback compares different output strings, whose Gaussian bras generally have different $\kappa$, and therefore requires the absolute normalization above.

\subsubsection{Proof of the probability-query theorem}
\label{s1:sec:active-cost}

We now describe how to compute the two types of probabilities in Thm.~\ref{s1:thm:active-probability-query} using the preceding results. We then analyze the time and memory required for each calculation.

\begin{proof}[Proof of Thm.~\ref{s1:thm:active-probability-query}]
We first check the parity of $x$. If it equals the input parity, all neighboring probabilities are zero. If it differs from the input parity, the individual probability $p_\psi^O(x)$ is zero. For either calculation with compatible parity, construct the matrices $H$ and $Z$ describing $|\zeta\rangle=\hat O^\dagger|x\rangle$ using Lem.~\ref{s1:lem:gaussian-coefficients}.

For the neighboring probabilities, the subroutine in Sec.~\ref{s1:sec:pf-routine} returns $(u,v)$ by Lem.~\ref{s1:lem:grouped-pf-correctness}. Lemma~\ref{s1:lem:ladder-readout}, together with Eqs.~\eqref{eq:sampling-ladder-vectors} and~\eqref{eq:sampling-active-readout}, then gives all neighboring probabilities up to a common positive factor. For the individual probability, Corollary~\ref{s1:lem:scalar-born-query} gives the scalar Pfaffian sum using the same subroutine with one additional isolated coordinate and determines its absolute normalization.

We next bound the computation time. Constructing the Gaussian coefficient matrices and performing the initial elimination of common coordinate pairs cost $O(M^3)$ operations. After this elimination, any node that is not pruned satisfies $f\le2r+1$, where $f$ is the number of remaining common coordinates and $r$ is the number of unresolved blocks. Therefore, its expanded matrix has at most $f+4r\le6r+1$ coordinates.

Each of the two branches eliminates its selected pair and performs at most one additional common-pair elimination, requiring $O(r^2)$ operations. Restoring the eliminated coordinates, embedding the returned vectors, and adding them also take at most $O(r^2)$ operations. Thus, writing $T_r$ for the remaining recursive computation after the initial common elimination, we have
\begin{align}
    T_r &\le2T_{r-1}+O(r^2),
    \qquad T_0=O(1),
    \label{s1:eq:active-pf-cost}
\end{align}
which implies $T_g=O(2^g)$, as in Eq.~\eqref{s1:eq:binary-cost}.

Converting $(u,v)$ into the neighboring probabilities requires $O(M^2)$ additional operations. Hence, the total cost of computing all neighboring probabilities up to a common positive factor is $O(M^3+2^g)$. For the individual probability, the additional isolated coordinate is included among the common coordinates, so the same recursive cost bound applies. Determining $|\kappa|^2$ costs another $O(M^3)$ operations. Thus, the individual probability also requires $O(M^3+2^g)$ operations. Both calculations include singular matrices and zero vacuum overlap.

For memory, process the two branches successively, updating the matrix in place and undoing the updates after each branch. At a node with $r$ unresolved blocks, the saved pivot rows, scalars, permutations, and first branch's returned vector require $O(r)$ words. Summing over the recursive levels gives $O(g^2)$ space. The matrix and the records for the initial common elimination require $O(M^2)$ space. The Gaussian coefficient construction and absolute normalization also use $O(M^2)$ space. Therefore, either probability calculation uses $O(M^2)$ working memory.
\end{proof}

This proves the cost of one probability query. Section~\ref{sec:active-expected} and Appendix~\ref{s1:app:worst-proof} analyze the query counts and total runtimes of the expected-time and worst-case samplers, respectively.

The $O(M^2)$ memory bound applies to one probability query. In the expected-time sampler, storing one state-update record per downward step requires an additional $O(\tau)$ words, where $\tau$ is the length of the sampled instrument path.

\section{Worst-case active sampling: algorithm and proof}
\label{s1:sec:upgrade}

We give the complete implementation and proof of Thm.~\ref{thm:sampling-active-worst}. The algorithm groups a bounded batch of instrument paths, samples the output of a grouped endpoint through a Gaussian mixture, and restores an original path conditional on that output before applying the recovery updates. The calculations below establish the Gaussian replacement, implement the conditional path sampler, and bound the work over the recursion.

At the start of a call, the system has $M=4n$ modes, of which $R:=M-4g$ are resolved and $g$ four-mode blocks remain intact. Let $f_0$ be the resolved-mode occupation string and $\hat G$ the current Gaussian circuit. If $g=0$, use Gaussian occupation sampling. If $R>12g$, use the gate-by-gate sampler in App.~\ref{s1:app:fallback}. Otherwise set $\ell:=R+2$ and perform the following batch.
\begin{enumerate}
    \item \textbf{Generate and group the path.} Draw $\ell$ successive ladder operations using Eq.~\eqref{eq:active-instrument-probabilities}, saving the parent circuit and the intermediate input states. Each operation selects a mode uniformly; its loss or gain outcome is determined by the occupation on a resolved mode and is equiprobable on an intact mode. Classify visits relative to the blocks that were intact at the start of this batch. A Type A path never visits them; assign it an intact block $j$ uniformly. A Type B path visits them exactly twice, both times within the same block $j$. For either type, retain the group label $u=(f,j)$, where $f$ is the endpoint on the initially resolved modes. Every other path is ordinary: retain its full history, whose endpoint has fewer than $g$ intact blocks.
    \item \textbf{Recursively sample the endpoint.} For an ordinary path, use its endpoint with circuit $\hat G$. For a special group $u$, sample a Gaussian state $|\xi\rangle=\hat V_j|0000\rangle$ from the ensemble in App.~\ref{s1:app:ensemble}, using the group weights calculated in App.~\ref{s1:app:group-count}. Set the resolved modes to $f$, replace block $j$ by vacuum, and use the child circuit $\hat G\hat V_j$. In both cases the recursive input has fewer intact blocks.
    \item \textbf{Restore the path and recover.} Let $x$ be the output returned by the recursive call. For an ordinary path, keep its saved history. For a special group, discard the sampled Gaussian component and draw an original path conditioned on $(u,x)$ using App.~\ref{s1:app:bridge}. Reconstruct its intermediate input states, then apply Eq.~\eqref{eq:active-return-rule} in reverse order along the path using the saved parent circuit $\hat G$.
\end{enumerate}
The original Type A paths need not resolve a block: they are restored only after the recursive call, for use in recovery. Recursive progress is ensured by the Gaussian replacement of their grouped endpoint. We next justify this replacement and the path restoration, then combine their costs in App.~\ref{s1:app:worst-proof}.

\subsection{Computing the grouped endpoint state}
\label{s1:app:group-count}

For a special group, all endpoints agree outside the selected block. We calculate the averaged state of that block and show that it contains enough of the uniform even-parity contribution to admit the Gaussian mixture described in the main text. The choice $\ell=R+2$ and the nonterminal condition $R\leq12g$ enter through this weight bound.

Recall that a Type A path visits only initially resolved modes and is assigned a block label $j$ uniformly among the $g$ intact blocks. A Type B path visits initially intact modes exactly twice, both times in the same block $j$. In either case, let $f$ denote the endpoint on the $R$ initially resolved modes. Paths with the same $(f,j)$ form a special group. Visits are classified relative to the start of the batch, even though the first visit to $j$ already resolves that block. Every path outside these two classes resolves at least one initially intact block and retains its full history.

Let $B_t(f)$ count ordered sequences of $t$ resolved-mode flips taking $f_0$ to the specified string $f$. Define
\begin{align}
    b:=\frac{B_\ell(f)}{g},
    \qquad w :=16\binom{\ell}{2}B_{\ell-2}(f).
    \label{s1:eq:group-weights}
\end{align}
These are the Type A and Type B group weights after removing the common factor $M^{-\ell}$.

\begin{lemma}
\label{s1:lem:upgrade-group-weights}
The special group $(f,j)$ has probability $(b+w)/M^\ell$. If this probability is positive, its averaged endpoint state on block $j$ is
\begin{align}
    \hat\omega_{f,j}=\frac{b\hat\rho_\star+w\hat\Pi_e/8}{b+w},
    \qquad \hat\rho_\star:=|\phi_4\rangle\langle\phi_4|,
    \label{s1:eq:group-state}
\end{align}
where $\hat\Pi_e$ is the four-mode even-parity projector. All endpoint states in the group agree outside block $j$. Moreover, for $\ell=R+2$ and $R\le12g$, we have $w\ge4b/3$.
\end{lemma}

\begin{proof}
A Type A sequence has probability $M^{-\ell}$ and receives block label $j$ with probability $1/g$. Its state on block $j$ remains $\hat\rho_\star$. Summing over the $B_\ell(f)$ resolved-mode sequences gives total mass $b/M^\ell$.

For Type B, choose the two positions at which block $j$ is visited, the two local mode indices, and the resolved-mode sequence. There are respectively $\binom{\ell}{2}$, $4^2$, and $B_{\ell-2}(f)$ choices. Summing over the ladder outcomes therefore gives total mass $w/M^\ell$.

The first local visit offers the eight odd occupation strings uniformly: each of the four modes can be selected, and annihilation and creation each have probability $1/2$. The second visit flips a uniformly chosen local mode and yields the eight even-occupation strings uniformly. Equivalently, the $4\cdot2\cdot4=32$ local histories are equiprobable, and each even endpoint has four preimages. Thus, the averaged local endpoint state of Type B is $\hat\Pi_e/8$. Combining the two contributions proves the group probability and Eq.~\eqref{s1:eq:group-state}.

It remains to prove the weight bound. Suppose first that $R>0$. Every resolved-mode sequence of length $\ell=R+2$ repeats an index. Delete two occurrences of a repeated index, chosen by a fixed deterministic rule. The endpoint is unchanged because the two flips cancel. The shorter sequence, the removed index, and its two original positions uniquely reconstruct the original sequence. Consequently,
\begin{align}
    B_\ell(f)
    &\le R\binom{\ell}{2}B_{\ell-2}(f).
    \label{s1:eq:pair-deletion}
\end{align}
If $b>0$, this implies
\begin{align}
    \frac{w}{b}
    &\ge\frac{16g}{R}
    \ge\frac43.
\end{align}
If $b=0$, the claimed bound is immediate. For $R=0$, we have $\ell=2$, $B_2(f)=0$, and $B_0(f)=1$, so again $b=0$. This covers all cases.
\end{proof}

The lemma yields the convex decomposition used in the main text. For every special group of positive probability, write
\begin{align}
    \hat\omega_{f,j}=\eta\hat\sigma+(1-\eta)\frac{\hat\Pi_e}{8},
    \qquad \hat\sigma :=\frac37\hat\rho_\star+\frac47\frac{\hat\Pi_e}{8}, 
    \qquad \eta := \frac{7b}{3(b+w)}.
    \label{s1:eq:gaussian-mixture}
\end{align}
The bound $w\ge4b/3$ ensures $0\le\eta\le1$. 

\subsection{Gaussian mixture}
\label{s1:app:ensemble}

To implement the Gaussian replacement in Sec.~\ref{sec:active-worst}, we need to sample a pure Gaussian state whose average equals $\hat\omega_{f,j}$. Equation~\eqref{s1:eq:gaussian-mixture} expresses this state as a convex combination of $\hat\sigma$ and $\hat\Pi_e/8$. Since $\hat\Pi_e/8$ is already the uniform mixture of the eight even occupation states, it remains to construct an explicit Gaussian decomposition of $\hat\sigma$. The convex-Gaussian criterion of Ref.~\cite[Eqs.~(4)--(5)]{Oszmaniec2014} also explains the coefficient $3/7$: applied to $t\hat\rho_\star+(1-t)\hat\Pi_e/8$, it gives even-sector concurrence $\max\{0,(7t-3)/4\}$, which vanishes for $t\leq3/7$. Here the noise is confined to the even-parity sector. The construction below supplies a fixed ensemble that can be sampled directly at this threshold.

\begin{lemma}
\label{s1:lem:upgrade-gaussian-ensemble}
The state $\hat\sigma$ in Eq.~\eqref{s1:eq:gaussian-mixture} is the uniform mixture of fourteen normalized even pure Gaussian states. These states and their four-mode Gaussian preparations can be fixed in advance.
\end{lemma}

\begin{proof}
Let the following vectors form an orthonormal basis of the even four-mode subspace:
\begin{align}
    e_0 :=\frac{|0011\rangle+|1100\rangle}{\sqrt2},
    \qquad e_1:=\frac{i(|0011\rangle-|1100\rangle)}{\sqrt2}, \\
    e_2:=\frac{|0000\rangle-|1111\rangle}{\sqrt2},
    \qquad e_3 :=\frac{i(|0000\rangle+|1111\rangle)}{\sqrt2}, \\
    e_4 :=\frac{|0101\rangle-|1010\rangle}{\sqrt2},
    \qquad e_5 :=\frac{i(|0101\rangle+|1010\rangle)}{\sqrt2}, \\
    e_6 :=\frac{|0110\rangle+|1001\rangle}{\sqrt2}, 
    \qquad e_7:=\frac{i(|0110\rangle-|1001\rangle)}{\sqrt2}.
\end{align}
For an even vector $v$, write $v_z:=\langle z|v\rangle$ and define
\begin{align}
    Q(v)
    &:=v_{0000}v_{1111}-v_{0011}v_{1100}
       +v_{0101}v_{1010}-v_{0110}v_{1001}.
    \label{s1:eq:Q}
\end{align}

By the four-mode pure-Gaussian criterion~\cite[Eq.~(9)]{Oszmaniec2014}, a nonzero even vector $v$ is Gaussian if and only if $Q(v)=0$, expressed here in our occupation-basis convention. Define the fourteen normalized states $|\xi_j^\pm\rangle:=(e_0\pm i e_j)/\sqrt2$ for $j=1,\ldots,7$. Direct substitution into Eq.~\eqref{s1:eq:Q} gives $Q(\xi_j^\pm)=0$, so all these states are Gaussian. Averaging over the two signs cancels the cross terms. Since $e_0e_0^\dagger=\hat\rho_\star$ and $\sum_{a=0}^7e_ae_a^\dagger=\hat\Pi_e$, we obtain
\begin{align}
    \frac1{14}\sum_{j=1}^7\sum_{s\in\{-,+\}}
    |\xi_j^s\rangle\langle\xi_j^s|
    &=\frac{6\hat\rho_\star+\hat\Pi_e}{14}
    =\hat\sigma.
\end{align}
Each state is even Gaussian on four modes, so its preparation from vacuum by a Gaussian unitary can be fixed in advance.
\end{proof}

We can now sample the replacement state explicitly. With probability $\eta$, choose uniformly among the fourteen states in the lemma; otherwise, choose uniformly among the eight even occupation states. The average is $\hat\omega_{f,j}$. Replacing block $j$ by vacuum and absorbing the selected state's local Gaussian preparation into the circuit reduces the number of intact blocks by one. This implements the replacement used in Sec.~\ref{sec:active-worst}.

\subsection{Sampling a conditional instrument path}
\label{s1:app:bridge}

The Gaussian replacement preserves the group's output distribution, while the recovery updates require an original instrument history. We show how conditioning that history on the returned output restores the joint distribution needed for recovery, and then give an efficient sampler for this conditional distribution.

Let $\alpha_\lambda$ be the probability of an original batch path $\lambda$, let $\mu(u\mid\lambda)$ be its probability of assignment to group $u$, and let $|\chi_\lambda\rangle$ be its normalized endpoint. An ordinary group contains a single path. Under the saved parent circuit $\hat G$, the path's output law is
\begin{align}
    p_\lambda(x)&:=|\langle x|\hat G|\chi_\lambda\rangle|^2.
    \label{eq:sampling-path-law}
\end{align}
Since Gaussian replacement leaves each group's averaged endpoint state unchanged, choosing the group and sampling its output gives
\begin{align}
    \Pr[u,x]&=\sum_\lambda\alpha_\lambda\mu(u\mid\lambda)p_\lambda(x).
    \label{eq:sampling-group-output}
\end{align}
For a sampled pair $(u,x)$, restore the path from the posterior
\begin{align}
    \Pr[\lambda\mid u,x]&=\frac{\alpha_\lambda\mu(u\mid\lambda)p_\lambda(x)}{\sum_{\lambda'}\alpha_{\lambda'}\mu(u\mid\lambda')p_{\lambda'}(x)}.
    \label{eq:sampling-path-posterior}
\end{align}
The denominator is positive on the sampled support. Multiplying by the group--output probability and summing over $u$, with $\sum_u\mu(u\mid\lambda)=1$, yields
\begin{align}
    \Pr[\lambda,x]&=\alpha_\lambda p_\lambda(x).
    \label{eq:sampling-restored-law}
\end{align}
Thus the restored path and output have exactly the joint law of the original batch, as required for its reverse recovery updates.

We implement this reconstruction in two stages. First, we select an original endpoint state using its prior weight within the group and its Born probability for the returned output. Next, we sample an instrument path leading to that endpoint. A counting recurrence allows us to sample the resolved-mode flips without enumerating their possible sequences, while the operations within the selected block are sampled directly. The reconstructed path determines the intermediate input states needed to apply the recovery updates in reverse order.

For a special group, the possible endpoint tags are the Type A state $|\chi_A\rangle$ and the eight states $|\chi_z\rangle$ with block $j$ replaced by an even occupation string $z$. Writing $p_\chi^G(x):=|\langle x|\hat G|\chi\rangle|^2$, their posterior weights are
\begin{align}
    b\,p_{\chi_A}^G(x),\qquad \frac w8\,p_{\chi_z}^G(x)\quad(|z|\text{ even}).
    \label{eq:sampling-posterior-tags}
\end{align}
These are absolute Born probabilities under the parent circuit. The Gaussian component used in the recursive call does not enter these weights.

\begin{lemma}
\label{s1:lem:upgrade-path-reconstruction}
Given a special group of positive probability and an output $x$ with positive probability under that group's averaged output distribution, the original instrument path conditioned on the group and $x$ can be sampled exactly using at most nine scalar Born-probability queries and $O((R+1)\ell)$ additional operations.
\end{lemma}

\begin{proof}
We first select an endpoint tag using the posterior weights in Eq.~\eqref{eq:sampling-posterior-tags}. Let $S_A$ and $S_z$ denote the scalar Pfaffian sums in Eq.~\eqref{s1:eq:scalar-query} for the Type A endpoint and the Type B endpoint with block $j$ in the even occupation state $|z\rangle$, respectively. All queries use the saved parent circuit $\hat G$ and the same output $x$, so they share the Gaussian bra $\langle\zeta|=\langle x|\hat G$ and its coefficient representation. The relative weights are therefore
\begin{align}
    b\,2^{-g}|S_A|^2,
    \qquad
    \frac w8\,2^{-(g-1)}|S_z|^2
    \quad (|z|\ \text{even}).
    \label{s1:eq:scalar-tags}
\end{align}
The common factor $|\kappa|^2$ cancels, but the input normalization factors differ because the Type A endpoint has $g$ intact blocks, whereas each Type B endpoint has $g-1$. Zero-prior tags can be skipped. By Thm.~\ref{s1:thm:active-probability-query}, the at most nine queries take $O(M^3+2^g)$ operations in total.

Conditional on the selected endpoint state, the output $x$ no longer affects the path distribution, because all paths with that endpoint have the same Born likelihood. We therefore need to sample a resolved-mode flip sequence with the prescribed endpoint and, for Type B, a compatible local history.

For a specified endpoint at Hamming distance $\Delta$ from the current resolved occupation string, let $B_t(\Delta)$ count length-$t$ flip sequences reaching that particular endpoint. It does not count all endpoints at distance $\Delta$. Symmetry implies that this number depends only on $\Delta$. Set out-of-range values to zero and compute
\begin{align}
    B_0(\Delta) =\mathbf1_{\Delta=0},
    \qquad B_t(\Delta) =\Delta B_{t-1}(\Delta-1) +(R-\Delta)B_{t-1}(\Delta+1).
    \label{s1:eq:bridge}
\end{align}
The first flipped coordinate is mismatched with the endpoint in $\Delta$ cases and matched in $R-\Delta$ cases, proving the recurrence. The table for $0\le t\le\ell$ and $0\le\Delta\le R$ also supplies the counts defining $b$ and $w$. For $R>0$, constructing it takes $O(R\ell)$ operations. For $R=0$, only the empty sequence is possible, and we handle this case directly.

To generate a sequence with $t$ steps remaining, assign each mismatched coordinate weight $B_{t-1}(\Delta-1)$ and each matched coordinate weight $B_{t-1}(\Delta+1)$. Select a coordinate according to these weights, flip it, and continue. The weights sum to $B_t(\Delta)$. The conditional probabilities along any admissible sequence telescope to the reciprocal of the initial count, so the resulting sequence is uniform. No rejection is required, and a positive-count state never selects a zero-count move. Scanning all resolved coordinates at each step costs $O(R\ell)$ operations.

For the Type A tag, generate a resolved-mode sequence of length $\ell$. For a Type B endpoint $z$, generate one of length $\ell-2$ and choose the two local-visit positions uniformly among the $\binom{\ell}{2}$ pairs of positions. Of the 32 local histories described in Lem.~\ref{s1:lem:upgrade-group-weights}, exactly four end at $z$. Choose one of these four uniformly and interleave it with the resolved-mode sequence.

Every compatible choice of positions, local history, and resolved-mode sequence has the same prior probability. Thus, this procedure gives the required conditional path distribution. Including the local history and interleaving, the additional work is $O((R+1)\ell)$.
\end{proof}

This procedure samples Eq.~\eqref{eq:sampling-path-posterior} without enumerating the paths. Equation~\eqref{eq:sampling-restored-law} then ensures that the reconstructed intermediate states can be used in the original recovery updates.

\subsection{Gate-by-gate sampling}
\label{s1:app:fallback}

The recursion terminates with a gate-by-gate sampler when $R>12g$. We review gate-by-gate sampling for fermionic occupation measurements~\cite{bravyi2022simulate} and then explain how our probability-query algorithm supplies the weights required by this method.

We use the conditional update of Ref.~\cite[Algorithm 2]{bravyi2022simulate}. Its application to fermionic modes follows from one property: a parity-preserving gate $\hat G$ supported on $S$ commutes with every occupation projector outside $S$, including when the modes in $S$ are nonadjacent. If $q^-(u,v)$ and $q^+(u,v)$ are the occupation distributions before and after the gate, with $u$ outside $S$ and $v$ on $S$, then
\begin{align}
    \sum_v q^+(u,v)&=\sum_v q^-(u,v).
    \label{s1:eq:fallback-marginal}
\end{align}
Given a sample from $q^-$, retain $u$ and redraw $v$ from
\begin{align}
    \Pr(v\mid u)&=\frac{q^+(u,v)}{\sum_{v'}q^+(u,v')}.
    \label{s1:eq:fallback-update}
\end{align}
The unchanged outside marginal and the new conditional distribution give exactly $q^+$. The denominator is positive on sampled inputs. Starting from an exact sample of the input occupation distribution and applying this update to successive gates therefore gives an exact output sample.

We now apply this method to our inputs. Each Gaussian gate in a Givens decomposition acts on at most two fermionic modes, so one update requires at most four occupation probabilities. These probabilities can be evaluated using Thm.~\ref{s1:thm:active-probability-query}, giving the following cost bound.

\begin{lemma}
\label{s1:lem:upgrade-fallback}
Let $|\psi\rangle$ be an input of the form in Eq.~\eqref{s1:eq:residual-input}, with $g$ intact blocks on $M$ modes, and let $\hat O$ be a Gaussian unitary specified by $O\in\operatorname{SO}(2M)$. Applying gate-by-gate sampling with the probability-query algorithm of Thm.~\ref{s1:thm:active-probability-query} generates an exact occupation sample of $\hat O|\psi\rangle$ in
\begin{align}
    T_{\rm fallback}(M,g)
    &=O(M^5+M^2 2^g)
    \label{s1:eq:fallback-cost}
\end{align}
arithmetic operations.
\end{lemma}

\begin{proof}
We first sample the input occupation distribution by independently choosing one of the two occupation strings in each intact block with equal probability and retaining the resolved occupations. Next, decompose $O$ into $O(M^2)$ coordinate-plane Givens rotations, which takes $O(M^3)$ operations. Each corresponding Gaussian gate acts on at most two fermionic modes. We apply the update in Eq.~\eqref{s1:eq:fallback-update} successively to these gates.

At each step, let $\hat H$ be the Gaussian prefix consisting of all gates processed so far, including the gate currently being applied. The required weights are
\begin{align}
    q^+(u,v)
    &=|\langle x(u,v)|\hat H|\psi\rangle|^2.
\end{align}
Thus, every probability calculation uses the original coherent input $|\psi\rangle$ and the current circuit prefix. The sampled occupation string is used only for the conditional update; it does not replace $|\psi\rangle$ in this calculation.

For each gate, Thm.~\ref{s1:thm:active-probability-query} evaluates the at most four required absolute Born probabilities in $O(M^3+2^g)$ operations. Applying the update to all $O(M^2)$ gates therefore costs $O(M^5+M^2 2^g)$. The circuit decomposition and prefix-matrix updates are absorbed in this bound. Exactness follows from the gate-by-gate argument above.
\end{proof}

\subsection{Correctness and cost of the worst-case sampler}
\label{s1:app:worst-proof}

The preceding constructions make every step of the batched sampler explicit. We now prove Thm.~\ref{thm:sampling-active-worst} and sum the costs over its recursive calls.

\begin{proof}[Proof of Thm.~\ref{thm:sampling-active-worst}]
\emph{Exactness and recursive progress.} We argue by induction on the number $g$ of intact blocks. A terminal call is exact by Gaussian occupation sampling or Lem.~\ref{s1:lem:upgrade-fallback}. In a nonterminal call, every ordinary path has fewer intact blocks at its endpoint. For a special group, Lems.~\ref{s1:lem:upgrade-group-weights} and~\ref{s1:lem:upgrade-gaussian-ensemble} preserve the group's averaged endpoint while replacing one intact block by a sampled Gaussian state. Thus, the induction hypothesis applies to every recursive child and gives the correct group output law.

Lem.~\ref{s1:lem:upgrade-path-reconstruction} then samples the posterior in Eq.~\eqref{eq:sampling-path-posterior}. By Eq.~\eqref{eq:sampling-restored-law}, the restored path and the returned output have the same joint law as in the original $\ell$-step instrument procedure. Applying its recovery updates in reverse order therefore recovers the parent distribution exactly, by Thm.~\ref{thm:instrument-recovery}. A restored Type A path is used only in recovery and creates no new recursive call. Since every nonterminal call reduces $g$ by at least one, there are at most $n$ such calls and a single terminal call.

\emph{Cost of one batch.} A length-$\ell$ batch uses $\ell$ neighboring-probability queries for recovery and at most nine scalar Born-probability queries for path restoration. Recording the batch and reconstructing its intermediate input states takes $O(M\ell)$ operations; preparing the child circuit takes at most $O(M^2)$. The counting and conditional-path procedure of Lem.~\ref{s1:lem:upgrade-path-reconstruction} uses $O((R+1)\ell)$ additional operations. The probability-query bound of Thm.~\ref{s1:thm:active-probability-query} therefore gives total nonrecursive work
\begin{align}
    O\!\left(\ell M^3+M\ell+R\ell+(\ell+1)2^g\right).
    \label{eq:sampling-batch-cost}
\end{align}
The group weights are obtained from the same counting table used for path restoration, and sampling the fixed four-mode Gaussian ensemble has a constant cost.

\emph{Total cost.} At the $s$-th nonterminal call, write $k_s:=n-g_s$. Then $R_s=4k_s$ and $\ell_s=4k_s+2$. Strict decrease of $g_s$ makes the $k_s$ distinct, so the exponential contribution is bounded by
\begin{align}
    \sum_s(\ell_s+1)2^{g_s}&\leq2^n\sum_{k=0}^{\infty}\frac{4k+3}{2^k}=14\,2^n.
    \label{eq:sampling-worst-geometric}
\end{align}
Also $\sum_s\ell_s=O(n^2)$ and $M=4n$, so the accumulated polynomial work is $O(n^5)$.

If the terminal call has $g=0$, Gaussian occupation sampling has polynomial cost within this bound. Otherwise $R>12g$, which implies $g<n/4$. Lem.~\ref{s1:lem:upgrade-fallback} then bounds the terminal cost by
\begin{align}
    O(M^5+M^2 2^g)&=O(n^5+2^n).
    \label{eq:sampling-fallback-cost}
\end{align}
Adding the batch costs and the single terminal cost proves the claimed per-run bound.
\end{proof}

\section{Proof for exact expectation-value problem}\label{App:Main}

In the main text, we discussed the exact expectation-value problem for even-parity block-product inputs with constant block sizes and presented a dynamic programming algorithm. We illustrated its construction using $|\Psi_n\rangle=|\phi_4\rangle^{\otimes n}$. In this section, we provide a detailed derivation and proof of the algorithm. Here, we extend the main-text result to matrix elements between two possibly different even-parity block-product states with arbitrary block sizes. The expectation values considered in the main text are recovered by taking the two states to be identical.

\subsection{Restatement of the theorem}
\label{App:Mean theorem}

We restate the main-text result in a more general form, allowing different block sizes and matrix elements between two possibly different block-product states. Let $L_1,\ldots,L_n$ be a partition of the $M$ fermionic modes into disjoint blocks, with $|L_q|=m_q$ and $M=\sum_{q=1}^n m_q$. We order the modes block by block. Consider the block-product states
\begin{align}
    |\Psi\rangle=\bigotimes_{q=1}^n|\psi_q\rangle,
    \qquad |\Phi\rangle =\bigotimes_{q=1}^n|\phi_q\rangle,
\end{align}
where $|\psi_q\rangle$ and $|\phi_q\rangle$ are normalized pure states on block $q$. Their amplitudes in the local occupation basis are supplied as input. For simplicity, we assume that both states have even parity on every block. We explain below how the argument extends to blocks of either parity, provided that $|\psi_q\rangle$ and $|\phi_q\rangle$ have the same definite parity for each $q$.

Let $\hat O$ be an $M$-mode Gaussian unitary specified by its Majorana matrix $O\in\operatorname{SO}(2M)$, and let $\hat\Gamma$ be a Hermitian Majorana string of length $1\leq k\leq M$, i.e., $\hat\Gamma=i^{k(k-1)/2} \hat{\gamma}_{a_1}\cdots\hat{\gamma}_{a_k}$, where $1\leq a_1<\cdots<a_k\leq2M$. Then, our goal is to compute
\begin{align}
    \langle \hat\Gamma\rangle &:=\langle\Psi| \hat O^\dagger\hat\Gamma\hat O |\Phi\rangle.
\end{align}
Taking $|\Phi\rangle=|\Psi\rangle$ recovers the expectation value considered in the main text. We now state the algorithmic result.

\begin{theorem}
\label{Thm:Expectation value of a Majorana string}
Under the assumptions above, there is a classical algorithm that exactly computes $\langle\hat\Gamma\rangle$ in $O(k2^k\sum_{q=1}^n2^{2m_q})$ time. In particular, if $m_q=O(\log n)$ for all $q$, then the time complexity becomes $O(k2^k\poly(n))$. Additionally, if $m_q= O(1)$ for all $q$, then the result in the main text is recovered.
\end{theorem}

For $k=0$, the string is the identity, and the target quantity is computed directly as
\begin{align}
    \langle \hat I\rangle
    &=\langle\Psi|\Phi\rangle
      =\prod_{q=1}^n\langle\psi_q|\phi_q\rangle.
\end{align}
Moreover, strings longer than $M$ are handled using the complementary-string relation proved below.

We first record the parity reduction and the block factorization identity. We then represent the subset table in an auxiliary Fock space, where creation operators account for the fermionic signs. This representation gives a factorized block update whose direct implementation proves Thm.~\ref{Thm:Expectation value of a Majorana string}.

\subsection{Properties of Majorana matrix elements}

We use two consequences of fermionic parity and the Jordan--Wigner representation: a relation between complementary Majorana strings and factorization over even-parity input blocks. We record the signs explicitly in our conventions; for the underlying Majorana algebra and parity structure, see Refs.~\cite{Bravyi2005,Parity}.

For $S=\{s_1<\cdots<s_k\}\subseteq[2M]$, let $\bar S=[2M]\setminus S$ and define
\begin{align}
    \hat\Gamma_S &:= i^{\binom{k}{2}}\hat\gamma_{s_1}\cdots\hat\gamma_{s_k},
    \qquad \hat\Pi:=(-i)^M\hat\gamma_1\cdots\hat\gamma_{2M},
\end{align}
with $\hat\Gamma_{\varnothing}=\hat I$. If $k$ is even, the canonical anticommutation relations give
\begin{align}
    \hat\Gamma_S\hat\Gamma_{\bar S} &= \eta_M(S)\hat\Pi,
    \qquad \eta_M(S):=(-1)^{M+\operatorname{inv}(S,\bar S)},
\end{align}
where $\operatorname{inv}(S,\bar S)$ counts pairs $(s,t)\in S\times\bar S$ with $s>t$. Since $\hat\Gamma_{\bar S}^2=\hat I$ and $\hat O|\Psi\rangle$ has even parity, this yields
\begin{align}
    \langle\hat\Gamma_S\rangle &= \eta_M(S)\langle\hat\Gamma_{\bar S}\rangle.
    \label{eq:mean-complementary-string}
\end{align}
For odd $k$, both matrix elements vanish because $\hat O|\Psi\rangle$ and $\hat O|\Phi\rangle$ have even parity. Thus, Eq.~\eqref{eq:mean-complementary-string} allows us to use a string of length at most $M$, with a sign correction computable in $O(M)$ time.

The block factorization follows from the Jordan--Wigner representation and its tensor-product structure for locally even operators~\cite{JordanWigner,Parity}; see in particular Secs.~2 and 7.1--7.2 of Ref.~\cite{Parity}. Let $\hat P_q$ denote the parity operator on block $L_q$, and let $\hat A_q$ be supported on that block and have definite parity, $\hat P_q\hat A_q\hat P_q=(-1)^{\sigma_q}\hat A_q$ with $\sigma_q\in\{0,1\}$. Because the modes are ordered block by block, its Jordan--Wigner representation consists of its local matrix preceded by $\hat P_1^{\sigma_q}\otimes\cdots\otimes \hat P_{q-1}^{\sigma_q}$. In the ordered product $\hat A_1\cdots \hat A_n$, these parity factors act trivially on the even-parity block kets, so
\begin{align}
    \langle\Psi|\hat A_1\cdots \hat A_n|\Phi\rangle &= \prod_{q=1}^n\langle\psi_q|\hat A_q|\phi_q\rangle.
    \label{eq:mean-block-factorization}
\end{align}
On the right-hand side, $\hat A_q$ denotes its matrix on block $q$. If any $\hat A_q$ is odd, both sides vanish. Each product of block-local Majorana linear combinations has definite parity, so Eq.~\eqref{eq:mean-block-factorization} applies to every term obtained by grouping the transformed Majorana factors by input block.

The block factorization also holds when corresponding bra and ket blocks have the same definite parity, whether even or odd, because matrix elements of odd local operators still vanish. Therefore, the recurrence and table update below remain unchanged, with the same runtime. If the common total parity of the two states is odd, the right-hand side of Eq.~\eqref{eq:mean-complementary-string} acquires an additional minus sign. Odd-length matrix elements still vanish.

\subsection{Auxiliary-fermion representation}
\label{App:Mean recurrence}

We encode subsets of the transformed Majorana factors as occupation states of auxiliary fermions. Their creation operators keep track of both which factors have been used and the signs from reordering them. This turns the expansion over assignments to input blocks into a product of operators on a $2^k$-dimensional space.

For the indices $a_1<\cdots<a_k$ of $\hat\Gamma$, define
\begin{align}
    \hat\eta_t&:=\hat O^\dagger\hat\gamma_{a_t}\hat O=\sum_{q=1}^n\hat\eta_t^{[q]},\qquad \hat\eta_t^{[q]}:=\sum_{r\in\widetilde L_q}O_{a_t r}\hat\gamma_r,
\end{align}
where $\widetilde L_q:=\{2j-1,2j:j\in L_q\}$ is the Majorana operator index set of block $q$. For $A=\{t_1<\cdots<t_s\}\subseteq[k]$, let
\begin{align}
    \hat\eta_A^{[q]}&:=\hat\eta_{t_1}^{[q]}\cdots\hat\eta_{t_s}^{[q]},\qquad c_q(A):=\langle\psi_q|\hat\eta_A^{[q]}|\phi_q\rangle.
    \label{Eq:Local eta}
\end{align}
As in Eq.~\eqref{eq:mean-block-factorization}, the matrix element uses the local representation on block $q$. We set $\hat\eta_{\varnothing}^{[q]}=\hat I$, so $c_q(\varnothing)=\langle\psi_q|\phi_q\rangle$. Matching block parities imply $c_q(A)=0$ for odd $|A|$.

Let $\mathcal F_k$ be the Fock space of $k$ auxiliary modes with creation operators $\hat b_1^\dagger,\ldots,\hat b_k^\dagger$. Attach it to the physical space by the ordinary tensor product $\mathcal H\otimes\mathcal F_k$. Thus $[\hat\eta_t\otimes\hat I,\hat I\otimes\hat b_s^\dagger]=0$: physical and auxiliary operators commute across the tensor factors, while the canonical anticommutation relations hold within each space. The auxiliary modes record subsets and their reordering signs. Contracting a physical bra and ket leaves an operator on $\mathcal F_k$; we display the auxiliary identity $\hat I_{\mathcal F_k}$ explicitly in such contractions below. For $S=\{s_1<\cdots<s_\ell\}$, use the basis and ordered creation products
\begin{align}
    |S\rangle&:=\hat b_S^\dagger|\varnothing\rangle,\qquad \hat b_S^\dagger:=\hat b_{s_1}^\dagger\cdots\hat b_{s_\ell}^\dagger,\qquad \hat b_{\varnothing}^\dagger:=\hat I.
\end{align}
All products indexed by $t$ below are ordered increasingly from left to right. Define the auxiliary operator for block $q$ by
\begin{align}
    \hat M_q&:=\sum_{A\subseteq[k]}c_q(A)\hat b_A^\dagger.
    \label{eq:mean-aux-transfer}
\end{align}

\begin{lemma}\label{lem:mean-aux-factorization}
As an operator on $\mathcal F_k$,
\begin{align}
    (\langle\Psi|\otimes\hat I_{\mathcal F_k})\left[\prod_{t=1}^k\bigl(\hat I+\hat\eta_t\otimes\hat b_t^\dagger\bigr)\right](|\Phi\rangle\otimes\hat I_{\mathcal F_k})&=\hat M_1\cdots\hat M_n.
    \label{eq:mean-aux-factorization}
\end{align}
\end{lemma}

\begin{proof}
Set $\hat X_t^{[q]}:=\hat\eta_t^{[q]}\otimes\hat b_t^\dagger$, using the global physical operators, including their Jordan--Wigner parity strings. If $q\neq q'$, the physical factors anticommute, as do the auxiliary creation operators. The two signs cancel, giving
\begin{align}
    \hat X_t^{[q]}\hat X_s^{[q']}&=\hat X_s^{[q']}\hat X_t^{[q]},\qquad q\neq q'.
\end{align}
For $s=t$, both products are zero. The same nilpotence eliminates all terms containing two copies of $\hat b_t^\dagger$, so
\begin{align}
    \hat I+\hat\eta_t\otimes\hat b_t^\dagger&=\prod_{q=1}^n(\hat I+\hat X_t^{[q]}).
\end{align}
Substitute this identity for every $t$ and regroup by $q$. This only exchanges factors on distinct blocks and preserves their order within each block, yielding
\begin{align}
    \prod_{t=1}^k\bigl(\hat I+\hat\eta_t\otimes\hat b_t^\dagger\bigr)&=\prod_{q=1}^n\left(\sum_{A\subseteq[k]}\hat\eta_A^{[q]}\otimes\hat b_A^\dagger\right).
\end{align}
There is no additional sign within each tensor product. Taking the physical matrix element and applying Eq.~\eqref{eq:mean-block-factorization} term by term replaces each ordered product of physical block operators by the product of the corresponding $c_q(A)$. The right-hand side becomes $\hat M_1\cdots\hat M_n$.
\end{proof}

Only the term selecting all $k$ factors on the left-hand side contributes between $\langle[k]|$ and $|\varnothing\rangle$. Since $\hat b_1^\dagger\cdots\hat b_k^\dagger|\varnothing\rangle=|[k]\rangle$, the lemma gives
\begin{align}
    \langle\hat\Gamma\rangle&=i^{k(k-1)/2}\langle[k]|\hat M_1\cdots\hat M_n|\varnothing\rangle.
    \label{eq:mean-aux-output}
\end{align}
Moreover, each $\hat M_q$ contains only even products of creation operators. Such products commute: exchanging disjoint even products gives a positive sign, while overlapping products vanish. We can therefore apply the blocks in the order $1,\ldots,n$, defining
\begin{align}
    |w_0\rangle&:=|\varnothing\rangle,\qquad |w_q\rangle:=\hat M_q|w_{q-1}\rangle,\qquad f_q(S):=\langle S|w_q\rangle.
    \label{eq:mean-aux-recurrence}
\end{align}
The index $q$ counts the processed blocks, whereas $S\subseteq[k]$ labels an entry in the table. The initial table has $f_0(\varnothing)=1$ and all other entries zero; the first update gives $f_1(S)=c_1(S)$. Subsequent updates use the complete preceding table, and the final answer is $i^{k(k-1)/2}f_n([k])$. Every table vanishes on odd subsets.

This representation already specifies the block-by-block recurrence. Its direct expansion would visit $3^k$ disjoint pairs of subsets per block. We next apply $\hat M_q$ in factorized form, computing the same table with only $O(k2^k)$ local matrix--vector operations.

\subsection{Efficient evaluation of the block updates}\label{App:Mean proof}\label{App:Subroutine}

The auxiliary operator $\hat M_q$ need not be constructed as a matrix, and its coefficients $c_q(A)$ need not be evaluated individually. To apply it, we retain the physical state of block $q$ while processing the $k$ factors and contract with $\langle\psi_q|$ only at the end. Expanding the ordered product shows that, in the local representation on block $q$,
\begin{align}
    \hat M_q&=(\langle\psi_q|\otimes\hat I_{\mathcal F_k})\left[\prod_{t=1}^k\bigl(\hat I+\hat\eta_t^{[q]}\otimes\hat b_t^\dagger\bigr)\right](|\phi_q\rangle\otimes\hat I_{\mathcal F_k}).
    \label{eq:mean-local-transfer}
\end{align}
Here $\hat\eta_t^{[q]}$ is the $2^{m_q}\times2^{m_q}$ local matrix, without the preceding-block parity strings. Applying the product to a vector processes its factors from right to left, hence $t=k,\ldots,1$.

\begin{lemma}\label{lem:mean-one-block-update}
Given the preceding table $f_{q-1}$, the next table in Eq.~\eqref{eq:mean-aux-recurrence} can be computed in $O(k2^k2^{2m_q})$ arithmetic operations.
\end{lemma}

For a fixed block $q$, introduce one local vector $v_S$ for each $S\subseteq[k]$. Together they represent the joint vector $\sum_S v_S\otimes|S\rangle$ in the block space tensored with $\mathcal F_k$. The update consists of three steps:
\begin{enumerate}
    \item \textbf{Initialize from the preceding table.} Set
    \begin{align}
        v_S&:=f_{q-1}(S)|\phi_q\rangle,\qquad S\subseteq[k].
        \label{eq:mean-vector-initialization}
    \end{align}
    This represents $|\phi_q\rangle\otimes|w_{q-1}\rangle$.
    \item \textbf{Apply the factors.} For $t=k,\ldots,1$, update every subset containing $t$ by
    \begin{align}
        v_S&\leftarrow v_S+(-1)^{|S\cap[t-1]|}\hat\eta_t^{[q]}v_{S\setminus\{t\}},\qquad t\in S,
        \label{Eq:Update rule m=1}
    \end{align}
    leaving the other vectors unchanged, with $[0]=\varnothing$. The source vector has no $t$ in its index and is unchanged during the pass, so these updates can be performed in place.
    \item \textbf{Contract to obtain the new table.} Set
    \begin{align}
        f_q(S)&=\langle\psi_q|v_S\rangle.
        \label{eq:mean-vector-contraction}
    \end{align}
    Pass this scalar table to the next block; the local vectors can then be discarded.
\end{enumerate}

\begin{proof}
For $t\in S$, the canonical anticommutation relations give
\begin{align}
    \hat b_t^\dagger|S\setminus\{t\}\rangle&=(-1)^{|S\cap[t-1]|}|S\rangle,
    \label{eq:mean-creation-sign}
\end{align}
while creation on an already occupied mode gives zero. Consequently, Eq.~\eqref{Eq:Update rule m=1} is exactly the coefficient update for applying $\hat I+\hat\eta_t^{[q]}\otimes\hat b_t^\dagger$ to $\sum_Sv_S\otimes|S\rangle$. Processing $t=k,\ldots,1$ and taking the final physical inner product therefore applies the operator in Eq.~\eqref{eq:mean-local-transfer}. The output is $\hat M_q|w_{q-1}\rangle=|w_q\rangle$, proving correctness. In particular, the sign in the vector update follows directly from auxiliary creation and requires no separate parity invariant.

Each vector has dimension $2^{m_q}$, and each update uses one dense local matrix--vector multiplication, costing $O(2^{2m_q})$ operations. Each of the $k$ passes updates $2^{k-1}$ subsets; initialization and final contractions cost $O(2^k2^{m_q})$. The local matrices can be assembled within $O(k2^{2m_q})$ operations: each local Majorana matrix has one nonzero entry per row, and $m_q2^{m_q}=O(2^{2m_q})$. Accessing the required entries of $O$ is also absorbed in this bound. Thus the total cost is $O(k2^k2^{2m_q})$.
\end{proof}

\begin{proof}[Proof of Thm.~\ref{Thm:Expectation value of a Majorana string}]
The complete algorithm is as follows:
\begin{enumerate}
    \item Initialize $f_0(\varnothing)=1$ and $f_0(S)=0$ for nonempty $S\subseteq[k]$.
    \item For $q=1,\ldots,n$, apply Lem.~\ref{lem:mean-one-block-update} to compute the complete table $f_q$ from $f_{q-1}$, then discard the preceding table.
    \item Return $i^{k(k-1)/2}f_n([k])$.
\end{enumerate}
The table after block $q$ contains the coefficients of $\hat M_q\cdots\hat M_1|\varnothing\rangle$. Since the $\hat M_q$ commute, Eq.~\eqref{eq:mean-aux-output} proves that the returned coefficient is the desired matrix element. The identity case and the complementary-string reduction are handled as described above. Summing the costs of the block updates gives
\begin{align}
    O\!\left(k2^k\sum_{q=1}^n2^{2m_q}\right).
\end{align}
This bound assumes access to entries of the supplied matrix $O$. Reading the entire dense matrix adds $O(M^2)$, as included in the main-text runtime. For constant-size blocks, the total is $O(n^2+nk2^k)$, completing the proof.
\end{proof}

\section{Proof for the approximate expectation-value problem}\label{app:approx}

We prove the transition-amplitude estimation guarantee of Thm.~\ref{thm:active-flo-inner-product-estimation}, which underlies the additive-error probability and observable estimates in Sec.~\ref{Sec:Approximate mean value}. We build the estimator one reduction at a time, starting with occupation sampling and then replacing the remaining amplitude by a Gaussian-overlap estimate. This construction identifies the second-moment bound needed for efficiency. We prove that bound before collecting the steps into the complete algorithm.

\subsection{Restatement of the theorem}
Let $\hat O$ be a parity-preserving fermionic Gaussian unitary on $4n$ modes, corresponding to $O\in\operatorname{SO}(8n)$, and let
\begin{align}
    \ket{\Phi} &=\bigotimes_{j=1}^n\ket{\phi_j}, \qquad \ket{\Psi}=\bigotimes_{j=1}^n\ket{\psi_j},
\end{align}
where every block is a normalized pure even-parity state on four modes. Local Gaussian unitaries bring the blocks to the canonical forms below~\cite{bako2025fermionic}. Absorbing these unitaries into $\hat O$, we may write
\begin{align}
    \ket{\phi_j} &=\cos\alpha_j\ket{1100}+\sin\alpha_j\ket{0011}:=u_j\ket{1100}+v_j\ket{0011},\\
    \ket{\psi_j} &=\cos\beta_j\ket{1100}+\sin\beta_j\ket{0011}:=c_{j,a}\ket{1100}+c_{j,b}\ket{0011},
\end{align}
with $\alpha_j,\beta_j\in[0,\pi/4]$. Our target is the transition amplitude
\begin{align}
    A &:=\bra{\Phi}\hat O\ket{\Psi}.
\end{align}

The unitary $\hat O$ is supplied as a polynomial-size Gaussian circuit including its overall phase.

\begin{theorem}[Additive-error estimation of Gaussian transition amplitudes]\label{thm:Approximate mean value}
Given classical descriptions of $\ket{\Phi}$, $\ket{\Psi}$, and $\hat O$ as above and $0<\epsilon,\delta<1$, a randomized classical algorithm returns a complex number $\widetilde A$ satisfying
\begin{align}\label{eq:app-amplitude-guarantee}
    \Pr\!\left[|\widetilde A-A|\leq\epsilon\right] &\geq1-\delta.
\end{align}
After a one-time polynomial-time preprocessing of the Gaussian circuit that retains its overall phase, the algorithm uses $O(\epsilon^{-2}\log(2/\delta))$ independent samples, each generated with $O(n^3)$ arithmetic operations. The arithmetic cost after preprocessing is therefore $O(n^3\epsilon^{-2}\log(2/\delta))$.
\end{theorem}

This also covers an occupation-state bra $\langle x|$, as stated in the main text. If $x$ has even total occupation, a Gaussian unitary prepares $|x\rangle$ from vacuum; absorbing its inverse into $\hat O$ reduces the bra to a product of four-mode vacuum states. For odd total occupation the amplitude vanishes by parity conservation. The number-correlator guarantee then follows from the parity-operator reduction in the proof of Thm.~\ref{thm:active-number-correlators}.

\subsection{Constructing the estimator}\label{app:approx-construction}

We construct the estimator by reducing the target amplitude to successively simpler quantities. First, sampling an occupation component of $|\Psi\rangle$ leaves an amplitude with a Gaussian ket. We then incorporate the coefficients of the bra into a Gaussian operation and estimate the remaining overlap by sampling a Gaussian decomposition. The freedom in this decomposition will let us control the second moment.

\subsubsection{Sampling an occupation component}

Write $|a\rangle:=|1100\rangle$ and $|b\rangle:=|0011\rangle$. The canonical form of the ket gives
\begin{align}\label{eq:canonical-expansion}
    |\Psi\rangle &=\sum_{J\in\{a,b\}^n}c_J|J\rangle, \qquad c_J:=\prod_{j=1}^n c_{j,J_j}.
\end{align}
The coefficients are nonnegative and satisfy $\sum_Jc_J^2=1$. We can therefore sample $J$ with probability $c_J^2$ by choosing its entries independently. This writes the target as
\begin{align}\label{eq:app-outer-sampling}
    A &=\mathbb E_{J\sim c_J^2}\!\left[\frac{\langle\Phi|\hat O|J\rangle}{c_J}\right].
\end{align}
Only strings with $c_J>0$ are sampled, so the division is well defined. If the amplitude in the numerator were available exactly, this estimator would already have second moment at most one:
\begin{align}\label{eq:app-ideal-second-moment}
    \mathbb E_{J\sim c_J^2}\!\left[\left|\frac{\langle\Phi|\hat O|J\rangle}{c_J}\right|^2\right] &=\sum_{J:c_J>0}|\langle\Phi|\hat O|J\rangle|^2\leq1.
\end{align}
The inequality follows from the orthonormality of the states $\hat O|J\rangle$ and the normalization of $|\Phi\rangle$. Thus occupation sampling gives a useful starting point. Its remaining difficulty is the numerator: $\hat O|J\rangle$ is Gaussian, whereas the bra $\langle\Phi|$ is generally non-Gaussian. We next construct an efficiently computable unbiased estimate of this numerator, keeping track of the extra second moment introduced by that estimate.

\subsubsection{Moving the bra coefficients into a Gaussian operation}

For fixed $J$, absorb the bra coefficients into the Gaussian ket using the local non-unitary Gaussian operations
\begin{align}\label{eq:app-bounded-filter}
    \hat F_j &:=\operatorname{diag}(v_j^{1/4},u_j^{1/4})^{\otimes2}\otimes\operatorname{diag}(u_j^{1/4},v_j^{1/4})^{\otimes2}.
\end{align}
The factors act on the four modes in the occupation basis. Since $\hat F_j|a\rangle=u_j|a\rangle$ and $\hat F_j|b\rangle=v_j|b\rangle$, we have
\begin{align}\label{eq:app-local-filter-identity}
    \langle\phi_j| &=\sqrt2\langle\psi_{\rm eq}|\hat F_j, \qquad |\psi_{\rm eq}\rangle:=\frac{|a\rangle+|b\rangle}{\sqrt2}.
\end{align}
These operations are Gaussian, including at $v_j=0$, where $\hat F_j=|a\rangle\langle a|$ is a Fock projection~\cite{Bravyi2005}.

With $\hat F:=\bigotimes_j\hat F_j$, the amplitude therefore reduces to an overlap with the equal-weight block-product state:
\begin{align}\label{eq:app-filtered-amplitude}
    |V_J\rangle &:=\hat F\hat O|J\rangle, \qquad \langle\Phi|\hat O|J\rangle=2^{n/2}\langle\psi_{\rm eq}|^{\otimes n}|V_J\rangle.
\end{align}
If $|V_J\rangle=0$, return zero for this sample, otherwise we continue to the next step.

\subsubsection{Estimating the remaining overlap by Gaussian decomposition}

For any phase $\varphi$, the equal-weight state admits the Gaussian decomposition
\begin{align}\label{eq:app-equal-state-decomposition}
    |\psi_{\rm eq}\rangle &=\frac{|g_1(\varphi)\rangle+|g_{-1}(\varphi)\rangle}{\sqrt2},
\end{align}
where
\begin{align}\label{eq:app-phase-gaussian}
    |g_s(\varphi)\rangle &:=\frac{|a\rangle+|b\rangle+s e^{-i\varphi/2}|\text{vac}\rangle+s e^{i\varphi/2}|\text{full}\rangle}{2}, \qquad s\in\{\pm1\}.
\end{align}
Here $|\text{vac}\rangle:=|0000\rangle$ and $|\text{full}\rangle:=|1111\rangle$. Each $|g_s(\varphi)\rangle$ is a product of normalized two-mode Gaussian states on pairs $(1,2)$ and $(3,4)$.

For now, fix any phases $\bm\varphi=(\varphi_1,\ldots,\varphi_n)$ and draw $\bm s$ uniformly from $\{\pm1\}^n$. Define
\begin{align}
    |G_{\bm s}(\bm\varphi)\rangle &:=\bigotimes_{j=1}^n|g_{s_j}(\varphi_j)\rangle, \qquad Z_J:=2^n\langle G_{\bm s}(\bm\varphi)|V_J\rangle.\label{eq:app-inner-estimator}
\end{align}
The Gaussian decomposition and Eq.~\eqref{eq:app-filtered-amplitude} immediately give
\begin{align}\label{eq:app-inner-mean}
    \mathbb E_{\bm s}[Z_J] &=2^{n/2}\langle\psi_{\rm eq}|^{\otimes n}|V_J\rangle=\langle\Phi|\hat O|J\rangle.
\end{align}
Thus $Z_J$ is an unbiased estimate of the required numerator, and each sample is a Gaussian overlap that can be evaluated in polynomial time, including its complex phase~\cite{DiasKonig2024,reardon2024improved}. 

At this point the mean is correct, but the prefactor $2^n$ could lead to a large second moment. The decomposition phases have not yet been specified. We choose them next so that each sampled overlap is small enough to compensate for this prefactor. The phases may depend on $J$ through $|V_J\rangle$, but they are fixed before the signs $\bm s$ are drawn; Eq.~\eqref{eq:app-inner-mean} then continues to hold for every $J$.

\subsubsection{Choosing the phases to control the second moment}

For a fixed $J$, we use the freedom in the Gaussian decomposition to obtain $\mathbb E_{\bm s}|Z_J|^2\leq\|V_J\|^2$. It suffices to make every squared overlap in Eq.~\eqref{eq:app-inner-estimator} at most $2^{-2n}\|V_J\|^2$. The key is that a single phase choice on each block will control both signs on that block, and hence all $2^n$ Gaussian components simultaneously. To see how, use the covariance convention
\begin{align}
    \Gamma_{pq}(\omega) &:=-\frac{i}{2}\langle\omega|[\hat\gamma_p,\hat\gamma_q]|\omega\rangle
\end{align}
for a normalized state $|\omega\rangle$. Let $K(\varphi)$ be the $8\times8$ covariance matrix of $|g_1(\varphi)\rangle$. Equation~\eqref{eq:app-phase-gaussian} gives covariance $sK(\varphi)$ for $|g_s(\varphi)\rangle$ and
\begin{align}\label{eq:Klinear-simple}
    K(\varphi) &=K_{\rm c}\cos(\varphi/2)+K_{\rm s}\sin(\varphi/2), \qquad K_{\rm c}:=K(0),\quad K_{\rm s}:=K(\pi).
\end{align}
Write $\Gamma_J$ for the covariance matrix of $|V_J\rangle/\|V_J\|$, with diagonal blocks $(\Gamma_J)_{jj}$. Varying $\varphi_j$ moves $K(\varphi_j)$ within the span of $K_{\rm c}$ and $K_{\rm s}$. We choose it orthogonal to $(\Gamma_J)_{jj}$ in the Frobenius inner product:
\begin{align}\label{eq:balance-simple}
    \Tr\!\left[K(\varphi_j)^T(\Gamma_J)_{jj}\right] &=0.
\end{align}
Then $-K(\varphi_j)$ is orthogonal to the same block as well, so this one condition handles both choices of $s_j$. By Eq.~\eqref{eq:Klinear-simple}, finding the phase requires solving a single homogeneous equation in $\cos(\varphi_j/2)$ and $\sin(\varphi_j/2)$, with coefficients given by two fixed-size matrix traces. A solution always exists; if both coefficients vanish, take $\varphi_j=0$. Denote these phases by $\bm\varphi(J)$.

Each Gaussian component $|G_{\bm s}\rangle$ has a block-diagonal covariance matrix, so its trace inner product with $\Gamma_J$ is the sum of the local terms above, multiplied by the signs $s_j$. All these terms vanish separately. Thus $n$ scalar conditions enforce the same global trace constraint for every sign string, even when $|V_J\rangle$ has correlations between blocks. The Gaussian overlap formula expresses the squared overlap through a determinant; the proof below uses the trace constraint and the arithmetic--geometric mean inequality to bound that determinant uniformly over $\bm s$.

\begin{lemma}\label{lem:branch-bound}
Let $|V\rangle\neq0$ be an unnormalized even pure Gaussian vector on $n$ four-mode blocks, with covariance matrix $\Gamma$ for $|V\rangle/\|V\|$. Write $\Gamma_{jj}\in\mathbb R^{8\times8}$ for its $j$-th diagonal principal block. If the phases satisfy
\begin{align}\label{eq:lemma-phase-balance}
    \Tr\!\left[K(\varphi_j)^T\Gamma_{jj}\right] &=0 \qquad (j=1,\ldots,n),
\end{align}
then every $\bm s\in\{\pm1\}^n$ obeys
\begin{align}\label{eq:branch-pointwise}
    |\langle G_{\bm s}(\bm\varphi)|V\rangle|^2 &\leq 2^{-2n}\|V\|^2.
\end{align}
\end{lemma}

\begin{proof}
Fix $\bm s$ and let $K_{\bm s}:=\bigoplus_{j=1}^ns_jK(\varphi_j)$ be the covariance matrix of $|G_{\bm s}(\bm\varphi)\rangle$. Purity gives $\Gamma^T\Gamma=K_{\bm s}^TK_{\bm s}=I$, so the positive semidefinite matrix
\begin{align}\label{eq:Hs-expanded}
    H_{\bm s} &:=(\Gamma+K_{\bm s})^T(\Gamma+K_{\bm s})=2I+K_{\bm s}^T\Gamma+\Gamma^TK_{\bm s}
\end{align}
has trace $16n$. Indeed, the cross-term trace is $2\sum_js_j\Tr[K(\varphi_j)^T\Gamma_{jj}]$, which vanishes block by block by the phase condition. This is why the same bound holds for every choice of signs. The arithmetic--geometric mean inequality for the $8n$ eigenvalues of $H_{\bm s}$ gives
\begin{align}\label{eq:global-det-bound}
    \det H_{\bm s} &\leq\left(\frac{\Tr H_{\bm s}}{8n}\right)^{8n}=2^{8n}.
\end{align}
Using the pure-Gaussian overlap formula~\cite{DiasKonig2024}, we obtain
\begin{align}
    \frac{|\langle G_{\bm s}(\bm\varphi)|V\rangle|^2}{\|V\|^2} &=2^{-4n}|\operatorname{Pf}(\Gamma+K_{\bm s})|=2^{-4n}(\det H_{\bm s})^{1/4}\leq2^{-2n}.
\end{align}
Multiplication by $\|V\|^2$ proves the claim.
\end{proof}

Apply this lemma to $|V_J\rangle$ with the phases chosen in Eq.~\eqref{eq:balance-simple}. The factor $2^{-2n}$ cancels the squared prefactor in $Z_J$, giving
\begin{align}\label{eq:app-inner-second-moment}
    \mathbb E_{\bm s}[|Z_J|^2] &\leq\|V_J\|^2.
\end{align}
Together with Eq.~\eqref{eq:app-inner-mean}, this completes the estimate for a fixed occupation string. Both statements also hold when $|V_J\rangle=0$, for which we set $Z_J=0$. We next include the randomness in $J$; the bound above must still be summed over occupation strings.

\subsection{Averaging over occupation strings}\label{app:correctness}

We have constructed, for each fixed $J$, a polynomial-time sample $Z_J$ with the correct mean and a second moment controlled by $\|V_J\|^2$. Returning to the occupation sampling in Eq.~\eqref{eq:app-outer-sampling}, define the complete sample by
\begin{align}\label{Eq:Definition of Y}
    Y &:=\frac{Z_J}{c_J}, \qquad J\sim c_J^2.
\end{align}
The mean of $Y$ will be $A$. To control its second moment, we must sum the fixed-$J$ bounds over all sampled occupation strings. Introduce the positive Gaussian operator and the occupation-subspace projector
\begin{align}\label{eq:Qdef}
    \hat Q_j &:=\hat F_j^2, \qquad \hat Q:=\bigotimes_{j=1}^n\hat Q_j=\hat F^2, \qquad \hat P_{\mathcal C}:=\left(|a\rangle\langle a|+|b\rangle\langle b|\right)^{\otimes n}=\sum_{J\in\{a,b\}^n}|J\rangle\langle J|.
\end{align}
Then $\|V_J\|^2=\langle J|\hat O^\dagger\hat Q\hat O|J\rangle$, so the sum of these norms can be written as a single trace.

\begin{proposition}\label{prop:importance}
The sample $Y$ defined in Eq.~\eqref{Eq:Definition of Y} satisfies
\begin{align}
    \mathbb E[Y] &=\langle\Phi|\hat O|\Psi\rangle, \qquad \mathbb E[|Y|^2]\leq\Tr\!\left(\hat Q\hat O\hat P_{\mathcal C}\hat O^\dagger\right).
\end{align}
\end{proposition}

\begin{proof}
Using the conditional mean in Eq.~\eqref{eq:app-inner-mean}, we obtain
\begin{align}
    \mathbb E[Y] &=\sum_{J:c_J>0}c_J\mathbb E_{\bm s}[Z_J]=\sum_Jc_J\langle\Phi|\hat O|J\rangle=\langle\Phi|\hat O|\Psi\rangle.
\end{align}
For the second moment, the probability $c_J^2$ cancels the squared weight $1/c_J^2$. The bound in Eq.~\eqref{eq:app-inner-second-moment} therefore gives
\begin{align}
    \mathbb E[|Y|^2] &=\sum_{J:c_J>0}\mathbb E_{\bm s}[|Z_J|^2]\leq\sum_{J\in\{a,b\}^n}\langle J|\hat O^\dagger\hat Q\hat O|J\rangle=\Tr\!\left(\hat Q\hat O\hat P_{\mathcal C}\hat O^\dagger\right).
\end{align}
Including strings with $c_J=0$ can only increase the upper bound because $\hat Q\geq0$.
\end{proof}

The remaining task is to bound this trace independently of $n$. Bounding each $\|V_J\|^2$ separately by one would give only $2^n$, which would not establish the desired sample complexity. We instead use the Gaussian structure of $\hat Q$ and $\hat O$ to bound their total contribution by one.

\subsection{Complete algorithm and estimation guarantee}

To complete the second-moment estimate in Prop.~\ref{prop:importance}, we need the following bound.

\begin{lemma}[Uniform trace bound]\label{lem:filtercontraction-simple}
For $\hat Q$ and $\hat P_{\mathcal C}$ defined in Eq.~\eqref{eq:Qdef}, every Gaussian unitary $\hat G$ on the $4n$ modes satisfies
\begin{align}\label{eq:filtercontraction-simple}
    \Tr\!\left(\hat Q\hat G \hat P_{\mathcal C}\hat G^\dagger\right) &\leq1.
\end{align}
\end{lemma}

We prove this lemma in Sec.~\ref{app:approx-compression-proof} and use it here to complete the estimation guarantee. The following proof collects the sampling steps constructed above and establishes their accuracy and cost.

\begin{proof}[Proof of Thm.~\ref{thm:Approximate mean value}]
Generate one independent sample as follows:
\begin{enumerate}
    \item \textbf{Sample an occupation string.}\label{Step 1} Draw $J\sim c_J^2$ by sampling its entries independently, as in Eq.~\eqref{eq:canonical-expansion}.
    \item \textbf{Compute the unnormalized Gaussian state.}\label{Step 2} Form a phase-sensitive Gaussian description of $\hat O|J\rangle$, and compute the norm and covariance matrix of $|V_J\rangle=\hat F\hat O|J\rangle$. If this vector is zero, return $Y=0$ and finish this sample.
    \item \textbf{Choose the phases.}\label{Step 3} From the covariance matrix of $|V_J\rangle/\|V_J\|$, find one phase for each block satisfying Eq.~\eqref{eq:balance-simple}.
    \item \textbf{Sample and evaluate a Gaussian overlap.}\label{Step 4} Draw $\bm s$ uniformly from $\{\pm1\}^n$ and return $Y=2^n\langle G_{\bm s}(\bm\varphi(J))|V_J\rangle/c_J$.
\end{enumerate}
Proposition~\ref{prop:importance} and Lemma~\ref{lem:filtercontraction-simple}, with $\hat G=\hat O$, give
\begin{align}
    \mathbb E[Y] &=A, \qquad \mathbb E[|Y|^2]\leq\Tr\!\left(\hat Q\hat O\hat P_{\mathcal C}\hat O^\dagger\right)\leq1.
\end{align}
Thus $\mathbb E[|Y-A|^2]\leq1$. It remains to combine independent samples and account for their cost.

\emph{Accuracy.} Take $m=\lceil8\epsilon^{-2}\rceil$ and let $L$ be the smallest odd integer at least $8\log(2/\delta)$. Generate $L$ groups of $m$ independent samples, with group means $\overline Y_g=m^{-1}\sum_{t=1}^mY_{g,t}$, and return their coordinatewise median:
\begin{align}\label{eq:app-amplitude-median}
    \widetilde A &=\operatorname{median}_{g\in[L]}\operatorname{Re}\overline Y_g+i\operatorname{median}_{g\in[L]}\operatorname{Im}\overline Y_g.
\end{align}
Independence implies $\mathbb E[|\overline Y_g-A|^2]\leq1/m$. For either coordinate $h=\operatorname{Re}$ or $h=\operatorname{Im}$, Chebyshev's inequality implies
\begin{align}
    \Pr\!\left[|h(\overline Y_g)-h(A)|>\epsilon/\sqrt2\right] &\leq\frac{2}{m\epsilon^2}\leq\frac14.
\end{align}
The median can fail this coordinate error bound only if at least half the groups fail it. Hoeffding's inequality bounds that probability by $e^{-L/8}\leq\delta/2$. A union bound over the two coordinates gives $|\widetilde A-A|\leq\epsilon$ with probability at least $1-\delta$.

\emph{Cost.} First compile the Gaussian circuit, including the local changes of basis, into a phase-sensitive $KAK$ representation. This one-time preprocessing is polynomial in the circuit size and $n$; retaining the scalar phase is necessary because the Majorana transformation alone does not determine a complex amplitude. The update rules of Ref.~\cite[Sec.~4.2]{reardon2024improved} then apply the compiled unitary to any Gaussian state in $O(n^3)$ arithmetic operations.

For each sample, form a phase-sensitive description of $\hat O|J\rangle$ in $O(n^3)$ operations. To choose the phases, only the covariance matrix of the normalized Gaussian vector proportional to $\hat{F}\hat{O}\ket{J}$ is needed. Apply the $4n$ single-mode diagonal factors of $\hat F$ to this covariance matrix, tracking the norm. Each update costs $O(n^2)$: Wick's theorem expresses each updated two-point function using a constant number of entries of the preceding covariance matrix. This includes occupation projections when a coefficient vanishes; a zero norm terminates the sample with $Y=0$. Thus all updates together cost $O(n^3)$~\cite{Bravyi2005}. The blockwise phase choice and the draws of $J$ and $\bm s$ cost $O(n)$.

For the final complex overlap, use the Hermiticity of $\hat F$ to write
\begin{align}\label{eq:app-filter-overlap-cost}
    \langle G_{\bm s}|V_J\rangle &= \langle \hat F G_{\bm s}|\hat O|J\rangle.
\end{align}
The Gaussian vector $\hat F|G_{\bm s}\rangle$ is a product of four-mode blocks, so its norm and phase are obtained by constant-size calculations within each block. Its complex overlap with $\hat O|J\rangle$ costs $O(n^3)$ using phase-sensitive Gaussian overlap routines~\cite{DiasKonig2024,reardon2024improved}. In particular, no phase-sensitive update is required for each factor of $\hat{F}$. Each sample therefore costs $O(n^3)$, and the $mL=O(\epsilon^{-2}\log(2/\delta))$ samples, including their means and medians, cost $O(n^3\epsilon^{-2}\log(2/\delta))$ after preprocessing.
\end{proof}

For a normalized pure Gaussian input $|\zeta\rangle$, the same phase choice gives an unbiased estimator $Z=2^n\langle G_{\bm s}|\hat F|\zeta\rangle$ of $\langle\Phi|\zeta\rangle$, without the occupation-string draw. Lemma~\ref{lem:branch-bound} and the operator norm of the non-unitary Gaussian operators give
\begin{align}
    \mathbb E[|Z|^2]&\leq\|\hat F|\zeta\rangle\|^2\leq\|\hat F\|_{\rm op}^2=\prod_{j=1}^n\max\{u_j^2,v_j^2\}.
    \label{eq:app-gaussian-block-bound}
\end{align}
Conjugating the target amplitude gives the Gaussian--block-product case stated in the main text. The Gaussian-input case needs no occupation-string draw, so the norm of the non-unitary operation alone suffices for this bound.

\subsection{Proof of the uniform trace bound}\label{app:approx-compression-proof}

We now prove Lem.~\ref{lem:filtercontraction-simple}. The main step is a compression inequality of the form $\Tr(\hat P_E\hat X)\leq\|\hat X\|_4$ for positive Gaussian operators. Here $\hat P_E$ selects the vacuum and fully occupied state on each four-mode block:
\begin{align}
    \hat P_{E,1} &:=|\text{vac}\rangle\langle\text{vac}|+|\text{full}\rangle\langle\text{full}|, \qquad \hat P_E:=\hat P_{E,1}^{\otimes n}.
\end{align}
A local particle--hole transformation maps $\hat P_{\mathcal C}$ to $\hat P_E$, so this inequality will bound the required trace by $\|\hat Q\|_4$. Here the Schatten norm is
\begin{align}\label{eq:Schatten-norm}
    \|\hat X\|_p &:=\left[\Tr\!\left((\hat X^\dagger \hat X)^{p/2}\right)\right]^{1/p}, \qquad \|\hat X\|_4=(\Tr \hat X^4)^{1/4}\quad\text{for }\hat X\geq0.
\end{align}
All traces and norms are over the full Fock space, and an even operator preserves fermion parity. We first prove the one-block inequality, then show how to apply it successively in the presence of correlations between blocks.

\begin{lemma}\label{lem:scalarcompression-simple}
For every positive even Gaussian operator $\hat Z$ on four modes,
\begin{align}\label{eq:scalarcompression-simple}
    \Tr(\hat P_{E,1}\hat Z) &\leq\|\hat Z\|_4.
\end{align}
\end{lemma}

\begin{proof}
The claim is immediate for $\hat Z=0$. Otherwise, the Gaussian canonical form~\cite{Bravyi2005} gives a parity-preserving Gaussian unitary $\hat V$ and nonnegative numbers $\lambda_{k,0},\lambda_{k,1}$ such that
\begin{align}
    \hat Z &=\hat V^\dagger\left[\bigotimes_{k=1}^4\left(\lambda_{k,0}|0\rangle\langle0|+\lambda_{k,1}|1\rangle\langle1|\right)\right]\hat V.
\end{align}
For $x\in\{0,1\}^4$, write $d_x:=\prod_{k=1}^4\lambda_{k,x_k}$, $p_x:=|\langle x|\hat V|\text{vac}\rangle|^2$, and let $\bar x$ be the bitwise complement of $x$. Complementing all occupation numbers reverses the covariance matrix of a Fock state. Thus $\hat V|\text{vac}\rangle$ and $\hat V|\text{full}\rangle$ have opposite covariance matrices, as do $|x\rangle$ and $|\bar x\rangle$. For even-weight $x$, the pure-Gaussian overlap formula~\cite[Eq.~(23)]{DiasKonig2024} is unchanged when both covariance matrices are reversed, giving
\begin{align}\label{eq:complementary-amplitude}
    |\langle x|\hat V|\text{full}\rangle|^2 &=|\langle\bar x|\hat V|\text{vac}\rangle|^2=p_{\bar x}.
\end{align}
For odd-weight strings both sides vanish by parity preservation.

It follows that $\Tr(\hat P_{E,1}\hat Z)=\sum_xp_x(d_x+d_{\bar x})$. Each pair $d_x+d_{\bar x}$ is a sum of two products of four factors, so four-factor H\"older's inequality gives
\begin{align}
    d_x+d_{\bar x} &\leq\prod_{k=1}^4(\lambda_{k,0}^4+\lambda_{k,1}^4)^{1/4}=\|\hat Z\|_4.
\end{align}
This is where the Schatten fourth norm enters. Averaging with the probabilities $p_x$, which sum to one, proves the lemma.
\end{proof}

For several blocks, conditioning the first block on $|\text{vac}\rangle$ or $|\text{full}\rangle$ leaves two operators on the remaining modes. An inductive argument therefore needs a bound on the \emph{sum} of their fourth norms. The next lemma supplies exactly this bound, including when the chosen block is correlated with the remaining modes.

\begin{lemma}\label{lem:conditionalcompression-simple}
Let $\hat X$ be a positive even Gaussian operator on a four-mode subsystem $A$ and a fermionic environment $B$. Define
\begin{align}
    \hat X_0 &:=\langle\text{vac}|_A \hat X|\text{vac}\rangle_A, \qquad \hat X_4:=\langle\text{full}|_A \hat X|\text{full}\rangle_A.
\end{align}
Then
\begin{align}\label{eq:conditionalcompression-simple}
    \|\hat X_0\|_4+\|\hat X_4\|_4 &\leq\|\hat X\|_4.
\end{align}
\end{lemma}

\begin{proof}
We first express the sum of norms as a single trace and then use monotonicity under the partial trace to remove $B$. This reduces the problem to the one-block bound in Lem.~\ref{lem:scalarcompression-simple}.

Write $|0\rangle_A:=|\text{vac}\rangle_A$ and $|4\rangle_A:=|\text{full}\rangle_A$. To represent each norm as a trace, choose $\hat\sigma_s:=\hat X_s^4/\Tr\hat X_s^4$ when $\hat X_s\neq0$, so that $\Tr(\hat X_s\hat\sigma_s^{3/4})=\|\hat X_s\|_4$. When $\hat X_s=0$, any even normalized density operator $\hat\sigma_s$ gives the same identity. Combining the two choices into
\begin{align}
    \hat\Omega &:=\sum_{s\in\{0,4\}}|s\rangle\langle s|_A\otimes\hat\sigma_s
\end{align}
gives
\begin{align}\label{eq:Qidentity-simple}
    \Tr(\hat X\hat\Omega^{3/4}) &=\|\hat X_0\|_4+\|\hat X_4\|_4, \qquad \Tr_B\hat\Omega=\hat P_{E,1}.
\end{align}
The first identity uses the orthogonality of $|0\rangle_A$ and $|4\rangle_A$, and the second uses $\Tr\hat\sigma_s=1$. No Gaussianity assumption on $\hat\Omega$ is needed.

We now use the partial-trace form of Petz's monotonicity~\cite{Petz1985}, stated for arbitrary positive semidefinite operators in Ref.~\cite[Cor.~2.19(1), arXiv version]{Hiai2023}. At exponent $1/4$, it reads
\begin{align}\label{eq:app-petz-partial-trace}
    \Tr(\hat\rho^{1/4}\hat\sigma^{3/4}) &\leq\Tr\!\left[(\Tr_B\hat\rho)^{1/4}(\Tr_B\hat\sigma)^{3/4}\right].
\end{align}
Neither normalization nor invertibility is required. Apply this inequality with $\hat\rho=\hat X^4$ and $\hat\sigma=\hat\Omega$. Since $\Tr_B\hat\Omega=\hat P_{E,1}$ is a projector, Eq.~\eqref{eq:Qidentity-simple} gives
\begin{align}\label{eq:conditional-petz-bound}
    \|\hat X_0\|_4+\|\hat X_4\|_4 &=\Tr(\hat X\hat\Omega^{3/4})\leq\Tr\!\left[\hat P_{E,1}(\Tr_B\hat X^4)^{1/4}\right].
\end{align}
The Gaussian canonical form and Wick's rule~\cite{Bravyi2005} show, respectively, that positive powers and partial traces preserve Gaussianity. Hence $(\Tr_B\hat X^4)^{1/4}$ is a positive even Gaussian operator on the four modes of $A$, to which Lem.~\ref{lem:scalarcompression-simple} applies:
\begin{align}
    \Tr\!\left[\hat P_{E,1}(\Tr_B\hat X^4)^{1/4}\right] &\leq\|(\Tr_B\hat X^4)^{1/4}\|_4=(\Tr\hat X^4)^{1/4}=\|\hat X\|_4.
\end{align}
Together with Eq.~\eqref{eq:conditional-petz-bound}, this proves the claim.
\end{proof}

Fock conditioning preserves Gaussianity by the occupation-measurement update~\cite{Bravyi2005}; positivity is preserved by taking a matrix element. We can therefore use the conditional bound at each successive block, which gives the desired compression inequality.

\begin{corollary}\label{thm:compression-simple}
For every positive even Gaussian operator $\hat X$ on $4n$ modes,
\begin{align}\label{eq:compression-simple}
    \Tr(\hat P_E\hat X) &\leq\|\hat X\|_4.
\end{align}
\end{corollary}

\begin{proof}
We use induction on $n$. The case $n=1$ is Lem.~\ref{lem:scalarcompression-simple}. For the inductive step, take the first block as $A$ and the remaining blocks as $B$, and define $\hat X_0,\hat X_4$ as in Lem.~\ref{lem:conditionalcompression-simple}. Both conditional operators are positive and Gaussian, so the induction hypothesis applies to each of them. Consequently,
\begin{align}
    \Tr(\hat P_E\hat X) &=\sum_{s\in\{0,4\}}\Tr\!\left(\hat P_{E,1}^{\otimes(n-1)}\hat X_s\right)\leq\|\hat X_0\|_4+\|\hat X_4\|_4\leq\|\hat X\|_4,
\end{align}
where the last step is the conditional bound.
\end{proof}

It remains to return from $\hat P_E$ to $\hat P_{\mathcal C}$ and evaluate the norm of the non-unitary Gaussian operation. The resulting bound is uniform in both the number of blocks and the Gaussian unitary.

\begin{proof}[Proof of Lem.~\ref{lem:filtercontraction-simple}]
Write $\hat\gamma_{j,r}:=\hat\gamma_{8(j-1)+r}$ for the $r$th Majorana operator in block $j$, with $r\in[8]$. The Gaussian unitary $\hat C_j:=\exp[(\pi/2)\hat\gamma_{j,6}\hat\gamma_{j,8}]$ complements the last two occupation numbers. It maps $|a\rangle=|1100\rangle$ to $|\text{full}\rangle$ and $|b\rangle=|0011\rangle$ to $|\text{vac}\rangle$, up to phases. Thus $\hat C:=\bigotimes_j\hat C_j$ satisfies $\hat C \hat P_{\mathcal C}\hat C^\dagger=\hat P_E$.

Set $\hat X:=\hat C\hat G^\dagger \hat Q\hat G\hat C^\dagger$, which is positive, even, and Gaussian. The compression inequality and unitary invariance of the Schatten norm give
\begin{align}
    \Tr\!\left(\hat Q\hat G \hat P_{\mathcal C}\hat G^\dagger\right) &=\Tr(\hat P_E\hat X)\leq\|\hat X\|_4=\|\hat Q\|_4.
\end{align}
Each of the four one-mode factors of $\hat Q_j=\hat F_j^2$ has eigenvalues $\sqrt{u_j}$ and $\sqrt{v_j}$, by Eq.~\eqref{eq:app-bounded-filter}. Hence $\Tr \hat Q_j^4=(u_j^2+v_j^2)^4$. Multiplicativity of the Schatten norm and the normalization $u_j^2+v_j^2=1$ now imply
\begin{align}
    \|\hat Q\|_4 &=\prod_{j=1}^n\|\hat Q_j\|_4=\prod_{j=1}^n(u_j^2+v_j^2)=1.
\end{align}
These identities also hold when $v_j=0$, so the bound includes all allowed coefficients.
\end{proof}

\addtocontents{toc}{\protect\AppendixTOCHide}
\bibliography{reference.bib}

@article{oh2024quantum,
  title={Quantum-inspired classical algorithms for molecular vibronic spectra},
  author={Oh, Changhun and Lim, Youngrong and Wong, Yat and Fefferman, Bill and Jiang, Liang},
  journal={Nature Physics},
  volume={20},
  number={2},
  pages={225--231},
  year={2024},
  publisher={Nature Publishing Group UK London},
  doi = {10.1038/s41567-023-02308-9}
}

@article{bouland2019complexity,
  title={On the complexity and verification of quantum random circuit sampling},
  author={Bouland, Adam and Fefferman, Bill and Nirkhe, Chinmay and Vazirani, Umesh},
  journal={Nature Physics},
  volume={15},
  number={2},
  pages={159--163},
  year={2019},
  doi={10.1038/s41567-018-0318-2}
}

@inproceedings{clifford2018classical,
  title={The classical complexity of boson sampling},
  author={Clifford, Peter and Clifford, Rapha{\"e}l},
  booktitle={Proceedings of the Twenty-Ninth Annual ACM-SIAM Symposium on Discrete Algorithms},
  pages={146--155},
  year={2018},
  organization={SIAM},
  doi={10.1137/1.9781611975031.10},
  eprint={1706.01260},
  archivePrefix={arXiv}
}

@article{clifford2024faster,
  title={Faster classical boson sampling},
  author={Clifford, Peter and Clifford, Rapha{\"e}l},
  journal={Physica Scripta},
  volume={99},
  pages={065121},
  year={2024},
  doi={10.1088/1402-4896/ad4688},
  eprint={2005.04214},
  archivePrefix={arXiv}
}

@article{Brod2020classicalsimulation,
  author  = {Brod, Daniel J. and Oszmaniec, Micha{\l}},
  title   = {Classical simulation of linear optics subject to nonuniform losses},
  journal = {Quantum},
  volume  = {4},
  pages   = {267},
  year    = {2020},
  doi     = {10.22331/q-2020-05-25-267},
  archivePrefix = {arXiv},
  eprint  = {1906.06696},
  primaryClass = {quant-ph}
}

@article{Oszmaniec2014,
  title = {Classical simulation of fermionic linear optics augmented with noisy ancillas},
  author = {Oszmaniec, Micha{\l} and Gutt, Jan and Ku{\'s}, Marek},
  journal = {Physical Review A},
  volume = {90},
  number = {2},
  pages = {020302},
  numpages = {5},
  year = {2014},
  month = {Aug},
  publisher = {American Physical Society},
  doi = {10.1103/PhysRevA.90.020302},
  url = {https://link.aps.org/doi/10.1103/PhysRevA.90.020302}
}

@misc{knill2001fermioniclinearopticsmatchgates,
      title={Fermionic Linear Optics and Matchgates}, 
      author={E. Knill},
      year={2001},
      eprint={quant-ph/0108033},
      archivePrefix={arXiv},
      primaryClass={quant-ph},
      url={https://arxiv.org/abs/quant-ph/0108033}, 
}

@article{JordanWigner,
  author  = {Jordan, P. and Wigner, E.},
  title   = {{{\"U}ber das Paulische {\"A}quivalenzverbot}},
  journal = {Zeitschrift f{\"u}r Physik},
  volume  = {47},
  pages   = {631--651},
  year    = {1928},
  doi     = {10.1007/BF01331938},
  number = {9--10}
}

@article{Parity,
  author  = {Szalay, Szil{\'a}rd and Zimbor{\'a}s, Zolt{\'a}n and M{\'a}t{\'e}, Mih{\'a}ly and Barcza, Gergely and Schilling, Christian and Legeza, {\"O}rs},
  title   = {Fermionic systems for quantum information people},
  journal = {Journal of Physics A: Mathematical and Theoretical},
  volume  = {54},
  pages   = {393001},
  year    = {2021},
  doi     = {10.1088/1751-8121/ac0646},
  number = {39}
}

@article{Phasecraft,
  title={Fermionic dynamics on a trapped-ion quantum computer beyond exact classical simulation},
  author={Alam, Faisal and Bosse, Jan Lukas and {\v{C}}epait{\.e}, Ieva and others},
  journal={arXiv preprint arXiv:2510.26300},
  year={2025},
  url = {https://arxiv.org/abs/2510.26300}
}

@article{DiasKonig2024,
  author  = {Dias, Beatriz and Koenig, Robert},
  title   = {Classical simulation of non-{Gaussian} fermionic circuits},
  journal = {Quantum},
  volume  = {8},
  pages   = {1350},
  year    = {2024},
  doi     = {10.22331/q-2024-05-21-1350},
  eprint  = {2307.12912},
  archivePrefix = {arXiv},
  primaryClass  = {quant-ph}
}

@article{Bravyi2005,
  author  = {Bravyi, Sergey},
  title   = {Lagrangian representation for fermionic linear optics},
  journal = {Quantum Information and Computation},
  volume  = {5},
  number  = {3},
  pages   = {216--238},
  year    = {2005},
  eprint  = {quant-ph/0404180},
  archivePrefix = {arXiv},
  doi = {10.26421/QIC5.3-3}
}

@article{Petz1985,
  author  = {Petz, D\'enes},
  title   = {Quasi-entropies for states of a von {Neumann} algebra},
  journal = {Publications of the Research Institute for Mathematical Sciences},
  volume  = {21},
  number  = {4},
  pages   = {787--800},
  year    = {1985},
  doi     = {10.2977/PRIMS/1195178929}
}

@article{bravyi2022simulate,
  title={How to simulate quantum measurement without computing marginals},
  author={Bravyi, Sergey and Gosset, David and Liu, Yinchen},
  journal={Physical Review Letters},
  volume={128},
  number={22},
  pages={220503},
  year={2022},
  publisher={APS},
  doi = {10.1103/PhysRevLett.128.220503}
}

@article{oh2026classical,
  title={Classical simulation of free-fermionic dynamics and quantum chemistry with magic input},
  author={Oh, Changhun and Oszmaniec, Micha{\l} and Reardon-Smith, Oliver and Zimbor{\'a}s, Zolt{\'a}n},
  journal={arXiv preprint arXiv:2604.26813},
  year={2026},
  url = {https://arxiv.org/abs/2604.26813}
}

@article{oszmaniec2022fermion,
  title={Fermion sampling: A robust quantum computational advantage scheme using fermionic linear optics and magic input states},
  author={Oszmaniec, Micha{\l} and Dangniam, Ninnat and Morales, Mauro E. S. and Zimbor{\'a}s, Zolt{\'a}n},
  journal={PRX Quantum},
  volume={3},
  number={2},
  pages={020328},
  year={2022},
  publisher={APS},
  doi = {10.1103/PRXQuantum.3.020328}
}

@article{AvezBender2012,
 author={Avez, B. and Bender, M.},
 title={Evaluation of overlaps between arbitrary fermionic quasiparticle vacua},
 journal={Phys. Rev. C},
 volume={85},
 pages={034325},
 year={2012},
 eprint={1109.2078},
 archivePrefix={arXiv},
  number = {3},
  doi = {10.1103/PhysRevC.85.034325}
}

@article{Wimmer2011,
 author={Wimmer, M.},
 title={Algorithm 923: Efficient numerical computation of the {Pfaffian} for dense and banded skew-symmetric matrices},
 year={2012},
 eprint={1102.3440},
 archivePrefix={arXiv},
  journal = {ACM Transactions on Mathematical Software},
  volume = {38},
  number = {4},
  pages = {30},
  numpages = {17},
  doi = {10.1145/2331130.2331138}
}

@article{mcardle2020quantum,
  title={Quantum computational chemistry},
  author={McArdle, Sam and Endo, Suguru and Aspuru-Guzik, Al{\'a}n and Benjamin, Simon C and Yuan, Xiao},
  journal={Reviews of Modern Physics},
  volume={92},
  number={1},
  pages={015003},
  year={2020},
  publisher={APS},
  doi = {10.1103/RevModPhys.92.015003}
}

@article{valiant2002quantum,
  title={Quantum circuits that can be simulated classically in polynomial time},
  author={Valiant, Leslie G},
  journal={SIAM Journal on Computing},
  volume={31},
  number={4},
  pages={1229--1254},
  year={2002},
  publisher={SIAM},
  doi = {10.1137/S0097539700377025}
}

@article{hebenstreit2020computational,
  title={Computational power of matchgates with supplementary resources},
  author={Hebenstreit, Martin and Jozsa, Richard and Kraus, Barbara and Strelchuk, Sergii},
  journal={Physical Review A},
  volume={102},
  number={5},
  pages={052604},
  year={2020},
  publisher={APS},
  doi = {10.1103/PhysRevA.102.052604}
}

@article{DaviesLewis1970,
  author = {Davies, E. B. and Lewis, J. T.},
  title = {An operational approach to quantum probability},
  journal = {Communications in Mathematical Physics},
  volume = {17},
  pages = {239--260},
  year = {1970},
  doi = {10.1007/BF01647093},
  number = {3}
}

@book{Watrous2018,
  author = {Watrous, John},
  title = {The Theory of Quantum Information},
  publisher = {Cambridge University Press},
  year = {2018},
  doi = {10.1017/9781316848142}
}

@article{bako2025fermionic,
  title={Fermionic {Born} machines: Classical training of quantum generative models based on fermion sampling},
  author={Bak{\'o}, Bence and Kolarovszki, Zolt{\'a}n and Zimbor{\'a}s, Zolt{\'a}n},
  journal={arXiv preprint arXiv:2511.13844},
  year={2025},
  url = {https://arxiv.org/abs/2511.13844}
}

@article{reardon2024improved,
  title={Improved simulation of quantum circuits dominated by free fermionic operations},
  author={Reardon-Smith, Oliver and Oszmaniec, Micha{\l} and Korzekwa, Kamil},
  journal={Quantum},
  volume={8},
  pages={1549},
  year={2024},
  publisher={Verein zur F{\"o}rderung des Open Access Publizierens in den Quantenwissenschaften},
  doi = {10.22331/q-2024-12-04-1549}
}

@article{jozsa2008matchgates,
  title={Matchgates and classical simulation of quantum circuits},
  author={Jozsa, Richard and Miyake, Akimasa},
  journal={Proceedings of the Royal Society A: Mathematical, Physical and Engineering Sciences},
  pages={3089--3106},
  year={2008},
  publisher={The Royal Society},
  volume = {464},
  number = {2100},
  doi = {10.1098/rspa.2008.0189}
}

@article{cudby2023gaussian,
  title={Gaussian decomposition of magic states for matchgate computations},
  author={Cudby, Joshua and Strelchuk, Sergii},
  journal={arXiv preprint arXiv:2307.12654},
  year={2023},
  url = {https://arxiv.org/abs/2307.12654}
}

@article{dias2026optimal,
  title={Optimal and improved gate decompositions for accelerated classical simulation of near-{Gaussian} fermionic circuits},
  author={Dias, Beatriz and Bosse, Jan Lukas and Seddon, James R},
  journal={arXiv preprint arXiv:2603.18869},
  year={2026},
  url = {https://arxiv.org/abs/2603.18869}
}

@article{kerenidis2026scalable,
  title={Scalable Quantum Machine Learning: Trainability, Expressivity and Efficiency},
  author={Kerenidis, Iordanis},
  journal={arXiv preprint arXiv:2607.24014},
  year={2026},
  url = {https://arxiv.org/abs/2607.24014}
}

@article{Hiai2023,
  title={Equality cases in monotonicity of quasi-entropies, {Lieb}'s concavity and {Ando}'s convexity},
  author={Hiai, Fumio},
  journal={Journal of Mathematical Physics},
  volume={64},
  number={10},
  pages={102201},
  year={2023},
  publisher={AIP Publishing},
  doi={10.1063/5.0154271},
  eprint={2304.04361},
  archivePrefix={arXiv},
  primaryClass={quant-ph}
}

@article{aaronson2013computational,
  title={The computational complexity of linear optics},
  author={Aaronson, Scott and Arkhipov, Alex},
  journal={Theory of Computing},
  volume={9},
  number={4},
  pages={143--252},
  year={2013},
  doi={10.4086/toc.2013.v009a004}
}

\end{document}